\documentclass[american]{article}
\usepackage[T1]{fontenc}
\usepackage{geometry}
\usepackage{array}
\usepackage{textcomp}
\usepackage{graphicx}
\usepackage{xr}
\usepackage{amsfonts}
\usepackage{amsbsy}
\usepackage{amstext}
\usepackage{amsthm}
\usepackage{amssymb}
\usepackage{amsmath}
\usepackage{libertine}
\usepackage{bbold}
\usepackage{soul}
\makeatletter
\renewcommand{\@biblabel}[1]{\quad#1.}
\makeatother

\makeatletter

\usepackage{tikz}
\usetikzlibrary{positioning,calc,arrows}
\usepackage{dsfont}
\newtheorem{algo}{Algorithm}

\usepackage{multirow}
\usepackage{subcaption}
\definecolor{green}{RGB}{0,150,0}

\makeatother

\usepackage{babel}
  \addto\captionsamerican{\renewcommand{\corollaryname}{Corollary}}
  \addto\captionsamerican{\renewcommand{\lemmaname}{Lemma}}
  \addto\captionsamerican{\renewcommand{\propositionname}{Proposition}}
  \addto\captionsamerican{\renewcommand{\theoremname}{Theorem}}
  \providecommand{\corollaryname}{Corollary}
  \providecommand{\lemmaname}{Lemma}
  \providecommand{\propositionname}{Proposition}
\providecommand{\theoremname}{Theorem}
\usepackage{float}
\newcommand{\insertfig}[2][1.]{\begin{center}\includegraphics[width=#1\textwidth]{#2}\end{center}}

\newcommand{\citep}[1]{(\cite{#1})}
\title{Models of protein production along the cell cycle: an investigation of possible sources of noise}

\author{ Renaud Dessalles\footnote{Dept. of Biomathematics, UCLA, Los Angeles, CA 90095-1766}, Vincent Fromion*\footnote{MaIAGE, INRA, Universit\' e Paris-Saclay, 78350 Jouy-en-Josas, France, * Corresponding author: vincent.fromion@inra.fr},
Philippe Robert\footnote{INRIA de  Paris, 2 rue Simone Iff,  75589 Paris  Cedex 12, France}}
\date{}

\begin{document}

%%%%%%%%%%%%%%%%%%%%%%%%%%%%%% Textclass specific LaTeX commands.
\numberwithin{equation}{section}
%\numberwithin{figure}{section}
%\theoremstyle{plain}
\newtheorem{thm}{\protect\theoremname}
  \newtheorem{prop}[thm]{\protect\propositionname}
  \newtheorem{cor}[thm]{\protect\corollaryname}
  \newtheorem{lem}[thm]{\protect\lemmaname}

%%%%%%%%%%%%%%%%%%%%%%%%%%%%%%%%%%%%%%%%%%%%%%%%%%%%%%%%%
\newcommand{\diff}{\mathop{}\mathopen{}\mathrm{d}}
\renewcommand\labelitemi{---}

\global\long\def\P#1{\mathbb{P}\left[#1\right]}

\global\long\def\E#1{\left\langle #1\right\rangle }

\global\long\def\R{\mathbb{R}}

\global\long\def\N{\mathbb{N}}

\global\long\def\V#1{\text{{\rm Var}}\left[#1\right]}

\global\long\def\Cov#1{\text{{\rm Cov}}\left[#1\right]}

\global\long\def\Ehat#1{\overline{\left\langle #1\right\rangle }}

\global\long\def\Ehat#1{\overline{\raisebox{0pt}[1.1\height]{$\left\langle #1 \right\rangle $}}}

\global\long\def\Vhat#1{\overline{\raisebox{0pt}[1.2\height]{\text{{\rm Var}}}}\left[#1\right]}

\newcommand\ind[1]{ \mathbb{1}_{\left\{#1\right\}}}

\global\long\def\cal#1{\mathcal{#1}}

\global\long\def\Vhatt#1#2{\overline{\raisebox{0pt}[1.2\height]{\text{{\rm Var}}}_{#2}}\left[#1\right]}

\maketitle

\vspace*{0.2in}

% Insert additional author notes using the symbols described below. Insert symbol callouts after author names as necessary.
%
% Remove or comment out the author notes below if they aren't used.
%

%$\ddagger$ Performed the theoretical analysis and the simulations and wrote the manuscript.
%\Yinyang Conceived and coordinated the project.

\section*{Abstract}
In this article, we quantitatively study, through stochastic models, the effects of several intracellular phenomena, such as cell volume growth, cell division, gene replication as well as fluctuations of available RNA polymerases and ribosomes. These phenomena are indeed rarely considered in classic models of protein production and no relative quantitative comparison among them has been performed. The parameters for a large and representative class of proteins are determined using experimental measures. The main important and surprising conclusion of our study is to show that despite the significant fluctuations of free RNA polymerases and free ribosomes, they bring little variability to protein production contrary to what has been previously proposed in the literature. After verifying the robustness of this quite counter-intuitive result, we  discuss its possible origin from a theoretical view, and interpret it as the result of a mean-field effect.

%When considering classical stochastic models of protein production, several intracellular phenomena need still to be considered, such as cell volume growth, cell division, or gene replication;  furthermore, concentrations of RNA polymerases and ribosomes are in general assumed to be constant in these classical models. In this paper, we study the effects of all these different cellular mechanisms on protein variance. Starting from a classical two-stage protein production model, we successively integrate the volume growth, the random partition of compounds at division, the gene replication, and the sharing of RNA polymerases and ribosomes for the production of all types of proteins. Experimental measures are then used to determine the numerical values of the parameters of our models for a large class of several proteins. The main important and surprising conclusion of our study is to show that despite the significant fluctuations of free RNA polymerases and free ribosomes, they bring little variability to protein production as it was previously proposed in the literature. After verifying the robustness of this quite counter-intuitive result, we  discuss its possible origin from a theoretical view, and interpret it as the result of a mean field effect.

\section*{Introduction}

For some time now, fluorescent microscopy methods have provided quantitative measurements of gene expression on the level of individual cells, see for instance \cite{elowitz_stochastic_2002,ozbudak_regulation_2002}.  Measurements have shown that protein production is a highly variable process, even for genetically identical cells in constant environmental conditions.  The fluctuations can negatively affect genetic expression and impact the behavior of the  cell (see \cite{losick_stochasticity_2008}), or, on the contrary, beneficially participate in strategies to adapt to a changing environment \citep{balaban_bacterial_2004,acar_stochastic_2008}.

More recently, \cite{taniguchi_quantifying_2010} performed an extensive quantification of the variability of the gene expression of around a thousand different genes in \emph{E. coli}. One experiment per gene was performed in their study: for each cell of a population, the total fluorescence of the protein associated with that gene was measured; this quantity was then normalized by cell volume and the fluorescence of a single protein (see the Supplementary Material of \cite{taniguchi_quantifying_2010}), resulting in the protein concentration in each cell of the population. Furthermore, for a significant number of the genes considered, messenger RNA production is also quantified in each cell: using an mRNA-sequencing technique (RNA-seq), they were able to estimate the average production of 841 different types of mRNA, and using fluorescence \emph{in situ} hybridization, they could even measure some types of mRNA (137 types) with single molecule precision in each cell of a population. Statistics over the population then gave the average concentration of each protein and RNA messenger, as well as the coefficient of variation (CV) for both quantities.

By comparing the behavior of the CV with that associated with the classic two-stage model of gene expression proposed in \cite{rigney_stochastic_1977,berg_model_1978} and reviewed in \cite{paulsson_models_2005} (see below), Taniguchi et al. show that the behavior of the messengers' variability with respect to their abundances resembles the one expected on the basis of the two-stage model.  For protein variability, they identify two regimes of protein variability depending on the average protein concentration. For infrequently expressed proteins, the protein CV is shown to be inversely proportional to the average concentration, in accordance with the one expected with the two-stage model (as it is shown in \cite{taniguchi_quantifying_2010}). By contrast, for highly produced proteins, the CV becomes independent of the average protein concentration and the quantified variability is then significantly larger than the one expected on the basis of the two-stage model. This effect does not seem specific to the type of protein, as all the highly expressed proteins are similarly impacted; therefore, a possible gene-specific phenomenon (such as the variability induced by the regulation of the gene) seems unlikely to explain the additional noise observed. After having ruled out several possibilities, the authors then proposed in their study cell-scale phenomena as possible explanations for this shift, in particular the fluctuations in the availability of RNA-polymerases and ribosomes in the production of different proteins. But other cell-scale mechanisms can also contribute significantly to protein variability, for example the partitioning that occurs at division, in which each compound (mRNA or protein) goes to either of the two daughter cells \citep{huh_non-genetic_2011,huh_random_2011,swain_intrinsic_2002}, or the gene replication event in which the transcription rate is doubled at some point in the cell cycle \citep{swain_intrinsic_2002}. Finally, some assumptions considered in the two-stage model are questionable when applied in the context of Taniguchi~et~al., for example the fact that a death process acting on the proteins models the dilution effect due to cell volume growth (see \cite{fromion_stochastic_2013}).

The study of Taniguchi~et~al. gives an extensive set of measurements for the majority of \emph{E. coli} genes. This turns out to be very useful to link theoretical models with experimental data, not only to determine the parameters of the models, but also to compare the predictions of these models to experimental results. This was developed by Taniguchi~et~al. for the two-stage model, in which they concluded  this classic model misses some cell-scale mechanisms to fully reflect experimental protein variability observed. Yet they did not try to develop models that include cell-scale mechanisms to definitely corroborate their hypothesis.

To tackle this limitation, we propose in this paper a model that integrates several cellular mechanisms which are not present in the two-stage model, and which are usually regarded as possible contributors to the variability of proteins. To set the model parameters, we used the measurements of 841 genes provided by Taniguchi~et al.; those for which both the average mRNA and protein concentration have been measured. The predictions of analytical formulas or simulations using these parameters can then reveal and quantify the contribution of each of the cell mechanisms considered to the variability in gene expression.

In the next subsection, we present a review of the classic two-stage model as it is broadly used in the literature. We will explicitly present its limitations in regard to cell-scale mechanisms that are suspected to play a definitive role in protein variability. We then present aspects that need to be changed in order to represent these cell-scale mechanisms and that make it possible  a quantitative comparison on their effect on protein expression. The Results section presents the predictions and results in three steps: we successively study 1) the effect on protein variance of random partition at division, 2) gene replication, and 3) fluctuations in the availability of RNA-polymerases and ribosomes. In the Discussion section, we discuss the comparison with the variability predicted with our models to those experimentally observed. We will show that even if some mechanisms significantly add variability to the system, it is not enough to explain the variance of highly expressed proteins observed experimentally; contrary to what has been suggested by Taniguchi~et~al. We will also discuss some phenomena not modeled that could impact the protein variability. The Materials and Methods section exhaustively describes our complete model, derived from the two-stage model, and the procedure used to analyze it.

\subsection*{\label{twostagemodel} Limitations of the Two-Stage Model}

	We based our approach on  one of the simplest classic stochastic models of protein production, the two-stage model \citep{rigney_stochastic_1977,berg_model_1978,paulsson_models_2005,shahrezaei_analytical_2008}.  This model describes the production and degradation of mRNAs and proteins of one particular type. In contrast to a three-stage model \citep{paulsson_models_2005}, it considers a constitutive gene, as it does not integrate a regulatory stage at the transcription initiation level, and thus represents the expression of a constitutive gene. The additional noise observed in Taniguchi~et~al. uniformly impacts highly expressed proteins, and therefore seems to be the result of cell-scale phenomena.  Even if gene-specific mechanisms such as gene regulation can have a large impact on protein variance (see \cite{schwabe2012,dessalles_stochastic_2017} for instance), we rather are looking at mechanisms at the scale of the cell to explain the feature observed in Taniguchi~et~al.

The two-stage model represents the evolution of the random variables $M$ and $P$, respectively representing the number of mRNAs and proteins associated with this gene inside a single cell (see Figure~\ref{fig:mod_paulsson_scheme}).  It is important to note that the quantities considered in this model are integers. In Figure~\ref{fig:mod_paulsson_scheme}, $M$ and $P$ respectively denote the number of mRNAs and proteins inside the cell, but their concentrations are not explicit in the model, it does not incorporate the notion of cell volume.

\begin{figure}%wrapfigure
	\begin{center}\insertfig[.6]{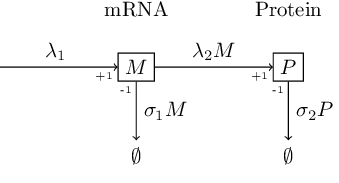}\end{center}
\caption{\label{fig:mod_paulsson_scheme}\textbf{Classic two-stage model of one constitutive gene.} $M$ and $P$ represent the respective numbers of mRNAs and proteins in a given cell. The parameter $\lambda_1$ is the molecular-specific rate at which mRNAs are produced, $\lambda_2M$ is the rate at which proteins are created, $\sigma_1$ indicates the molecular-specific mRNA degradation rate and $\sigma_2$ is the rate at which protein leaves the cell (by the effect of volume growth and division). }
\end{figure}

Both the transcription and translation mechanisms are represented in the model, as well as mRNA degradation. There is also a protein decay mechanism that can either represent the degradation of the protein or the dilution effect induced by the growth of cell volume, since the proteolysis often occurs on a timescale much longer than the cell cycle (see \cite{koch_protein_1955}). All the events, like the production and degradation of mRNAs and proteins, are assumed to occur at exponentially distributed random times, the rates of which depend on the current state of the system.

The simplicity of the two-stage model is such that the identification of closed-form formulas for the mean and variance of mRNAs and proteins is possible (see \cite{paulsson_models_2005}); moreover as it only represents mechanisms that are specific to protein production, the predicted variance can be naturally regarded as representative of the intrinsic noise, as has been done by Taniguchi~et~al. But the model lacks several features that are necessary for making quantitative comparisons with experimental results, and  neglects several mechanisms that might have an impact on protein variance.

\begin{description}

\item [No representation of the volume.] The fact that cell volume is not explicitly represented in the two-stage model (and those derived from it) does not seem to have been specifically highlighted in the literature. It appeared to us to be an important limitation of classic models (such as the two-stage model) for the three following intertwined reasons:
\begin{itemize}

\item Quantitative comparison with experimental concentrations. As previously stated, the two-stage model describes the evolution of the \emph{number} of messengers and proteins in a single cell. It appears that the measurements of Taniguchi~et~al. represent concentrations, as this quantity is often more relevant than the numbers as it determines the speed of each reaction  through the law of mass action. The direct quantitative comparison of means and variances obtained from the two-stage model to those obtained experimentally is not completely licit. In order to have a model for bacterial growth than can be compared with the experimental data, it is necessary to explicitly describe the volume, in order to represent concentrations in the cytosol, and not only numbers.

\item Impact of volume evolution on concentration. In real cells, the simple growth of cell volume has a direct mechanical impact on compound concentrations through dilution. This effect of dilution of proteins is approximated in the two-stage model by a degradation process at  exponentially distributed random times \citep{paulsson_models_2005}, whose rate $\sigma_2P$ depends on the protein number, the parameter $\sigma_2$ being linked to the doubling time. But in reality, dilution is a different process and should be represented as an explicit volume increase rather than a protein disappearance rate. Furthermore, the effect of volume growth on messenger concentration is completely neglected in the two-stage model (due to their quick degradation).

\item Impact of volume evolution on reaction rates. Volume growth also has another impact on protein production, as it tends to slow down the reaction rates by diminishing the reactant concentrations in the law of mass action. This will have important consequences, as we will explicitly represent RNA-polymerases and ribosomes (see below), and the volume growth will impact their respective concentrations and thus global protein production rate.

\end{itemize}

\item [No partition at division.] The two-stage model represents the expression of one gene in a single cell, so protein decay represents the tendency of proteins to disappear from the cell through the effect of division. In reality, division is a relatively sudden event that partitions both mRNAs and proteins in the two daughter cells.

\item [No gene replication.] The two-stage model does not consider genetic replication events, implicitly assuming that the gene promoter number is the same, thus keeping the rate of transcription $\lambda_{1}$ constant. In reality, the activity of the gene is linked to its promoter concentration, which tends to decrease with the cell cycle as the volume increases, until the promoter is replicated.  At this point, the rate of transcription is doubled, thus possibly inducing a transcription burst.

\item [Constant availability of RNA-polymerases and ribosomes.] RNA-polymerases and ribosomes are needed in protein production for transcription and translation respectively. As these resources are shared among all the productions of all the different proteins, their availability fluctuates over  time. Since the two-stage model is a gene-centered model (it represents the expression of only one gene), it is not able to represent the competitive interactions between the different production processes for available RNA-polymerases and ribosomes. In particular, the mRNA production rate $\lambda_{1}$ and protein production rate per mRNA $\lambda_{2}$ are constant, as if the concentrations of available RNA-polymerases and ribosomes remain constant throughout the cell cycle.
\end{description}

Several of these external  mechanisms have been experimentally and theoretically
tackled in the literature: random partitioning at division in
\cite{huh_random_2011,soltani_intercellular_2016,bertaux2018}, gene replication
and gene dosage in
\cite{soltani_intercellular_2016,barziv2016,narula2015,walker_generation_2016,kempe2014,padovanmerhar2015,thomas2019}
or fluctuations in the availability of ribosomes in
\cite{fromion_stochastic_2015}). Nonetheless, all these papers concentrate at
one mechanisms at a time and there is no global perspective with quantitative
comparison on the respective impact of these mechanisms. Questions like ``What
are the respective quantitative effects on protein fluctuations of random
partitioning, gene replication or the sharing of ribosomes and RNA polymerases
when comparing to the intrinsic noise due to the protein production mechanism
itself?'' still need to be handled. Theoretical models, associated with
biologically relevant parameters will allow us to answer this kind of
challenging questions.

During the writing of the current article, two studies \cite{thomas2018,lin2018} were published: they both have a similar global approach that tend to represent whole cell cycle with DNA replication and division. Yet our study model differently several key aspects of the protein productions. In Lin~and~Amir~\cite{lin2018}, the sequestration of RNA polymerases and ribosome during the elongation are not explicitly modelled; like many of the previous studied cited above, the rate of each reaction only depends on the total number of compounds rather their concentration thus do not take into account the dilution effect (for instance, all things equal otherwise, the activity of a single promoter will tend to decrease as the cell volume grows, simply by dilution). In Thomas et al.~\cite{thomas2018}, has a broader approach by gathering proteins into groups (Transporters, Enzymes, Housekeeping, Ribosomes, etc.) and the rate of productions are specific to these groups; if global cell mechanisms are more integrated in this modelling, our approach has a gene specific precision that can be used  to see the effect of different cell mechanisms on the whole diversity of different proteins. In both case, neither of these study has a direct quantitative comparison with experimental measures; as our study fit its parameters and compare directly its predictions to the experimental measures of Taniguchi~et~al.\cite{taniguchi_quantifying_2010}.

\section*{Results}\label{results}

In the Materials and Methods section, we describe in detail our approach that considers models that integrates the features that are missing in the two-stage model previously described. In order to determine the relative impact of each of these aspects, we have proceeded in successive steps of increasing complexity. Below is presented the results concerning three intermediate models that successively incorporate a specific feature; the impact of each of these features on protein variance will be studied one at a time. The three intermediate models are the  following:
\begin{itemize}
\item The first intermediate model only considers the growth of the cell and study the impact of the partitioning at division. In particular, each gene concentration, as well as the free RNA-polymerase and ribosome concentrations, are considered constant during the whole cell cycle. It focuses on the effect of the partition at division on protein variance.
\item The second intermediate model then study the effect of gene replication, while still considering cell volume growth and partitions at division but keeping free RNA-polymerase and ribosome concentrations as constant.
\item Then is considered the complete model, with fluctuating free RNA-polymerases and ribosomes, replication, volume growth and partitions at division.
\end{itemize}
The next three subsections will present each of these models one at a time. For each of the models, the parameters are fitted in order to correspond to each of the genes of Taniguchi~et~al. and the protein concentration variance using the procedure described in Materials and Methods.

\subsection*{\label{subsec:model1}Impact of Random Partitioning}

\begin{figure}%fullwidth
	{
	\phantomsubcaption\label{fig:mod1_principle}
	\phantomsubcaption\label{fig:mod1_scheme}
	\phantomsubcaption\label{fig:mod1_profile}
	\phantomsubcaption\label{fig:mod1_var_ratio}
	}
	\insertfig{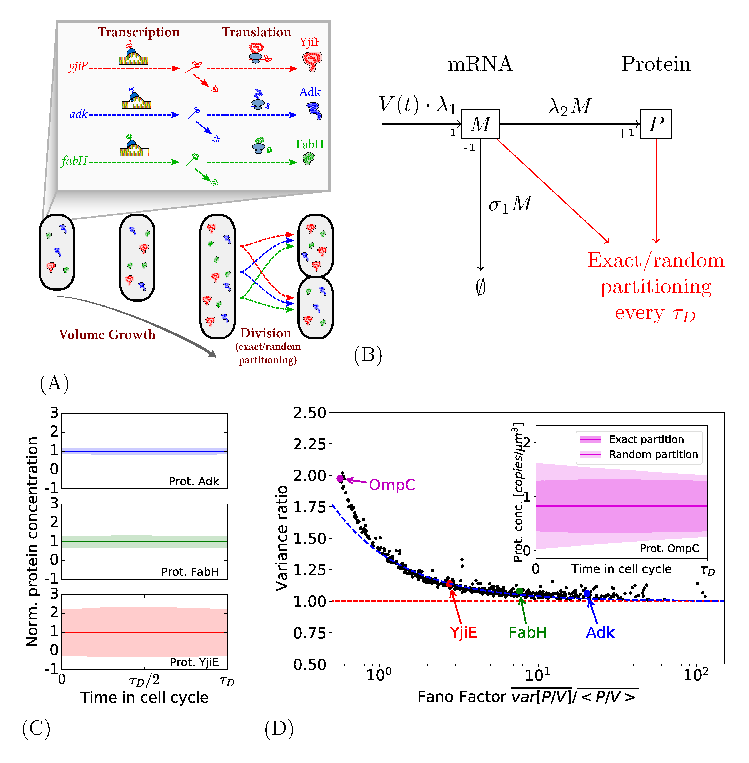}
\caption{\label{fig:mod1}\textbf{Intermediate model with volume growth and partition at division.} (A):  The model considers transcription,  translation, volume growth, and  partitioning at division. (B): Model of production for one type of protein. The parameters $\lambda_{1}$, $\sigma_{1}$ and $\lambda_{2}$ are specific to the type of protein. The gene concentration is considered as constant, giving a rate of transcription proportional to cell volume. (C): The evolution of the normalized protein concentration during the cell cycle for three proteins. The central thick line represents the average production during the cell cycle and the colored area  the standard deviation. (D) Main: For each type of protein, the variance in the case of random partition divided by the variance in the case of exact partition as a function of the Fano Factor (see main text and
%Section~\ref{subsec:S_mod1_simple_model_partition}
Section~1.5 of S1~Appendix). The proportion of variance attributed to the effect of the random partitioning globally follows the prediction of the simplified model in blue dashed line (see Section~1.5 of S1~Appendix).
 Section~%\ref{subsec:S_mod1_simple_model_partition} of S1~Appendix).
Inset: Comparison between the two partitions for the protein OmpC. The thick line is the average protein concentration during the cell cycle and the colored areas correspond to the standard deviation in the two cases. }%figsupp mod1
\end{figure}

The first intermediate model focuses on the effect of random partitioning. Fluctuations in the availability of RNA-polymerases and ribosomes, as well as the gene replication is therefore not considered. As the global sharing of ribosomes and RNA-polymerases are not represented, it results in a model where there are no interactions between the productions of different proteins; each production can be considered separately. It is a gene-centered model and one can focus on the production of one particular type of protein; $M$ and $P$ being respectively the number of mRNAs and proteins of this type.

Figure~\ref{fig:mod1} depicts this intermediate model. The parameter $\lambda_{1}$ of transcription (resp. $\lambda_{2}$ for the translation) implicitly includes some aspects specific to the gene (the promoter\textendash polymerase affinity, for instance) and the effective constant concentration of free RNA-polymerases (resp. free ribosomes). As the gene is in constant concentration, the rate of mRNA creation increases alongside the volume; it is therefore equal to $\lambda_{1}V(t)$, with $V(t)$ the volume of the cell at time $t$. Hence, the rate of mRNA production per volume unit remains constant.

For this intermediate model, a deterministic volume growth is considered.  Based on experimental results (see \cite{wang_robust_2010}), we represent the growth of the bacteria volume as exponential. For this model, if $t$ is the time spent since the last division, the volume is given by
\[
V(t)=V_{0}2^{t/\tau_{D}},
\]
with $V_{0}$ being the typical size of a cell at birth and $\tau_{D}$ the duration of the cell cycle. The explicit description of the volume leads to consider the concentration of mRNA $M(t)/V(t)$ and proteins $P(t)/V(t)$ at any time $t$ of the cell cycle.

Note that in this model and the subsequent ones, we interpret the division mechanism as a ``sizer'' model: the division occurs when the cell reaches the volume $2V_0$ \citep{tyson_sloppy_1986,wang_robust_2010,soifer_single-cell_2014,osella_concerted_2014}. Even if in the present model, the exact dependence of the volume to the cell cycle time can bring us to either interpret it as a ``timer'' (occurring at time $\tau_D$) or a ``sizer'' (occurring when reaching volume $2V_0$ ); the ``sizer'' interpretation will be preferred as it is the one that ensures cell size stability and has been shown to be a good first approximation to explain experimental cell size distributions \citep{robert_division_2014,osella_concerted_2014}. The last intermediate  model will later explicitly use a ``sizer'' mechanism.

In this first model, we study the effect of  partitioning on protein variance. Two mechanisms of segregating compounds at division are compared, either an exact or a random partitioning.
\begin{itemize}
\item for the exact partitioning,  the number of proteins and mRNAs  at division are equally allocated between the two daughter cells. Clearly, this mechanism does not have an impact on the variance of the mRNA and protein concentrations;
\item for the random partitioning at division, each mRNA and protein has an equal chance to be in either one of the two daughter cells (so that with probability $1/2$, they are in the next cell of interest). An additional variability,  due to this random allocation, should be therefore  added in this case.
\end{itemize}
As depicted in Materials and Methods, we first perform a theoretical analysis of this model in order to predict its mean concentrations of mRNA and protein averaged over the cell cycle: respectively $\Ehat{M/V}$ and $\Ehat{P/V}$ as they are defined in Equations~\eqref{eq:Ehat_P} of the Materials and Methods section. In this model, the mean concentration $\E{M(t)/V(t)}$ and $\E{P(t)/V(t)}$ \emph{during} the cell cycle remain constant. We then have that for any time $t$ of the cell cycle:
\begin{equation}
\E{M(t)/V(t)}=\Ehat{M/V}=\frac{\lambda_{1}\tau_{D}}{\sigma_{1}\tau_{D}+\log2}\qquad\text{and}\qquad\E{P(t)/V(t)}=\Ehat{P/V}=\frac{\lambda_{2}\tau_{D}}{\log2}\cdot\frac{\lambda_{1}\tau_{D}}{\sigma_{1}\tau_{D}+\log2}\label{eq:mod1_EM_EP}
\end{equation}
% \begin{multline}
% \E{M(t)/V(t)}=\Ehat{M/V}=\frac{\lambda_{1}\tau_{D}}{\sigma_{1}\tau_{D}+\log2}\\
% \qquad\text{and}\qquad\E{P(t)/V(t)}=\Ehat{P/V}=\frac{\lambda_{2}\tau_{D}}{\log2}\cdot\frac{\lambda_{1}\tau_{D}}{\sigma_{1}\tau_{D}+\log2}.\label{eq:mod1_EM_EP}
% \end{multline}
Proofs of these formulas can be found in
%Sections~\ref{subsec:S_mod1_mrna_dynamic} and \ref{subsec:S_mod1_prot_dynamic}
Sections~1.1.1 and 1.1.2 of S1~Appendix. The mRNA degradation times and the time of cell cycles are directly
measured in \cite{taniguchi_quantifying_2010}, thus setting the parameters
$\sigma_1$ and $\tau_D$. With that, these previous formulas make it possible, for every
gene considered in \cite{taniguchi_quantifying_2010}, to set the parameters
$\lambda_{1}$ and $\lambda_{2}$ in order to have an average production that
corresponds to those experimentally measured. This gives a series of parameters
corresponding to a representative sample of real bacterial genes (more details
in Section~1.2 %\ref{subsec:S_mod1_param}
of S1~Appendix).  As described in
Section~1.3
%\ref{subsec:S_mod1_Gillespie-Algorithms}
 of S1~Appendix, simulations are
performed using an algorithm derived from Gillespie method in order to determine
the evolution of protein concentration across the cell cycle and determine its
variance averaged over the cell cycle $\Vhat{P/V}$ -- as it is defined in
Equation~\eqref{eq:Vhat_P} of Materials and Methods section.  We then check that the
average protein and mRNA concentrations correspond to those experimentally
measured (see S1 Fig).

Figure~\ref{fig:mod1_profile} shows results of exact simulations (using Gillepsie-related algorithm): the figures show the profile of protein concentration during the cell cycle for three representative genes (\emph{adk}, \emph{fabH} and \emph{yjiE}) which are respectively highly, moderately and lowly expressed. As predicted theoretically, the mean concentration $\E{P(t)/V(t)}$ does not change across the cell cycle (it is due to the fact that the gene concentration remains constant).
%The figure shows the global tendency that the CV (defined as $\sqrt{\Vhat{P/V}}/\Ehat{P/V}$) decreases as the average production increases.

Figure~\ref{fig:mod1_var_ratio} shows the proportion of variance that is added with the introduction of random partitioning.  It appears that for all genes, their protein variance indeed increases with random partition, up to be doubled for some genes (like for the protein OmpC whose profile is shown in the inset, where the average production remains constant in both cases, but the random partition increases protein variance at the beginning of the cell cycle.).

The $x$-axis of Figure~\ref{fig:mod1_var_ratio} is somewhat unusual as it is the protein Fano Factor, defined as $\Vhat{P/V}/\Ehat{P/V}$.  It is used because the proportion of variance added by the random partition shows a remarkable clear dependence to the Fano Factor: proteins with a low Fano Factor are particularly more impacted.  Note that rare  proteins also tend to have low Fano factor (see S1 Fig (B)), %Figure~\ref{fig:S1B})
 the global tendency remains the same with having the average production as an x-axis, even if this dependence is less strong (see S1 Fig (C)).

In order to theoretically explain this clear dependence on the Fano factor, another simplified model, that focuses  only on the partition effect without considering volume growth, has been analyzed (details in Section~1.5
%\ref{subsec:S_mod1_simple_model_partition}
 of S1~Appendix). Its predictions are shown in blue dotted line of the main Figure~\ref{fig:mod1_var_ratio}. It globally predicts the proportion of noise that can be attributed to the random partitioning. It confirms that this effect is only significant for proteins with very low Fano factor.

We also want to decompose the protein concentration variance in the same way as it is done experimentally using the dual reporter technique \cite{elowitz_stochastic_2002} where are observed correlated and non-correlated variances between the expression of two proteins with identical promoters: the part of the variance that is not correlated between the two protein expressions is interpreted as solely due to the variability of the production mechanism itself; while the correlated variance is due to the common environment in the cell of the two genes (such as the cell volume, the number of available RNA polymerases, etc.) that influences equally the two productions.
Hilfinger~and~Paulson~\cite{hilfinger_separating_2011} showed that the theoretical counterpart of such decomposition is the environmental state decomposition of the variance (see details in Materials~and~Methods). With the environmental state decomposition, the protein variance
$\Vhat{P/V}$ can be decomposed into the two terms $\Vhatt{P/V}{int}$ and $\Vhatt{P/V}{ext}$ that would respectively represent the intrinsic (uncorrelated variance in the dual reporter technique) and extrinsic contribution (correlated variance in the dual reporter technique)
to protein variance. A general formula for this decomposition can be found in the Materials~and~Methods, but in the current model, as the only external environment
considered is the cell cycle (represented by the cell volume $V$), the docomposition of the variance of each protein $\Vhat{P/V}$ is the sum of
\begin{align}
\Vhatt{P/V}{int} & =  \frac{1}{\tau_{D}}\int_{0}^{\tau_{D}}\V{P(t)/V(t)}\diff t,\label{eq:var_1}\\
\Vhatt{P/V}{ext} & =  \frac{1}{\tau_{D}}\int_{0}^{\tau_{D}}\E{P(t)/V(t)}^{2}\diff t-\Ehat{P/V}^{2}\label{eq:var_2}
\end{align}
(see Section~1.4 %\ref{subsec:S_mod1_ESD}
of~S1~Appendix for more details).

% Analogously to the dual reporter technique, we can decompose theoretically, using
% the environmental state decomposition, the variance of each protein
% $\Vhat{P/V}$ into the two terms $\Vhatt{P/V}{int}$ and $\Vhatt{P/V}{ext}$
% that would respectively represent the intrinsic and extrinsic contribution
% to protein variance (see Materials and Methods). In this model, as the only external environment
% considered is the cell cycle, the variance of each protein $\Vhat{P/V}$
% can be decomposed into the sum of
% \begin{align}
% \Vhatt{P/V}{int} & =  \frac{1}{\tau_{D}}\int_{0}^{\tau_{D}}\V{P(t)/V(t)}\diff t,\label{eq:var_1}\\
% \Vhatt{P/V}{ext} & =  \frac{1}{\tau_{D}}\int_{0}^{\tau_{D}}\E{P(t)/V(t)}^{2}\diff t-\Ehat{P/V}^{2}.\label{eq:var_2}
% \end{align}
% (see Section~\ref{subsec:S_mod1_ESD} of~S1~Appendix).}
In this intermediate model,
even with random partitioning, such a decomposition shows surprisingly no
external contribution $\Vhatt{P/V}{ext}$ for every protein (since the protein
concentration $P(t)/V(t)$ remains constant across the cell cycle). It is
therefore remarkable that this decomposition only captures a part of what is
generally accepted as the extrinsic noise. We verified this decomposition by a
simulation that reproduces the dual reporter technique: we considered the
expression of two identical promoters, with the same parameters, in the same
cell. The covariance of the two protein concentrations measured in the
simulation is much smaller than the variance of each protein concentration (see
Section~1.4
%\ref{subsec:S_mod1_ESD}
in S1~Appendix and S1~Fig (D)). %ure~\ref{fig:S1D}).

\subsection*{\label{subsec:model2} Impact of Gene Replication}

\begin{figure}%fullwidth
	{
	\phantomsubcaption\label{fig:mod2_principle}
	\phantomsubcaption\label{fig:mod2_scheme}
	\phantomsubcaption\label{fig:mod2_profiles}
	\phantomsubcaption\label{fig:mod2_va_ratio_and_profile_AdK}
	}
	\insertfig{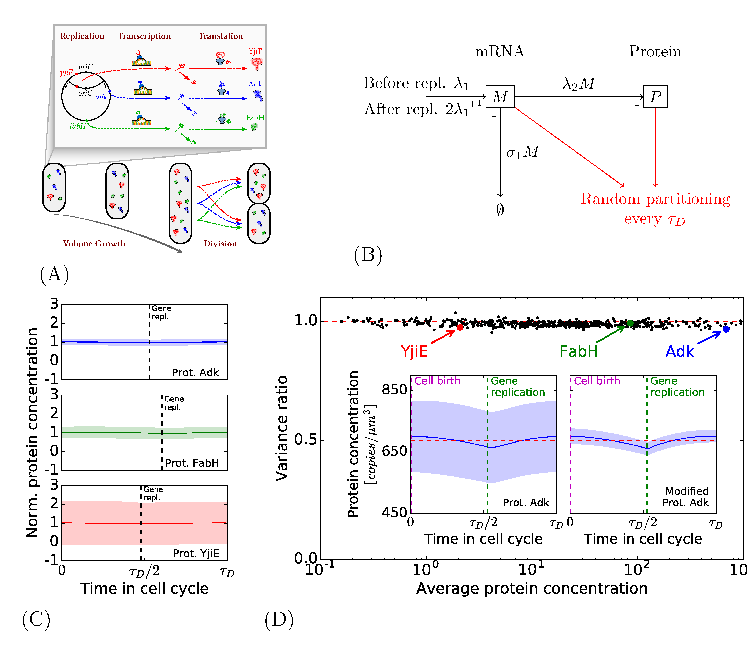}
\caption{\label{fig:mod2}\textbf{Intermediate model with cell cycle and gene replication.} (A): Features of the intermediate model. The model now takes into account replication. (B): The production of one particular type of protein. The number of mRNAs and proteins are respectively $M$ and $P$; the difference with the previous model is the introduction of replication at a time $\tau_{R}$ after the cell birth where the transcription rate is doubled. (C): The evolution of the normalized protein concentration during the cell cycle for three proteins. The central thick line and the colored areas represent the same quantities as in Figure~\ref{fig:mod1_profile}. (D) Main: for each type of protein, protein variance of the previous model (gene in constant concentration and random partitioning) divided by the one in this model. The replication paradoxically tends to slightly diminish the fluctuations of the concentration. Insight: on the left, the concentration through the cell cycle for the protein Adk (a close-up of the one presented in Figure~\ref{fig:mod2_profiles}); on the right, profile of a modified version of Adk with parameters chosen in order to minimize the variance (see main text).}%figsupp mod2
\end{figure}

Taking back the model previously described with volume growth and random
partitioning at division, we introduce the notion of gene replication. As in the
case of the slowing growing bacteria of \cite{taniguchi_quantifying_2010}, we
consider only one DNA replication per cell cycle. The gene is represented for
now on as an entity that is replicated at some instant $\tau_{R}$ of the cell
cycle, hence doubling the transcription rate ($\tau_R$ is determined depending
on the position of the gene on the chromosome, see Materials and Methods): before
a time $\tau_{R}$ after the cell birth, the rate of transcription will be
$\lambda_{1}$, after the time $\tau_{R}$, the rate of transcription is doubled
(see Figure~\ref{fig:mod2_scheme}).

The mean of mRNAs concentration can be determined at any moment of the cell cycle. At any time $t$ in the cell cycle (with $0\leq t<\tau_{D}$ ), it is given by,
\begin{equation}
\E{M(t)/V(t)}=\frac{\lambda_{1}}{\sigma_{1}V(t)}\left[1-\frac{e^{-(t+\tau_{D}-\tau_{R})\sigma_{1}}}{2-e^{-\tau_{D}\sigma_{1}}}+\ind{t\geq\tau_{R}}\left(1-e^{-(t-\tau_{R})\sigma_{1}}\right)\right],\label{eq:mod2-EMs}
\end{equation}
with $\ind{x\geq y}$ being the Heaviside function ($\ind{x\geq y}=1$ if $x\geq y$ and $0$ otherwise). One can refer to Section 2.1.1 %~\ref{subsec:S_mod2_mRNA-number}
 of S1~Appendix for the proof.

Similarly, the mean of  protein concentration can also be explicitly determined. Before the replication, if $t$ indicates the time after the birth of the cell (i.e. $0\leq t<\tau_{R}$ ), we can determine the mean of $P(t)/V(t)$ as a function of the mean of $M(0)$ and $P(0)$. Similarly, after the replication, the mean of $P(t+\tau_{R})/V(t+\tau_{R})$ (with $t$ such as $0\leq t<\tau_{D}-\tau_{R}$ ) is known as a function of the mean of $M(\tau_{R})$ and $P(\tau_{R})$.  They are given by the formula,
\begin{equation}
\E{P(t+\tau)/V(t+\tau)}=\frac{1}{V(t+\tau)}\left[\E{P(\tau)}+\frac{\eta\lambda_{1}}{\sigma_{1}}\lambda_{2}t+\left(\E{M(\tau)}-\frac{\eta\lambda_{1}}{\sigma_{1}}\right)\frac{1-e^{-\sigma_{1}t}}{\sigma_{1}}\right],\label{eq:mod2-EPs}
\end{equation}
with $\tau=0$ and $\eta=1$ for the case before the replication
and $\tau=\tau_{R}$ and $\eta=2$ for the case after the replication.
At steady state we are able to determine explicit values for $\E{P(0)}$
and $\E{P(\tau_{R}}$ as a function of the parameters (see Section~2.1.2 %\ref{subsec:S_mod2_Protein-number}
 of S1~Appendix).

As for the previous model, Equations~\eqref{eq:mod2-EMs} and~\eqref{eq:mod2-EPs} make it possible, for every gene experimentally measured in \cite{taniguchi_quantifying_2010}, to give a set of parameters $\lambda_{1}$, $\sigma_{1}$ and $\lambda_{2}$ that corresponds to it (Section~2.2 %\ref{subsec:S_mod2_param}
of S1~Appendix and S2 Fig (A)).%Figure~\ref{fig:S2A}).
 To determine the variance of protein concentration, simulations can be performed, but it is noticeable that we also managed to have  formulas for mRNA and protein variances. These formulas greatly simplify the analysis of this intermediate model (see Sections~2.1.1 and 2.1.2 %\ref{subsec:S_mod2_mRNA-number} and \ref{subsec:S_mod2_Protein-number}
 of S1~Appendix).

Figure~\ref{fig:mod2_profiles} presents the normalized profiles for three different proteins with different replication times in the cell cycle: their normalized average concentration (thick line) and their normalized standard deviation (colored area) during the cell cycle are shown. Globally, these profiles present little changes compared to the ones of the previous model (Figure~\ref{fig:mod1_profile}). In the left figure of the inset of Figure~\ref{fig:mod2_va_ratio_and_profile_AdK} is shown the profile of the protein Adk (a closeup of the first profile of Figure~\ref{fig:mod2_profiles}). It appears that the mean concentration at any given time $t$ of the cell cycle $\E{P(t)/V(t)}$ (the thick line of the profile) is not constant during the cell cycle, as it was the case in the model of Figure~\ref{fig:mod1}. The curve of $\E{P(t)/V(t)}$ fluctuates around $2\%$ of the global average protein concentration $\Ehat{P/V}$.

The main Figure~\ref{fig:mod2_va_ratio_and_profile_AdK} shows the effect of replication on the variance: it represents the ratio of protein variance between the previous model (with genes in constant concentration and random partitioning) and the current model with gene replication.  For all the genes, the variances predicted show little difference from the previous intermediate model. The ratio is even surprisingly slightly above one for many genes, indicating that for these genes the replication tends to reduce the variance.

As for the previous intermediate model, we can use the environmental state decomposition to separate the part of variance $\Vhatt{P/V}{int}$ (defined in Equation~\eqref{eq:var_1}) specific to the gene expression and $\Vhatt{P/V}{ext}$ (defined in Equation~\eqref{eq:var_2}) attributed to cell cycle fluctuations. As the mean concentration, $\E{P(t)/V(t)}$ is no longer constant during the cell cycle $\Vhatt{P/V}{ext}$ is no longer null. Yet it appears that $\Vhatt{P/V}{ext}$ only represents a very small part of the global variance $\Vhat{P/V}$ (for $99\%$ of the genes, it represents less than $2\%$). For  this intermediate model, the extrinsic contribution of DNA replication computed with this decomposition is small.

Using the analytical formula of protein variance (Section~2.1.2 %\ref{subsec:S_mod2_Protein-number}
 of S1~Appendix) and by performing variations on certain parameters, we can analyze the effect of several aspects on protein variance in this model.
By considering a given protein (protein Adk), we subsequently modified different parameters while making sure that the average protein concentration remain unchanged, to see how each of these changes impact the protein variance:
\begin{itemize}
	\item Changing the position of the gene on the chromosome (specifically changing $\tau_R$ and slightly adapting the gene activation rate $\lambda_1$ to keep the same average mRNA concentration)
	\item Changing the mRNA lifetime (by increasing the gene activity  $\lambda_1$ and increasing the mRNA degradation rate $\sigma_1$ so that it keeps the same average mRNA concentration)
	\item Changing the mRNA number (by increasing the gene activity  $\lambda_1$ while decreasing the mRNA activity $\lambda_2$ so that it keeps the same average protein concentration)
\end{itemize}

Results are shown in S2 Fig(C) of S1 appendix. %Figure~\ref{fig:S2C}.
Changing the gene position have almost no impact on global protein variance $\Vhat{P/V}$. The effect of mRNA lifetime is more noticeable as a shorter mRNA lifetime can diminish protein variance at most about $40\%$. The mRNA number seems to have the most important effect on protein production: for the same average protein concentration, having more mRNAs greatly diminish protein variance; such effect has been experimentally observed \citep{blake_noise_2003,ozbudak_regulation_2002}. This can be interpreted as lower bursting effect in protein production: as it is known that mRNAs in few copies with large activity display a protein production with large bursts, conversely a large number of mRNAs less active leads to a more stable protein production.

The right insight of S2 Fig(C) of S1 Appendix  %~\ref{fig:mod2_va_ratio_and_profile_AdK}
shows an example of such a protein with reduced variance: this protein is based on Adk, the protein average production is the same but there are ten times more mRNAs, with a ten times shorter lifetime. The variance is indeed reduced but with the cost of the production of additional mRNAs. Yet, even in this case, the protein expression is not strongly cycle-dependent (see inset of S2 Fig(D) of S1 Appendix); %Figure~\ref{fig:S2D});
 in particular, this profile is not precise enough to be used as a ``trigger'' for periodic cell events (such as DNA replication initiation, or partition at division): the evolution of the protein concentration across the cell cycle is not precise enough to robustly distinguish different phases of the cell cycle.

\subsection*{Impact of the Sharing of RNA-polymerases and Ribosomes }\label{subsec:model3}

\begin{figure}%fullwidth
	\insertfig[.8]{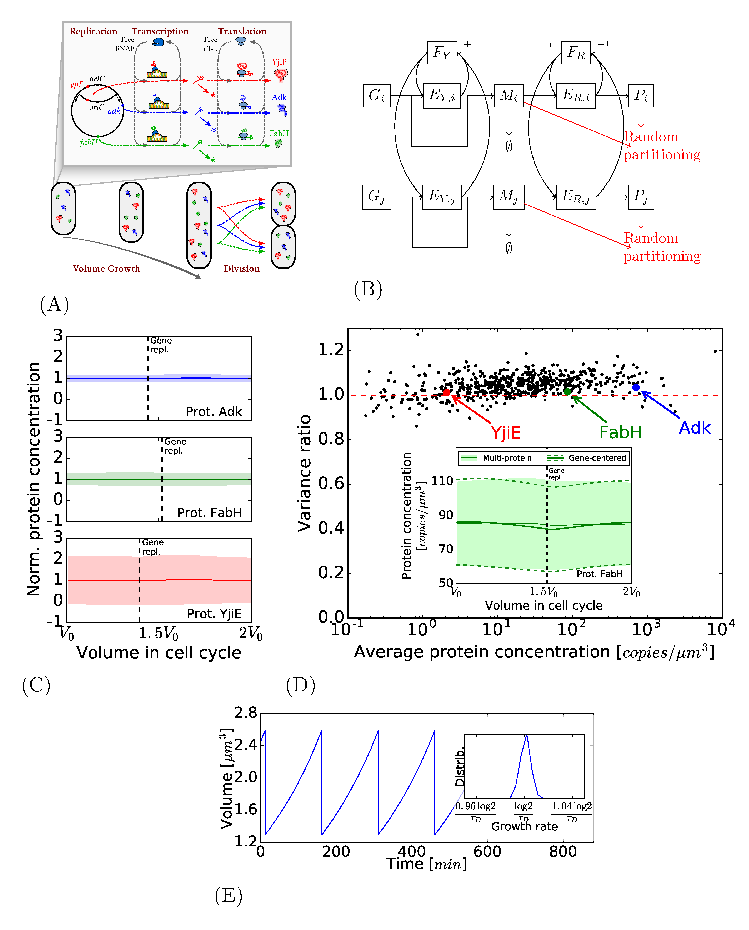}
	{
	\phantomsubcaption\label{fig:mod3_principle}
	\phantomsubcaption\label{fig:mod3_scheme}
	\phantomsubcaption\label{fig:mod3_profiles}
	\phantomsubcaption\label{fig:mod3_var_ratio}
	\phantomsubcaption\label{fig:mod3_exp-growth}
	}
\caption{\label{fig:mod3}\textbf{Complete model.} (A): The model now considers the sharing of RNA-polymerases and ribosomes between the different productions. (B): The model of production of all proteins. The $i$-th gene is associated with one particular type of mRNA (whose number is $M_{i}$) and proteins (whose number is $P_{i}$).  The number of free RNA-polymerases (resp. ribosomes) is $F_{Y}$ (resp.  $F_{R}$), the number of those sequestered on the $i$-th gene is $E_{Y,i}$ (resp. $E_{R,i}$). (C): The evolution of the normalized protein concentration during the cell cycle for three proteins. The central thick line and the colored areas represent the same quantities as in Figure~\ref{fig:mod1_profile}. (D) Main: variance in the previous model with replication divided by the variance in this model with sharing of RNA-polymerases and ribosomes. Insight: the concentration through the cell cycle for the protein FabH (a close-up of the one presented in Figure~\ref{fig:mod3_profiles}).  (E) Main: A simulation sample that shows that cell volume grows exponentially from around $V_{0}$ up to around $2V_{0}$. Insight: The growth rates of the simulation are centered around the expected growth rate $\log2/\tau_{D}$. }%figsupp mod3
\end{figure}

We now consider the complete model, which includes a limited amount of RNA-polymerases and ribosomes: the model is explained in detail in the Materials and Methods. Now, RNA-polymerases and ribosomes are explicitly represented in the model and each of these macromolecules is considered either allocated (i.e. sequestered on a gene if it is an RNA-polymerase, or on an mRNA if it is a ribosome), or free (i.e. either moving freely in the cytoplasm or, in the case of RNA-polymerases, potentially non-specifically sliding on the DNA).

The previous intermediate models were ``gene-centered'': each class of proteins was considered independently from each other.  The common sharing of RNA-polymerases and ribosomes is an additional key feature that leads us to investigate a multi-protein model where all the genes of the bacteria are considered altogether. For each type of protein $i$, we denote by $G_{i}(t)$, $M_{i}(t)$ and $P_{i}(t)$ respectively the number of gene copies, of messengers and of proteins at time $t$ in the cell cycle. For each gene $i$, $E_{Y,i}(t)$ is the number of RNA-polymerases sequestered on the $i$-th gene for transcription and $E_{R,i}(t)$ is the number of ribosomes sequestered on an mRNA of type $i$ for translation. The non-allocated RNA-polymerases and ribosomes are respectively denoted by $F_{Y}(t)$ and $F_{R}(t)$.  In a first step, we have considered that the gene pool of \cite{taniguchi_quantifying_2010} ($841$ genes with their mRNA and protein expression measured) would represent the whole genome. We will later see that the addition of new genes does not change significantly  our results.

New ribosomes and RNA-polymerases are added to the system as cell volume increases: in a first step, these macromolecules are regularly added such as their total concentration in the cell remains constant during the cell cycle; we will later consider a more realistic way to represent RNA-polymerase and ribosome production.

The previous intermediate models represented the production of a specific type of protein immersed into a ``background environment'' where the cell grows and divides, the current model includes simultaneously all the genes altogether. In this model, as we are on the scale of the whole cell, we would like to model the impact of global protein production on the cell growth. We therefore can no longer consider that the production of each type of protein has no effect on the global performance of the cell. The volume $V(t)$ depends now on the global production of proteins, and it is not an independent and deterministic feature anymore. As the density of cell components tends to be constant in real-life experiment \citep{marr_growth_1991} and proteins represent more than half of the dry mass of the cell \citep{neidhardt_chemical_1996}, the model considers the volume as proportional to the current total mass of proteins in the cell. The mass of each protein is given by the length of its peptide chain (see Section~3.2.1 %\ref{subsec:S_mod3_DeterModel-pres}
 of S1~Appendix for an exhaustive description of the model). Like in the previous model, division occurs when the cell reaches the volume $2V_0$ (making it a ``sizer'' model).

The processes of mRNA and protein productions are both separated in two parts: the binding and initiation on one side, and the elongation and termination on the other side. The rate at which an RNA-polymerase is sequestered on a gene of type $i$ at time $t$ depends on the copy number of the $i$-th gene $G_{i}(t)$, the free RNA-polymerase concentration $F_{Y}(t)/V(t)$ and a parameter $\lambda_{1,i}$ specific to the gene that takes into account the RNA-polymerase\textendash promoter affinity. The elongation rate of each mRNA only depends on the average transcription speed and the length of the gene. The mechanism for translation is similar.

As this model is more complex than the previous ones, the complete analytical description of mRNA and protein dynamics seems to be out of reach. To address this problem, we try to predict the average behavior of this model using a system of ordinary differential equations (ODEs). We used the predictions of these equations to fix the parameters. An \emph{a posteriori} validation has been made to check that this system of ODEs well predicts the average behavior of the stochastic model (see Section~3.2.1 %\ref{subsec:S_mod3_DeterModel-pres}
 of S1~Appendix for more details on the procedure).

In Figure~\ref{fig:mod3_exp-growth} is shown an example of the cell volume evolution in a simulation. It appears that the volume seems to grow exponentially during the cell cycle. The number of available ribosomes and RNA-polymerases also changes rapidly, of the order of the second for the RNA-polymerases, and of the order of one tenth of a second for the ribosomes (see S3 Fig(C) and S3 Fig(D) of S1 Appendix).
% Figure~\ref{fig:S3C} and Figure~\ref{fig:S3D}).

Figure~\ref{fig:mod3_var_ratio} compares the results of the simulations with the previous intermediate model with gene replication and random partitioning. It shows that, for $90\%$ of the genes, the interactions between protein productions only represent at most $10\%$ of variability.  In inset is shown the example of the protein FabH profile during the cell cycle, showing that sharing of RNA-polymerases and ribosomes introduces little change.

We also analyze the model using the environmental state decomposition.  Two genes of the dual reporter technique would undergo the same volume growth with the same evolution of free RNA-polymerases and ribosomes; as a consequence, the common cellular environment on which the decomposition is operated now includes the concentration of free RNA-polymerases and ribosomes (in addition to the cell cycle). Details can be found in Section~3.7 % \ref{subsec:S_mod3_ESD}
 of S1~Appendix.  It appears that for all the genes, the extrinsic contribution of the variance $\Vhatt{P_{i}/V}{ext}$ represents only a very small portion of $\Vhat{P_{i}/V}$ (for $99\%$ of the cells, the ratio $\Vhatt{P_{i}/V}{ext}/\Vhat{P_{i}/V}$ represents less than $1\%$).

We compared these simulation results with a simplified theoretical model: Section~3.6 %\ref{sec:S_mod3_Other-influences}
of S1~Appendix presents a multi-protein model, that is inspired by the one described in \cite{fromion_stochastic_2015}.  Even if it is a multi-protein model as it represents the expression of a large number of genes altogether, it is a simpler model than the one presented here as it considers separately transcription and translation, and it does not consider neither volume growth, partitioning at division, nor DNA replication. We show that the predicted distributions of free RNA-polymerases (and ribosomes in the adapted model) fits well the one observed in our simulations (see S4 Fig). %Figure~\ref{fig:S4}).
 As we will see in the discussion, this good correspondence between the models would suggest that the mean-field mathematical properties proven for the simplified model could be applied to our complete model.

\subsection*{Model and Parameter Sensitivity}

\begin{figure}
	\insertfig{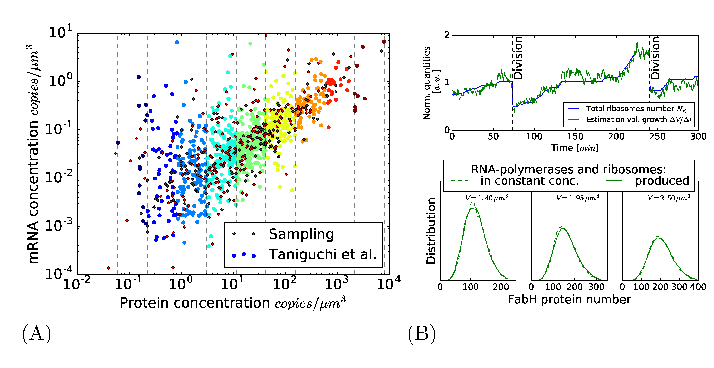}
	{
	\phantomsubcaption\label{fig:additional-genes}
	\phantomsubcaption\label{fig:production-RNAP-ribo}
	}
\caption{\label{fig:param_sensibility}\textbf{ Model and parameter sensitivity.} (A): Creation of additional genes by sampling of mRNA and protein concentrations in accordance to their correlation (see Section~3.6.1 %\ref{subsec:S_Other-Additional-genes}
 of S1~Appendix for details on the procedure). (B):  A variation on the complete model where the RNA-polymerases and ribosomes are produced by the cell. Above: clear correlation between the total number of ribosomes and the estimation of volume growth. Below: comparison between the two versions of the model (with or without production of RNA-polymerases and ribosomes); the distribution of protein number distribution of FabH protein for cells of different volumes in both cases.}%figsupp param_sensibility
\end{figure}

The complete model supposes a series of modeling and parameter choices that might legitimately influence protein production. We have analyzed several of these aspects and have  shown that they do not appear to  significantly change the results previously presented.
\begin{description}
\item [{Quantity of free RNA-polymerases and ribosomes}.] The average concentration of free RNA-polymerases and ribosomes in a cell cannot be deduced from Taniguchi~et al. They are nevertheless needed to estimate the parameters of our model, see Section 3.3 %~\ref{subsec:S_mod3_param}
 of S1~Appendix. Globally, one can expect to have a low concentration of free ribosomes and a higher concentration of free RNA-polymerases, see Section~3.5.1  %ref{subsec:F-RNAP-and_ribo}
 of S1~Appendix. But precise numbers  seem to be difficult to obtain.  We therefore perform several simulations with different values for these concentrations (for each macromolecule, a concentration taking successively 1, 10, 100 and 1000~$\text{copies}/\text{\textmu m}^3$), without significant changes. See S4~Fig %Figure~\ref{fig:S4}
 and  Section~3.5 %\ref{sec:Impact-RNAP-ribo}
 of S1~Appendix for details.
\item [{Additional Genes}.] As previously said, to perform our first simulations, only 841 genes from which the average mRNA and protein concentration have been measured in Tanuguchi~et~al. are considered. To have a pool of proteins that might represent a global gene expression in \emph{E. coli}, we studied the case of a simulation with a  set of genes that represent about 2000 genes, more in  accordance to the expected number of genes expressed in a growing condition. To propose realistic parameters for these fictional genes, we sample them according to different empirical distributions observed in the empirical data, and also by taking into account the possible correlations observed (the correlation that exists between the average mRNA and protein concentration for instance). See Figure~\ref{fig:additional-genes} and Section~3.6.1 %\ref{subsec:S_Other-Additional-genes}
 of S1~Appendix. No changes in protein concentration variance is observed.

\item [{Non-specific binding of RNA-polymerases}.] It has been proposed that many of the RNA-polymerases are non-specifically bound on the DNA (see \cite{klumpp_growth-rate-dependent_2008} for instance). We have done a simulation where RNA-polymerase can bind non-specifically on the DNA. When in this state, they are not available for the transcription. As previously it does not appear to change the protein expression behavior.
 See Section~3.6.3 %\ref{subsec:Other-Sequestered-polymerases}
 of S1~Appendix.

\item [{Production of RNA-polymerases and ribosomes}.] The total amount of RNA-polymerases and ribosomes (whether free or not) were at first considered in constant concentration: the RNA-polymerases and ribosomes were added as cell volume increases. We have done a simulation that considers a way to represent their production to have a more realistic representation: both RNA-polymerases and ribosomes are produced as if they were one of the proteins of the system (this goal of the simulation is just to have an insight of the effect of RNA-polymerase and ribosome production, not to represent precisely their production mechanisms). The introduction of such mechanism indeed changes some aspects of the simulation: in particular, the growth of the cell is more erratic as it then directly correlated with the total number of ribosomes (see Figure~\ref{fig:production-RNAP-ribo}, above). But as the production of proteins increases with a higher number of ribosomes, so does the volume of the cell. In terms, of concentration, the induced fluctuations in the number of ribosomes have little impact in terms of protein concentration variance (see Figure~\ref{fig:production-RNAP-ribo}, below). See Section~3.6.2 %\ref{subsec:S_Other-Production-of-RNAP-ribo}
 of S1~Appendix for more details.

\item [{Precision in the division and DNA replication initiation timing}.] The initial simulation triggers DNA replication initiation and division when the cell reaches some volume. Yet fluctuations in the timing of division has been proposed to have an impact on the protein variance \cite{thomas2019}.
 We investigate approximate models of division and of DNA replication initiation by introducing a volume-dependent rate of division as it is commonly used in the literature, see \cite{tyson_sloppy_1986,wang_robust_2010,soifer_single-cell_2014,osella_concerted_2014} (see Section~3.6.4 %\ref{subsec:Uncertainty-in-divi-init}
 of S1~Appendix).
 As previously (when we investigated the effect of the production of RNA-polymerases and ribosomes), if the protein \emph{number} is impacted, the protein \emph{concentration} remain relatively unchanged due to the fact that fluctuations induced by uncertainty in the timing of division affect correlatively the protein number and the cell volume.

\end{description}

\section*{\label{discussion}Discussion}

\subsection*{Interpretation of the Model Predictions}

The experimental data of Taniguchi et al. gave us the opportunity to systematically and quantitatively inspect the impact on protein variance of many cell mechanisms that are not often considered in stochastic models of protein production. The broad variety of genes experimentally measured in Taniguchi et al. has been a good opportunity for us to realistically test our models for a wide number of different genes, with different mRNA and protein concentrations, different mRNA lifetimes or gene position on the chromosome or gene length.

From this analysis, it appears that among the different features included in the model, the random partitioning has the most significant effect on the variance of protein concentration, especially for the less expressed proteins. We recover here one of the conclusions made in \cite{thomas2018}. The gene replication induces little difference (it even tends, to a small extent, to reduce the variance in some cases); the important fluctuations of free RNA-polymerases and ribosomes have little impact on protein production, which does not fit the hypothesis made in \cite{taniguchi_quantifying_2010}. It is confirmed by the environmental state decomposition, which separates the intrinsic and extrinsic contribution to protein variance (in an analog way as it is done with the dual reporter technique): the extrinsic contributions represent at most a few percents of the total variance.

We interpret the surprising little impact of sharing of RNA-polymerases and ribosomes on the proteins variance by noticing the similarities of our model with the one described in \cite{fromion_stochastic_2015}. Indeed, as previously explained, the global behavior of free RNA-polymerases and ribosomes can be predicted by a simplified model derive from \cite{fromion_stochastic_2015}, where the RNA-polymerases and ribosomes are shared among the different productions. The main result of \cite{fromion_stochastic_2015} is a mean-field theorem: as the number of genes increases, the production process of different types of proteins can be seen as independent production processes. The reason is that the dynamic of free RNA-polymerases and ribosomes is much faster than the production of mRNAs and proteins of one particular type. The rate at which an mRNA or protein is elongated only depends on the ``local steady state'' concentrations of free RNA-polymerases and ribosomes (a similar phenomenon can be found in \cite{dessalles_stochastic_2017b}). Our model seems to display such mechanism: with a global sharing of RNA-polymerases and ribosomes by a large amount of protein productions, the dynamic of free RNA-polymerases and ribosomes is faster than the production of each mRNA and protein of each type. As a consequence, it is not surprising to see that this multi-protein model, which takes into account the production of all proteins displays little difference with the intermediate ``gene-centered'' model due to a mean-field effect.

\subsection*{Comparison with Experimental Measures}

In the end, we can compare the results of our models with the experiments. Firstly, one can remark that the profile of the mean production $\E{P(t)/V(t)}$ (the plain line in Figure~\ref{fig:mod2_va_ratio_and_profile_AdK} is representative of all cells) during the cell cycle corresponds to the one observed experimentally in \cite{walker_generation_2016}. Furthermore, the maximum deviation of the average $\E{P(t)/V(t)}$ around the global average protein $\Ehat{P/V}$ (red dashed line in Figure~\ref{fig:mod2_va_ratio_and_profile_AdK}) is between $2\%$ and $4\%$ for all the proteins of our models, and  Walker et al. measure such fluctuations also around $2\%$ of the global average for genes at different positions on the chromosome (see Figure 1.d and Figure S6.b of \cite{walker_generation_2016}).

Secondly, we compared the global mRNA  and protein fluctuations predicted in our models with those measured in Taniguchi et al.  Figures~\ref{fig:final_pannel1} and~\ref{fig:final_pannel}  shows respectively, for every gene, the protein CV of mRNAs and proteins (resp. defined as $\Vhat{M/V}/\Ehat{M/V}^{2}$ and  $\Vhat{P/V}/\Ehat{P/V}^{2}$) against their respective average concentration (resp. $\Ehat{M/V}$ and $\Ehat{P/V}$); it is compared in both cases with the same results obtained experimentally (respectively corresponding to measurements shown in Figures~2.D and~2.B of \cite{taniguchi_quantifying_2010}). For the mRNAs (Figure~\ref{fig:final_pannel1}), the noise globally scales inversely the average mRNA concentration. Experimental measurements in \cite{taniguchi_quantifying_2010} -- made using the FISH technique for 137 of the highest expressed mRNAs -- show a similar tendency (the CVs are normalized on the figure because the uncertainty in the cell volume at birth can introduce a shift, but the shape of the normalized CV of both experimental and simulated CV remains exactly the same regardless of this effect). For the protein CV (Figure~\ref{fig:final_pannel}), it appears that the noise approximately scales inversely the average protein concentration like in the first ``intrinsic noise'' regime of \cite{taniguchi_quantifying_2010}.  But unlike in the \cite{taniguchi_quantifying_2010} experiment, there is no lower plateau for highly expressed proteins: for the highest produced proteins, the CV should be in the order of $10^3$ fold higher than the one predicted. It confirms that the features considered here cannot correctly explain the noise observed experimentally.

\begin{figure}
	\insertfig{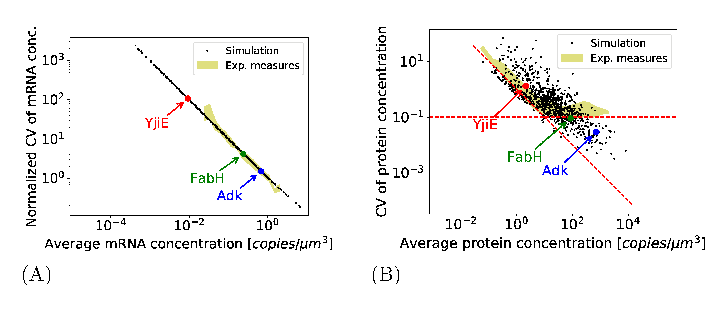}
	{
	\phantomsubcaption\label{fig:final_pannel1}
	\phantomsubcaption\label{fig:final_pannel}
	}
\caption{ \textbf{Comparison with experimental dataset of Taniguchi et al.} (A): The normalized CV of mRNA concentration
(defined as $\protect\Vhat{M/V}/\protect\Ehat{M/V}^2$) for each gene
predicted by the complete model (in yellow, the corresponding experimental measurements).
(B): The CV of protein concentration for each gene
predicted by the complete model (in yellow, the corresponding experimental measurements).}
\end{figure}

We can propose different interpretations to explain the discrepancy between the predictions of the models and the experimental measures. For the  biological processes not included in our models and that can have an impact on the variability, one can first mention the gene regulation as in our models, all the genes are considered as constitutive. The introduction of a gene regulation can indeed induce a large variability in protein concentration \citep{paulsson_models_2005,shahrezaei_analytical_2008}. Nonetheless, the ``extrinsic noise plateau'' observed in Taniguchi et al. only concerns the proteins and not the mRNAs (compare the yellow areas of Figures~\ref{fig:final_pannel1} and ~\ref{fig:final_pannel}). As a consequence, one can expect that the mechanism explaining the extrinsic noise plateau takes place at the translation level and not at the transcription. Moreover for highly expressed proteins, the protein CV is independent of the protein expression; it is therefore not gene-dependent as it would be the case for gene regulation. Finally, we have considered a simple model with gene regulation (like the three-stage model of \cite{paulsson_models_2005}), and determined the regulations parameters in order to predict protein variance observed in \cite{taniguchi_quantifying_2010}; we came with an activation/deactivation timescale has to be very high (several times the doubling time) in order to reproduce the variance experimentally observed, which is way above the typical biologically expected parameters.

One can also mention other possible mechanisms not represented in our models such as the fluctuations of availability of amino-acids or free RNA nucleotides in the medium, thus inducing additional fluctuations in the translation speed. Even if one can see here a clear analogy with the fluctuations in RNA-polymerase and ribosome availability (which also impact the transcription and translation speeds), the different timescales of the dynamics of amino-acids or free RNA nucleotides might result in a different effect. One can also challenge the hypothesis time intervals between events are modeled by exponentially distributed random variables: for instance, elongation times would be better represented as having Erlang distribution, that is the sum of independent exponential random variables. Yet, some results incline to say that it has a limited impact \citep{fromion_stochastic_2013}. Also, in this model, the binding and initiation (either of RNA-polymerases or ribosomes) are considered as a single event. A more precise representation would be to describe them as two different processes (Reference \cite{siwiak_transimulation_2013} gives for instance a median transcriptional initiation time of $15\,\text{t}$ which is of the same order of magnitude as the elongation time).

One can also consider that, as this effect mainly affects proteins with the highest fluorescence, it is possible that some saturation induces a bias in the estimation of variance of highly produced proteins. To our knowledge, exhaustive measures of \cite{taniguchi_quantifying_2010} for mean and variance of protein and mRNA concentrations have not been replicated at the same scale, so we have not been able to confront our results to other measures.

\section*{\label{methods}Materials and Methods}

\subsection*{\label{method:complete model} The Complete Model}
In this subsection is presented in detail the complete model of Figure~\ref{fig:mod3} that includes the sharing of RNA-polymerases and Ribosomes (the other models being simplifications of this complete model). It represents, for any gene, both the number of mRNA and protein molecules associated with the gene inside a given cell. But, contrary to the two-stage model (Figure~\ref{fig:mod_paulsson_scheme}) it also explicitly represents the volume $V(t)$ that is changing across the time $t$ due to cell growth, so that, if $M(t)$ and $P(t)$ represent the respective number of mRNAs and proteins of a given gene at any time $t$, one can now explicitly represent their concentration by
$$M(t)/V(t)\qquad\text{and}\qquad P(t)/V(t).$$
Furthermore, contrary to the two-stage model, all genes in the bacteria are represented altogether here (in order to represent the global sharing of RNA-polymerases and ribosomes in the different  productions of proteins) and the division (to represent a partitioning of components at septation). When dividing, the model focuses on only one of the two daughter cells in order to follow one lineage of cells.

\begin{figure}%fullwidth
	\insertfig{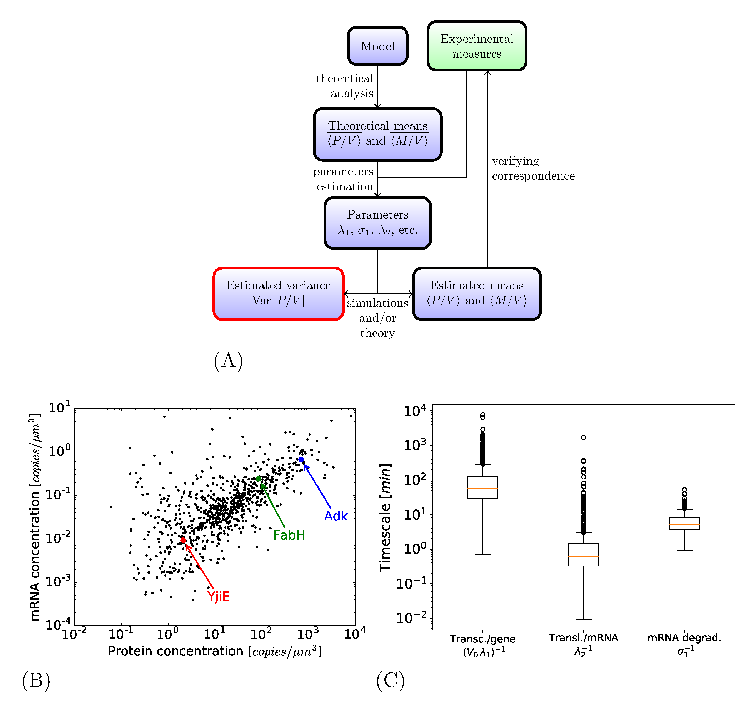}
	{
	% \phantomsubcaption\label{fig:Model-general}
	\phantomsubcaption\label{fig:procedure}
	\phantomsubcaption\label{fig:mrna_prot_taniguchi}
	\phantomsubcaption\label{fig:Quantitative-summary-Model0-1}
	}
\caption{\textbf{Experimental data and modeling principles.} (A): The complete model. (B): Scheme of analysis to determine the variance of protein concentration for every gene predicted in each intermediate model.  (C): \cite{taniguchi_quantifying_2010} measures of mRNAs and proteins for $841$ genes. (D): Box-plots presenting the parameters deduced from experiments to corresponds to the model of Figure~\ref{fig:mod1}; for the other models, these rates are in the same order of magnitude. }
\end{figure}

\begin{description}
\item [{Volume~growth~and~division}.] The volume of the cell is represented in the model and increases alongside the growth of the cell. As the number of mRNAs and proteins of each type is represented, this volume also makes it possible to explicitly represent their concentration inside the cell. When the cell doubles its volume, division occurs: it is a sizer model. All the compounds (mRNAs and proteins) are then randomly partitioned in the two daughter cells (this partition is considered as equally likely  as each compound has an equal chance to be in either one of the two daughter cells, we do not consider strong asymmetry in cell volume division).
Then the model only follows one of the two daughter cells beginning a new cell cycle.

\item [{Units~of~Production}.] Each type of protein has a specific type of mRNA and a unique gene associated with (in particular, there is no notion of operons). In the $i$-th unit of production, the number of gene copies $G_i$, mRNAs $M_i$, and proteins $P_i$ inside the cell is explicitly represented. Each copy of the gene can be transcribed into an mRNA.  The mRNA can then be translated into a protein until its degradation, the degradation rate is specific to the type of mRNA. We do not consider any rate of protein degradation:  the proteolysis occurring in a timescale much longer than the cell cycle (see \cite{koch_protein_1955}) for most proteins, its decay is then dominated by protein partition that occurs at division.

\item [{DNA-Replication}.] Each gene can be present in one or two copies in the cell (only one DNA replication is considered as in the slowly growing cells of Taniguchi et al.).  The instant of replication of each gene is simply determined by its position on the chromosome. When replicated, the rate of transcription of the gene is doubled.

\item [{RNA-polymerases~and~ribosomes}.] The production of mRNAs and proteins respectively requires RNA-polymerases and ribosomes. The concentrations of non-allocated (or free) RNA-polymerases and ribosomes respectively determine the rates of transcriptions and translations. During the time of elongation, the RNA-polymerase (resp. ribosome) remains sequestered on the DNA (resp. the mRNA). As the cell grows, new RNA-polymerases and ribosomes are created and participate in the production of proteins.
\end{description}

\subsection*{\label{method:analysis} Analysis of Each Intermediate Models}

Each of the intermediate models is systematically analyzed with the same method, see Figure~\ref{fig:procedure}. The average behavior of the model is analytically predicted (either with exact formulas for the first two intermediate models, or approximately for the last complete model); it makes it possible, for each of the 841 genes of Taniguchi~et~al.  (for which both the average protein and mRNA production have been measured), to determine the set of parameters of the model. An overview of the parameters hence determined is shown in Figure~\ref{fig:Quantitative-summary-Model0-1}.  Simulations are then performed \textemdash{} using methods derived from Gillespie algorithm \textemdash{} with these parameters over 10000 cell cycles: it makes it possible  to check the accuracy of the average concentration and to predict the variance of the concentration of proteins. For each model, we then have the variance of protein concentration predicted for the wide range of genes measured in \cite{taniguchi_quantifying_2010}.

\subsection*{\label{methods:mean variance} Means and Variances}
Throughout this paper we use the notation $\E{X(t)}$ and $\V{X(t)}$ for the mean and variance of a random variable $X(t)$ at a time $t$ of the cell cycle. We introduced the concentrations for mRNAs and proteins. If $M(t)$ and $P(t)$ denote the random variables representing the number of mRNAs and proteins of a given type at time $t$ and $V(t)$ is the volume at this instant, the corresponding concentrations are $M(t)/V(t)$ and $P(t)/V(t)$. One of the goals of this work is to study the properties of the  mean, $\E{P(t)/V(t)}$, and variance, $\V{P(t)/V(t)}$,  of these concentrations.  These quantities correspond to the mean and the variance of a population of synchronized cells of volume $V(t)$. The measurements of \cite{taniguchi_quantifying_2010} consider a cell population in exponential growth; by consequence, we also have to define the notions of global mean and variance for a heterogeneous population (see \cite{collins_rate_1962,sharpe_bacillus_1998,robert_division_2014} for the population distribution in exponential growth). By denoting by $u$ the age distribution of the cell population, we can define the global average $\Ehat{P/V}$ and variance $\Vhat{P/V}$ averaged over the population,
\begin{align}
\Ehat{P/V} & =\int_{0}^{\tau_{D}}\E{P(t)/V(t)}u(t)\,\diff t\label{eq:Ehat_P}\\
\Vhat{P/V} & =\int_{0}^{\tau_{D}}\left[\E{\left(P(t)/V(t)\right)^{2}}-\Ehat{P/V}^{2}\right] u(t)\,\diff t.\label{eq:Vhat_P}
\end{align}
We observe that the choice of the age distribution  $u$
(either uniform or corresponding to an exponentially growing population)
does not seem to have much impact; see Section~2.2.3 %\ref{subsec:other_pop_distrib}
 of S1~Appendix for more details. We therefore simply considered $u$ as uniform in the Results section.

\subsection*{\label{methods:Intrinsic extrinsic} Intrinsic and Extrinsic Effects}
As for the previous studies of \cite{elowitz_stochastic_2002,swain_intrinsic_2002}, we want to decompose protein variance that can be attributed to the intrinsic expression from the one due to extrinsic mechanisms.  It appears that this decomposition can be computed using two different ways:

\begin{itemize}
	\item First, there is the method used in Taniguchi~et~al. As the intrinsic noise is usually attributed to the protein production mechanism alone, a model that represents only transcription and translation,
	as the classic two-stage model, are usually considered as predicting the intrinsic noise.

	Yet, in our case, the consideration of a more realistic mechanism for protein disappearance (through segregation at division rather than regular decay) prevent us from using directly the two-stage models as a quantitative representation of the intrinsic variance. Our very first intermediate model, where the number of proteins and mRNAs are exactly halved at division, is considered as our baseline model. This baseline model is very close to the two-stage model in that sense that it contains no other features than those intrinsically linked to protein production. Therefore, we use protein variance predicted by this baseline model as intrinsic protein variance.

	From this baseline model, any additional variance added by the introduction to the model of external mechanisms (random partition at division, gene replication, etc.) would be considered as extrinsic. For every type of protein, we will look how the global variance $\Vhat{P/V}$ changes with the subsequent introduction of the external mechanisms.

 \item Secondly, it appears that the previous method of extrinsic noise estimation is not exactly the same from the first attempt made by \cite{elowitz_stochastic_2002} using the dual reporter technique. \cite{hilfinger_separating_2011} showed that decomposition using the dual reporter technique can be interpreted as an estimator of the environmental state decomposition (also called the law of total variance). It decomposes  protein variance between the effects specifically due to the stochastic nature of the instants of birth and death of mRNAs and proteins (intrinsic noise) and the external effect of the biological environment (extrinsic noise).  If $Z$ represents the state of the cell, the number of RNA-polymerases, the volume, etc\ldots, then the protein concentration $P/V$ can be decomposed such as,

 \begin{align}
 \Vhat{P/V}  =\underbrace{\Vhatt{P/V}{int}}_{\text{unexplained by }Z}
 +\underbrace{\Vhatt{P/V}{ext}}_{\text{explained by }Z},\label{eq:ESD}
 \end{align}
where:
\begin{align*}
	\Vhatt{P/V}{int} & =\int_{0}^{\tau_{D}}\E{\V{P(t)/V(t)|Z(t)}}_{Z(t)}u(t)\,dt\\
	\Vhatt{P/V}{ext} & =\int_{0}^{\tau_{D}}\left[\E{\E{P(t)/V(t)|Z(t)}^{2}}_{Z(t)}-\Ehat{P/V}^{2}\right]u(t)\,dt
\end{align*}
where $\Ehat{\cdot}_{Z(t)}$ indicates the integration over the variable $Z$ in a population of cells of age $t$. (Note that in the case of the two first intermediate models, the environment $Z$ is simply the deterministic volume and simplified expressions for $\Vhatt{P/V}{int}$ and $\Vhatt{P/V}{ext}$ in these cases can be found in Equations~\eqref{eq:var_1} and~\eqref{eq:var_2}.)

 The variable $Z$ represents the common environment in which the two similar genes of the dual reporter technique evolve; yet mathematically, it is dependent on the model that we consider (for each model, it is what is considered as being part of the ``environment" of the gene). We have therefore explicitly described for each intermediate models what it represents in this context (the second intermediate model with gene replication shows an illustrative example of this decomposition). Once the variable $Z$ fixed, the decomposition is explicit and separates the total variance in two parts: $\Vhatt{P/V}{int}$ that corresponds to the intrinsic contributions and  $\Vhatt{P/V}{ext}$ corresponds to the variance induced by external contributions represented by the environment $Z$ (volume fluctuation, concentrations of free RNA-polymerases and ribosomes, etc.) The term $\Vhatt{P/V}{int}$ is indeed the variance that can be expected from a model without any external fluctuation. A model like the two-stage model does not consider any change in the environment of protein production, so the term $\Vhatt{P/V}{ext}$  of the decomposition would be null.

\end{itemize}

 In the intermediate models, we quantify the external contributions by these two means, either by looking at the increase of the global variance $\Vhat{P/V}$, or by performing the environmental state decomposition and looking the portion of external variance $\Vhatt{P/V}{int}/\Vhat{P/V}$ predicted. We have seen that these two ways to quantify the external contributions of protein variance are not strictly equivalent.

%Our first intermediate model considers genes in constant concentration and no interaction between the protein productions; in addition to that, we will propose a particular version of it (considered as our baseline model), which also do not consider any effect of compound partition at division (see Section~\ref{subsec:model1}). This baseline model will be very close to the two-stage model in that sense that it contains no other features than those intrinsically linked to protein production. Therefore, we will use protein variance predicted by this baseline model as intrinsic protein variance.

\section*{Acknowledgments}
We would like to thank Stephanie Lewkiewicz, Marc Dinh and  Wolfram Liebermeister for  their critical reading of the manuscript.

% \newpage

% \makeatletter
% Default
% \def\fnum@figure{\figurename\nobreakspace\thefigure}
% \renewcommand{\fnum@figure}{\thefigure\nobreakspace\figurename}
% \makeatother

\setcounter{figure}{0}
\renewcommand{\figurename}{\hspace{-3pt}}
\renewcommand{\thefigure}{S\arabic{figure} Fig}
\setcounter{figure}{0}

\begin{figure}[H]
	\insertfig{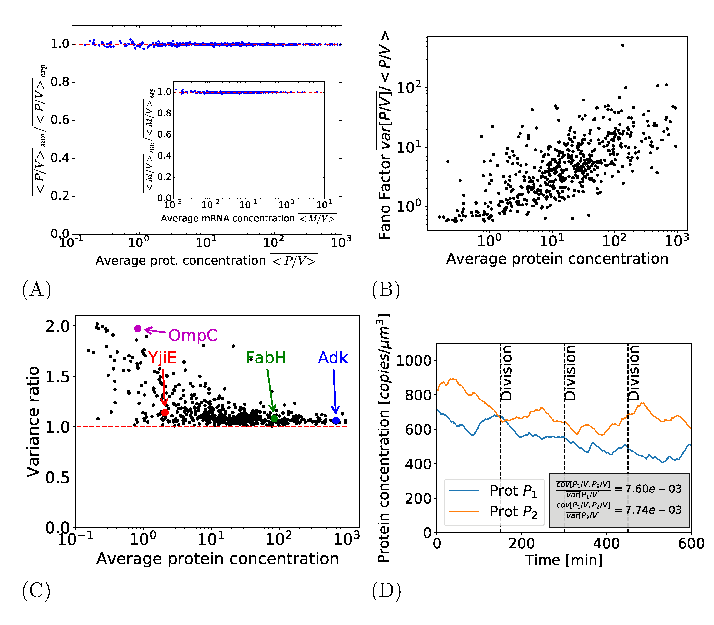}
	{
	\phantomsubcaption\label{fig:S1A}
	\phantomsubcaption\label{fig:S1B}
	\phantomsubcaption\label{fig:S1C}
	\phantomsubcaption\label{fig:S1D}
	}
\caption{\label{fig:S1} \textbf{Intermediate model with volume growth and partition at division: correspondence of simulations with experimental concentration, Fano Factor and average concentration, variance ratio and average concentration, example of dual reporter technique.} (A) Comparison
of the average productions of proteins (main) and mRNAs (inset) obtained
in the simulations and those experimentally measured. (B) The Fano
factor as a function of the average protein production: the production
with the lowest Fano Factor tends to be the less expressed. (C) For each
type of protein, the variance in the case of exact partition divided by
the variance in the case of random partition as a function of the
protein average production. (D) Example of the dual reporter technique with two promoters corresponding to the protein Adk.
}
\end{figure}

\begin{figure}[H]
	\insertfig{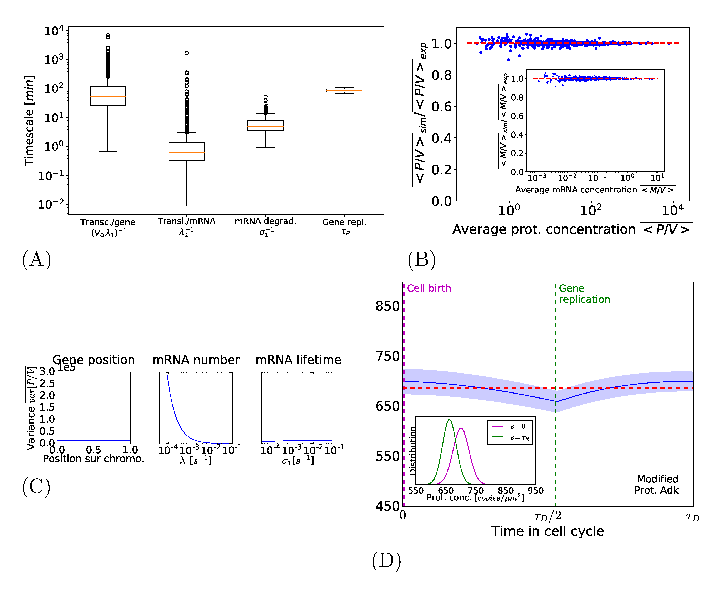}
	{
	\phantomsubcaption\label{fig:S2A}
	\phantomsubcaption\label{fig:S2B}
	\phantomsubcaption\label{fig:S2C}
	\phantomsubcaption\label{fig:S2D}
	}

	\caption{\label{fig:S2}  \textbf{Intermediate model with cell cycle and gene replication: parameters, correspondence of simulations with experimental concentration, influence of protein parameters on the the variance of its concentration, profile of a protein with an extreme low variance.} (A): Quantitative summary of the parameters for this model. (B): In the case of simulations, comparison of the average productions of protein (main) and mRNAs (inset) obtained in the simulations and those experimentally measured. Note that, in this case, in addition to the simulations, we can directly use theoretical formulas to directly predict the \emph{variance} of each protein (see Section 2 of S1~Appendix) (C): Evolution of $\protect\Vhat{P/V}$ while varying successively the gene position in the DNA, the mRNA number and the mRNA lifetime while keeping $\protect\Ehat{P/V}$ constant. (D): Main: Profile of a modified version of AdK with higher transcription rate (approximately ten times more) and a lesser mRNA lifetime (ten times less). The variance is reduced, but it is not enough to clearly separate between the distributions at birth (at time $t=0$) and at the replication of the gene (at time $t=\tau_{R}$) (Inset).}
\end{figure}

\begin{figure}[H]
	\insertfig{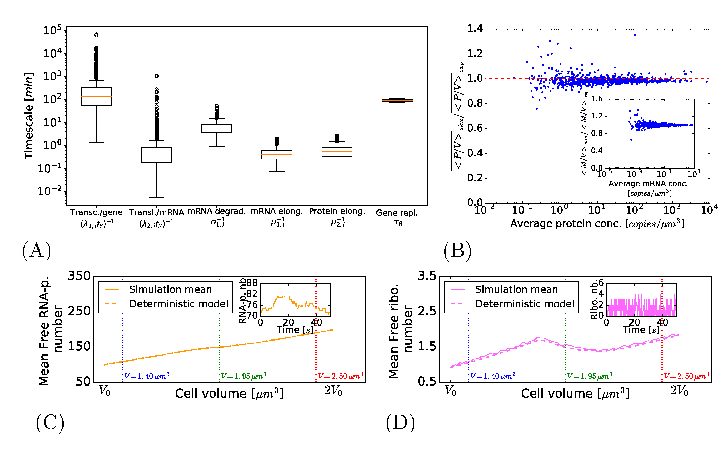}
	{
	\phantomsubcaption\label{fig:S3A}
	\phantomsubcaption\label{fig:S3B}
	\phantomsubcaption\label{fig:S3C}
	\phantomsubcaption\label{fig:S3D}
	}
	\caption{\label{fig:S3} \textbf{Complete model: parameters, correspondence of simulations with experimental concentration, evolution of free RNA-polymerases and free ribosomes.} (A): Quantitative summary of the parameters. (Different choice of $\overline{f_{Y}}$ and $\overline{f_{R}}$ when computing the parameters induce little changes for the rate of transcription per gene  $\lambda_{1,i}\overline{f_{Y}}$ and the rate of translation per mRNA  $\lambda_{2,i}\overline{f_{R}}$). (B): Ratio between the average concentration for protein (main figure) and mRNA (inset) in simulation and in experiments.  (C) and (D): the respective means of free RNA-polymerases and ribosomes at each moment of the cell cycle in the simulations (solid lines) and the ones predicted by the system of ODEs (dashed lines). Inset: an example of the dynamics of free RNA-polymerases and ribosomes for one simulation.}
\end{figure}

\begin{figure}[H]
    \insertfig[.8]{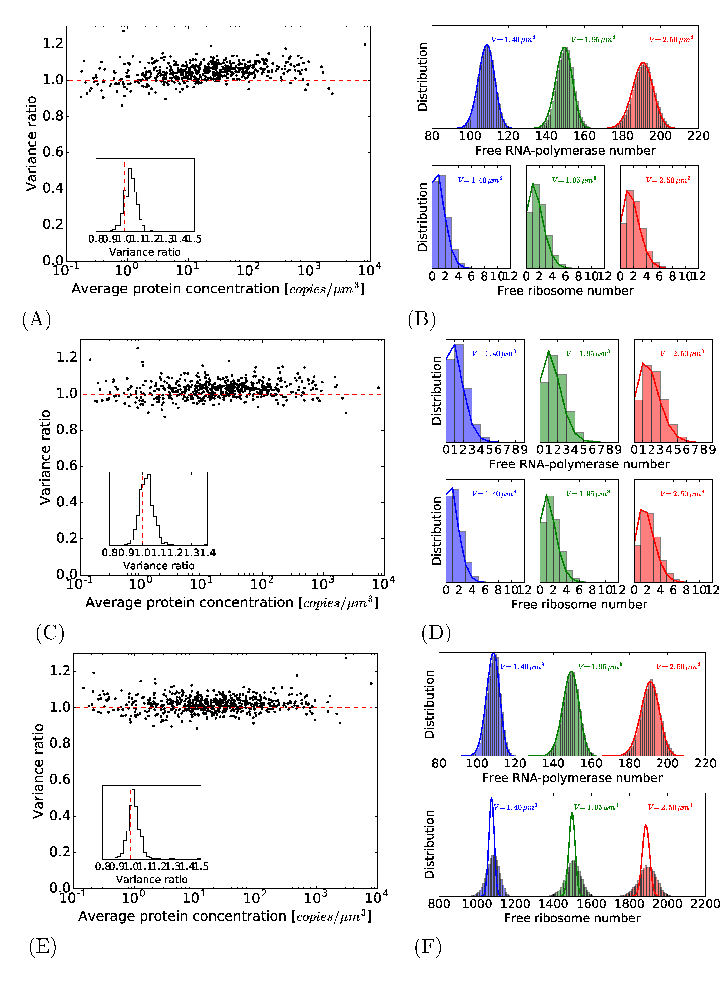}
	{
	\phantomsubcaption\label{fig:S4A}
	\phantomsubcaption\label{fig:S4B}
	\phantomsubcaption\label{fig:S4C}
	\phantomsubcaption\label{fig:S4D}
	\phantomsubcaption\label{fig:S4E}
	\phantomsubcaption\label{fig:S4F}
	}
\caption{\label{fig:S4}\textbf{Three simulations of the complete model:  different levels of free ribosomes and free RNA-polymerases} (1st line: few free ribosomes, many free RNA-polymerases; 2nd line: few free ribosomes, few free RNA-polymerases; 3rd line: many free ribosomes, many free RNA-polymerases). (A), (C) and (E): Ratio between protein variance of the multi-protein and the gene-centered models. Inset: the histogram of these variance ratios. (B), (D) and (F): Free RNA-polymerase (above) and free ribosome (below) number distribution for cells each of the volumes. In thick lines the binomial distribution predicted for the simplified model (see Section~3.8 %\ref{subsec:S_mod3_simplified_models}
 of S1~Appendix).}
\end{figure}

% \appendix
\newpage
\renewcommand{\thefigure}{(Appendix) \arabic{figure}}
\renewcommand{\theequation}{S\arabic{equation}}
\renewcommand{\thesection}{S\arabic{section}}
\renewcommand\thethm{S\arabic{thm}}

\setcounter{figure}{0}
\setcounter{equation}{0}
\setcounter{section}{0}
\setcounter{thm}{0}
\section*{ \label{app} Appendix}

\section{Impact of Random Partitioning}

Each of the three different models are analyzed following the scheme presented in Figure~7A %~\ref{fig:procedure}
 of the main article. We begin with a theoretical analysis to derive the average production of mRNAs and proteins along the cell cycle. Using these results, we then fit the parameters of the model to make them correspond to the experimental data of \cite{taniguchi_quantifying_2010}.  Then simulations are performed to predict protein variance for each gene considered.

We present now  results for the model with volume growth, constant gene concentration and partitioning at division which has been presented in Figure~2 %\ref{fig:mod1}
 of the main article. The results of the theoretical analysis part are similar for both cases of exact or random partitioning at division.

\subsection{Theoretical Analysis}

\subsubsection{Messenger-RNA Dynamic}\label{subsec:S_mod1_mrna_dynamic}

Here is shown the Proposition~\ref{prop:EM_s} that describes the average number of mRNAs at any instants of the cell cycle.

For any time $s\in\R_{+}$, denote by $M_{s}$ the number of mRNAs at this instant. We suppose that the initial time $s=0$ is a time of division; in this case, at each time $i\cdot\tau_{D}$ with $i\in\N$ are moments of division. For any $i\in\N$ , $M_{i\tau_{D}}$ denotes the number of mRNAs at the beginning of $i$-th cell cycle and $M_{i\tau_{D}-}$ the number of mRNAs in the $(i-1)$-th cell cycle just \emph{before} division.

We suppose that a lot of cell divisions have already occurred even before time $t=0$, and hence the considered cell cycle takes place when the embedded Markov chain $\left(M_{i\tau_{D}}\right)_{i}$ has already reached its steady state: the distribution $M_{i\tau_{D}}$ is the same as the distribution of $M_{(i+1)\tau_{D}}$.  If the steady state is already reached at time $0$, it implies that the distribution of any $M_{i\tau_{D}+t}$ for any $i\in\N$ and $s\in[0,\tau_{D}[$ is equal to the distribution of $M_{t}$. As a consequence, we can only consider the first cell cycle $t\in[0,\tau_{D}[$ to fully characterize the behavior of $M_{s}$ at any time $s\in\R_{+}$.

We propose here to describe the evolution of $\left(M_{t}\right)$ between times $0$ and $\tau_{D}$ (during this period of time, the number of mRNA approximately doubles). We first divide mRNAs into two categories,
\begin{itemize}
\item First group: mRNAs which were present at the birth of the cell. Each mRNA $i$ of the first group is characterized by $E_{\sigma_{1}}^{i}$, its lifetime given by an exponential random variable of rate $\sigma_{1}$.  The $i$-th mRNA still exists at time $t$ if and only if $E_{\sigma_{1}}^{i}>t$.  As a consequence, the number of mRNAs of this group still existing at time $t$ is given by
\begin{equation}
\sum_{i=1}^{M_{0}}\ind{E_{\sigma_{1}}^{i}>t}.\label{eq:mrna-first-group}
\end{equation}
\item Second group: mRNAs which have been created since the birth of the cell. The description of the number of mRNA of this group is more complicated. It is necessary to resort to the framework of Marked Poisson Point Processes (MPPP). An MPPP is a two-dimensional process.  It is based on a Poisson process where each of its random point is ``marked'' with another random variable; each point of a MPPP is a couple $(x,y)$ where $x$ is part of a Poisson point process and $y$ is the mark distributed according to a certain distribution.  One can refer to the first Chapter of \cite{robert_stochastic_2010} or \cite{kingman_poisson_1993} for the main results concerning MPPP.

We use this tool to characterize the number of mRNAs of the second group. In our case, the first variable $x$ represents the time at which the mRNA is created and the second variable $y$ represents the mRNA lifetime. Define $\mathcal{N}$ an MPPP of intensity
\[
\nu(\diff x,\diff y)=\lambda_{1}V(x)\diff x\otimes\sigma_{1}e^{-\sigma_{1}y}\diff y.
\]

It is noticeable that the underlying Poisson Process of this MPPP is not homogeneous. If the $i$-th mRNA of this group is born at time $x_{i}$ and its lifetime is $y_{i}$, then it exists at time $t$ if and only if $(x_{i},y_{i})\in\Delta_{t}$ with
\[
\Delta_{t}=\left\{ (x,y)\in\R_{+}^{2},\,0<x<t,\,y>t-x\right\} .
\]
One can refer to \ref{fig:MPPP-1}. Therefore the number of mRNAs of this group still present at time $t$ is given by
\begin{equation}
\mathcal{N}\left(\Delta_{t}\right)=\iint_{\R_{+}^{2}}\ind{(x,y)\in\Delta_{t}}\mathcal{N}\left(\diff x,\diff y\right).\label{eq:mrna-second-group}
\end{equation}

\end{itemize}
\begin{figure}%appendfig
\begin{centering}
\begin{tikzpicture}[
        scale=1.6,
        IS/.style={blue, thick},
        LM/.style={red, thick},
        axis/.style={very thick, ->, >=stealth', line join=miter},
        important line/.style={thick}, dashed line/.style={dashed, thin},
        every node/.style={color=black},
        dot/.style={circle,fill=black,minimum size=4pt,inner sep=0pt,
            outer sep=-1pt},
    ]
	\fill[black!10]  (1.8,0) --  (1.8,2) -- (0,2) --  (0,1.8)  -- cycle;
	\draw[line width=0.5mm,green] (2.8,2)--(2.8,0)
			node[pos=0.21,below,green,rotate=90]{Division};
	\draw[line width=0.5mm,green] (0.025,2)--(0.025,0)
			node[pos=0.15,below,green,rotate=90]{Birth};

	\node[black] at (1.2,1.3){$\Delta_t$};
    \draw[axis,<->] (3,0) node(xline)[right,align=left] {Time} -|
                    (0,2) node(yline)[above] {Lifetime};
	\coordinate (starpt) at (1.25,1);
	\foreach \x/\y/\z in {0.5/0.25/0, 1.4/0.6/0, 2/1.3/0.5, 0.2/0.4/0,
						2.3/0.7/0.2, 2.55/0.2/0}
		{
    	\node[dot,color=blue] at (\x,\y) {};
		\draw[dashed line,blue] (\x,\y)--($(\x+\y-\z,\z)$);
		}

	\node[dot,color=blue] at (starpt) {};
	\draw[dashed line,red] let \p1=(starpt) in
		(starpt)--($(\x1,0)$) node[below,red]{$t_i$}
		(starpt)--($(0,\y1)$)
		;
    \draw[dashed line,red] let \p1=(starpt) in
		(starpt)--($(\x1+\y1,0)$) node[below,red]{$t_i+\delta_i$};

	\draw[<->,red,line width=0.25mm] let \p1=(starpt) in
		(-0.1,0)--(-0.1,\y1) node[above,midway,rotate=90,red]{
						$\delta_i\sim{\cal E}(\sigma_{1})$};

	\draw[<->,black,line width=0.25mm] (0,-0.3)--(2.8,-0.3)
		node[below,midway]{Point Poisson Process ($\lambda_1$)};
	\draw[line width=0.5mm] (1.8,2)--(1.8,0) node[below]{$t$};

\end{tikzpicture}
\par\end{centering}
\caption{\label{fig:MPPP-1}\textbf{Illustration of the Marked Point Poisson Process describing the dynamic of mRNAs}: each mRNA is characterized by the point $(x_{i},y_{i})$ (with $(x_{i},y_{i})$ following the MPPP $\mathcal{N}$, whose distribution is of intensity $\nu$). The random variable $x_{i}$ represents the time at mRNA creation and $y_{i}$ its lifetime, hence this mRNA exists from volume $x_{i}$ up to volume $x_{i}+y_{i}$; that is to say that mRNA is still present at time $t$, if and only if the point $(x_{i},y_{i})$ is in the set with $\Delta_{t}=\left\{ (x,y)\in\protect\R_{+}^{2},\,0<x<t,\,y>t-x\right\} $. }
\end{figure}
By summing the number of mRNAs for each group (Equations~\eqref{eq:mrna-first-group}
and~\eqref{eq:mrna-second-group}), it follows the total number of
mRNAs present at time $t\in[0,\tau_{D}[$:
\begin{equation}
M_{t}=\sum_{i=1}^{M_{0}}\ind{E_{\sigma_{1}}^{i}>t}+\mathcal{N}\left(\Delta_{t}\right).\label{eq:M_s-1}
\end{equation}
This description of the dynamic of $M_{t}$, together with the steady state
hypothesis which implies that $M_{0}\overset{\mathcal{D}}{=}M_{\tau_{D}}$,
allows to prove the next proposition.
\begin{prop}
\label{prop:EM_s}At steady state, the concentration of mRNAs at time
$t\in[0,\tau_{D}[$ of the cell cycle is
\[
\E{M_{t}/V(t)}=\frac{\lambda_{1}\tau_{D}}{\sigma_{1}\tau_{D}+\log2}.
\]
\end{prop}
\begin{proof}
 By taking the mean of Equation~\eqref{eq:M_s-1}, it follows for
any time $t$ of the cell cycle:
\[
\E{M_{t}}=\E{\sum_{i=1}^{M_{0}}\ind{E_{\sigma_{1}}^{i}>t}}+\E{\mathcal{N}\left(\Delta_{t}\right)}.
\]

Since all $\left(E_{\sigma_{1}}^{i}\right)_{i}$ are i.i.d. and independent
of $M_{0}$, the first term is given by
\[
\E{\sum_{i=1}^{M_{0}}\ind{E_{\sigma_{1}}^{i}>t}}  =\E{M_{0}}e^{-t\sigma_{1}}.
\]
For the second term, one has to remark that as $\mathcal{N}$ is a
MPPP, $\mathcal{N}\left(\Delta_{\tau_{D}-}\right)$ is a Poisson random
variable (Proposition~1.13.a of \cite{robert_stochastic_2010}).
The parameter of this Poisson random variable is given by
\[
\nu\left(\Delta_{t}\right)  =  \iint_{\Delta_{t}}\nu\left(\diff x,\diff y\right)
  =  V_{0}\frac{\lambda_{1}\sigma_{1}}{\log2+\sigma_{1}\tau_{D}}\left(2^{t/\tau_{D}}-e^{-\sigma_{1}t}\right).
\]
As a consequence, one gets that, for any time $t$ in the cell cycle,
\[
\E{M_{t}}=\E{M_{0}}e^{-t\sigma_{1}}+V_{0}\frac{\lambda_{1}\tau_{D}}{\tau_{D}\sigma_{1}+\log2}\cdot\left(2^{t/\tau_{D}}-e^{-t\sigma_{1}}\right).
\]

We still have to specify the mean number of mRNAs at birth $\E{M_{0}}$.
At the end of the cell cycle, for $t=\tau_{D}-$, the average number
of mRNAs is given by
\[
\E{M_{\tau_{D}-}}=\E{M_{0}}e^{-\tau_{D}\sigma_{1}}+V_{0}\frac{\lambda_{1}\tau_{D}}{\tau_{D}\sigma_{1}+\log2}\cdot\left(2-e^{-\tau_{D}\sigma_{1}}\right),
\]
and since at steady state,
\[
\E{M_{\tau_{D}}}=\E{M_{0}}=\E{M_{\tau_{D}-}}/2.
\]
Hence
\[
\E{M_{0}}\left(2-e^{-\tau_{D}\sigma_{1}}\right)=V_{0}\frac{\lambda_{1}\tau_{D}}{\tau_{D}\sigma_{1}+\log2}\cdot\left(2-e^{-\tau_{D}\sigma_{1}}\right),
\]
which gives the result.
\end{proof}

In particular, as the mean $\E{M_{t}/V(t)}$ does not change across the cell cycle, the global mRNA average
$\Ehat{M/V}$ (as it is defined in Equation~(6) of the main article), does not depend on the choice of the population distributions $u$ in this case. It is given by,
\begin{equation}
\Ehat{M/V}=\frac{\lambda_{1}\tau_{D}}{\sigma_{1}\tau_{D}+\log2}.\label{eq:EhatM/V}
\end{equation}

\subsubsection{Protein Dynamic}\label{subsec:S_mod1_prot_dynamic}
The mean number of mRNAs is now determined for any moment of the cell
cycle. Each of the mRNAs potentially produces proteins at rate $\lambda_{2}$.
As for the mRNAs, we describe the number of proteins at time $t$
by grouping them into two categories.
\begin{itemize}
\item The $P_{0}$ proteins that were present at birth and which remain
in the bacteria during all the cell cycle (as said in the main article
the proteolysis is not considered in this model).
\item The proteins that have been created during the current cell cycle.
The rate of production is depending on the current number of mRNAs.
We consider $\mathcal{N}_{\lambda_{2}}^{i}$ (for $i\in\N$ and $i\geq1$)
independent Poisson Point Processes of intensity $\lambda_{2}$. If
the $i$-th mRNA exists at time $t$ (that is to say if $i\leq M_{t}$),
then the number of proteins produced by this mRNA between $t$ and
$t+\diff t$ is $\mathcal{N}_{\lambda_{2}}^{i}(\diff t)$.
\end{itemize}
To sum up, the number of proteins at a time $t$ of the cell cycle is given by
\begin{equation}
P_{t}=P_{0}+\sum_{i=1}^{\infty}\int_{0}^{t}\ind{i\leq M_{u}}\mathcal{N}_{\lambda_{2}}^{i}\left(\diff u\right)\mbox{.}\label{eq:P_s-1}
\end{equation}
The first term is the number of proteins at birth, and the second
take into account all the proteins created between times $0$ and
$t$. One can then determine the mean number of proteins at any time
$t$ of the cell cycle.
\begin{prop}
\label{prop:EP_s}At steady state, the concentration of proteins at
any time $t\in[0,\tau_{D}[$ of the cell cycle is
\[
\E{P_{t}/V(t)}=\frac{\lambda_{2}\tau_{D}}{\log2}\cdot\frac{\lambda_{1}\tau_{D}}{\sigma_{1}\tau_{D}+\log2}.
\]
\end{prop}

\begin{proof}
By taking the average of Equation~\eqref{eq:P_s-1}, one gets
\begin{multline*}
\E{P_{t}}  =\E{P_{0}}+\sum_{i=1}^{\infty}\E{\int_{0}^{t}\ind{i\leq M_{u}}\mathcal{N}_{\lambda_{2}}^{i}\left(\diff u\right)}=\E{P_{0}}+\sum_{i=1}^{\infty}\E{\int_{0}^{t}\ind{i\leq M_{u}}\lambda_{2}\,\diff u}\\
  =\E{P_{0}}+\lambda_{2}\int_{0}^{t}\E{\sum_{i=1}^{\infty}\ind{i\leq M_{u}}}\,\diff u=\E{P_{0}}+\lambda_{2}\int_{0}^{t}\E{M_{u}}\,\diff u.
\end{multline*}
As we know the mean number of mRNAs $\E{M_{u}}$ at  time $u$
of the cell cycle with Proposition~\ref{prop:EM_s},
\[
\E{P_{t}}=\E{P_{0}}+\frac{\lambda_{2}\tau_{D}}{\log2}\cdot\frac{\lambda_{1}\tau_{D}}{\sigma_{1}\tau_{D}+\log2}\cdot\left(V(t)-V_{0}\right).
\]
Since the system is at steady state, we have for time $\tau_{D}-$,
$\E{P_{\tau_{D}-}}=2\E{P_{0}}$; so
\[
\E{P_{0}}  =\frac{\lambda_{2}\tau_{D}}{\log2}\cdot\frac{\lambda_{1}\tau_{D}}{\sigma_{1}\tau_{D}+\log2}\cdot\left(V(\tau_{D}-)-V_{0}\right)
  =\frac{\lambda_{2}\tau_{D}}{\log2}\cdot\frac{\lambda_{1}\tau_{D}}{\sigma_{1}\tau_{D}+\log2}\cdot V_{0}.
\]
Consequently, for any time $t$ of the cell cycle,
\[
\E{P_{t}}=\lambda_{2}\frac{\tau_{D}}{\log2}\cdot\frac{\lambda_{1}\tau_{D}}{\sigma_{1}\tau_{D}+\log2}\cdot V_{0}\left(1+2^{t/\tau_{D}}-1\right);
\]
hence the result.
\end{proof}

In particular, as the mean $\E{P_{t}/V(t)}$ does not change across
the cell cycle, the global protein average is given by,
\begin{equation}
\Ehat{P/V}=\frac{\lambda_{2}\tau_{D}}{\log2}\cdot\frac{\lambda_{1}\tau_{D}}{\sigma_{1}\tau_{D}+\log2}.\label{eq:EhatP/V}
\end{equation}
\subsection{Parameters Estimation}\label{subsec:S_mod1_param}
For each gene measured in \cite{taniguchi_quantifying_2010}, we want to identify the set of corresponding parameters $\lambda_{1}$, $\sigma_{1}$ and $\lambda_{2}$. We also need to determine the ``global'' quantities $\tau_{D}$ and $V_{0}$. We first determine the parameters common to all genes. The division time $\tau_{D}$ is set to  $150$min in the article and the volume at birth $V_{0}$ is taken equal to $1.3\,\text{\textmu m}^{3}$.\footnote{The value of $V_{0}$, even if it is not explicitly given in \cite{taniguchi_quantifying_2010} can be deduced from the typical width given in its supplementary materials.}

Then we have to determine for each gene the three gene-specific parameters $\lambda_{1}$, $\sigma_{1}$ and $\lambda_{2}$. We consider the genes of the article for which was measured the empirical mean of messengers $\mu_{m}$ and proteins $\mu_{p}$ concentrations, as well as the mRNA half-life time $\tau_{m}$.

First we determine the rate of mRNA degradation for each gene with the measured mRNA half-life time $\tau_{m}$.  A half-life $\tau_{m}$ indicates that a mRNA has a probability $1/2$ to disappear within a duration $\tau_{m}$, hence $e^{-\sigma_{1}\tau_{m}}=1/2$. From that, we can compute the rate $\sigma_{2}$ (specific for each type of mRNA),
\[
\sigma_{1}=\log2/\tau_{m}\mbox{.}
\]

Then we can identify the averages of mRNA and protein concentrations of the model (respectively $\Ehat{M/V}$ and $\Ehat{P/V}$) with the empirical averages $\mu_{m}$ of mRNA concentration and $\mu_{p}$ of protein concentration of the article. With Equations~\eqref{eq:EhatM/V} and~\eqref{eq:EhatP/V}, the parameters $\lambda_{1}$ and $\lambda_{2}$ are
\[
\lambda_{1}  =\mu_{m}\cdot\frac{\sigma_{1}\tau_{D}+\log2}{\tau_{D}},\qquad
\lambda_{2}  =\mu_{p}\cdot\frac{\log2}{\tau_{D}}\cdot\frac{\sigma_{1}\tau_{D}+\log2}{\lambda_{1}\tau_{D}}.
\]

A summary of the different parameters can be seen in Figure~7C of the main article. %\ref{fig:Quantitative-summary-Model0-1}.
Having determined all the parameters allows to perform simulations of the model using stochastic algorithm in order to assess the variability of every protein and compare them with those experimentally obtained in \cite{taniguchi_quantifying_2010}.

\subsection{\label{subsec:S_mod1_Gillespie-Algorithms}Simulations}
When performing simulations, one needs to take care of the non-homogeneity of the Poisson processes describing mRNA creation times: the rate of protein production $\lambda_{1}V(t)$ is not a homogeneous rate as it changes with time. That does not allow a direct application of Gillepsie method \citep{gillespie_exact_1977},  an extension for non-homogenous processes has to be used.

Reference \cite{gillespie_exact_1977} describes an algorithm to simulate stochastic trajectories such as the quantities of different chemical species interacting together. The main idea is to consider the state of a system (for instance the number of each chemical compounds) and to compute the first reaction to occur, as well as the time when it happens. Once both pieces of information computed, one change the current state of the system accordingly with the reaction, and update the time.

One important hypothesis is that all reactions occur at exponential times (even if the rates of these exponential times may depend on the current state of the system).  In the current intermediate model, at any time $t$, the state is described by $\left(M_{t},P_{t}\right)$ (respectively, the number of mRNAs and proteins), and the rate of mRNA production is $\Lambda(t)=\lambda_{1}V(t)$ with $\lambda_{1}$ a parameter and $V(t)$ the non-constant volume of the cell. The parameter $\Lambda(t)$ does not depend on the state $\left(M_{t},P_{t}\right)$ but is time dependent through $V(t)$; for this reason, it is not an exponential time.

In this case, the duration of time $T$ until the next mRNA production is characterized by
\[
\P{T>x}=\exp\left(-\int_{0}^{x}\Lambda(t)\,\diff t\right),  \quad x>0.
\]
which is not an exponential distribution as $\Lambda$ is non-constant.  To compute $T$, we consider that $\Lambda(t)$ is strictly positive for any $t\in\R_{+}$, as a consequence $F(x):=\int_{0}^{x}\Lambda(t)\,\diff t$ is strictly increasing. Let $E$ be an exponential random variable with  parameter $1$. We have hence
\[
\forall y>0\quad\P{E>y}=\exp\left(-y\right).
\]
If we consider the case of $y=F(x)$, since $F$ is strictly increasing, hence
\[
\P{E>y}=\exp\left(-F(x)\right)\qquad\text{and}\qquad\P{E>y}=\P{F^{-1}(E)>x}.
\]
As a consequence the random variable $F^{-1}(E)$ has the same distribution as $T$.

Based on that we can propose a new version of the algorithm of Gillespie that can take into account non-exponential times such as $T$.
\begin{algo}
The equivalent of Gillespie algorithm that considers non-homogeneous events is
\begin{enumerate}
\item Initialization: Initialize time of molecules in the system and the
time.
\item Next exponential event: determine the next event that occurs at an exponential
time as in Gillespie algorithm.
\item Next non-homogeneous event: determine the next event that occurs at
non-homogeneous rates with the method previously described.
\item Update: choose between events of Step 2 or Step 3 that happen first.
Update the time and the molecule count accordingly.
\item Iterate: Consider again the Step 2 unless it has reached the end of
the simulation.
\end{enumerate}
\end{algo}

%% figS1

\subsection{\label{subsec:S_mod1_ESD}Environmental State Decomposition}

As explained in the main article, the dual reporter technique \citep{elowitz_stochastic_2002}
that compares the expression of two similar genes (with the same promoter
and RBS, at an equivalent position on the chromosome) in the same
cell can be interpreted as an estimator of the environmental state
decomposition \citep{hilfinger_separating_2011}. If $Z$ the cell
state (i.e., the common environment in which the two genes of the
dual reporter technique are expressed), then we can apply the law
of total variance on the protein number,
\[
\Vhat{P/V}=\underbrace{\vphantom{\left[\Ehat{P|}\right]}\Ehat{\Vhat{P/V|Z}}}_{\Vhatt{P/V}{int}}+\underbrace{\Vhat{\Ehat{P/V|Z}}}_{\Vhatt{P/V}{ext}}.
\]
In order to be applied here, we need to specify what does the gene
environment $Z$ refers to. In our model, two similar genes (with
the same parameters $\lambda_{1}$, $\sigma_{1}$,$\lambda_{2}$ )
in the same cell would undergo the same volume growth. But each gene
would undergo a specific partition at division (for instance, in the
case of random partitioning, the partitions of the proteins of one
of the gene, is uncorrelated with the partition of the proteins of
the other gene).

With only considering the volume as the gene environment, we end up
with the following decomposition,
\begin{align*}
\Vhatt{P/V}{int} & =  \frac{1}{\tau_{D}}\int_{0}^{\tau_{D}}\V{P_{t}/V(t)}\diff t,\\
\Vhatt{P/V}{ext} & =  \frac{1}{\tau_{D}}\int_{0}^{\tau_{D}}\E{P_{t}/V(t)}^{2}\diff t-\Ehat{P/V}^{2}.
\end{align*}

As in this model, the $\E{P_{t}/V(t)}$ remains constant during the cell cycle, the second term of the decomposition remains null.

It has been confirmed by a simulation of the dual reporter technique, where the expression of two identical promoters in the same cell has been compared (we took the example of the protein Adk). The concentrations of the two proteins, respectively $P_1/V$ and $P_2/V$, have been compared. In particular, we measured their covariance which is much smaller than their respective variances (see S1~Fig (D)).

\subsection{\label{subsec:S_mod1_simple_model_partition}Simplified model for
the random partitioning}

Here we explain the simplified model used to make the prediction of
protein noise ratio in blue dash line of Figure~2D %\ref{fig:mod1_var_ratio}
 of the main article.

For a  given  quantity $P$ associated with the gene, the partitioning can be performed in two ways, either exact or random. The result in each case  will  be denoted respectively by $P_{e}$ and $P_{r}$. During division, the volume is divided by two, changing from $2V_{0}$ to $V_{0}$. In order to be plotted in Figure~2D %\ref{fig:mod1_var_ratio}
 of the main article, we need to consider the variance of protein concentration after division
\[
\eta:=\frac{\V{P_{r}/V_{0}}}{\V{P_{e}/V_{0}}}\quad\text{as a function of\quad}x:=\frac{\V{P/(2V_{0})}}{\E{P/(2V_{0})}}.
\]

\begin{prop}
The Variance ratio $\eta$ as a function of $x$ is given by
\[
\eta=\frac{2V_{0}x+1}{2V_{0}x}.
\]
\end{prop}

\begin{proof}
We have the quantity before division $P$. Since by definition, we have
that $P_{e}=P/2$ and
\[
\E{P_{e}/V_{0}}  =\E{P/(2V_{0})}, \quad  \V{P_{e}/V_{0}}  =\V{P/(2V_{0})}.
\]

For the effect of binomial division, see Lemma~\ref{lem:Binom_div}, it describes the effect of the binomial division on the means and on the variances of several quantities. By the volume in order to observe the concentrations, one gets that
\begin{align*}
\E{P_{r}/V_{0}} & =\frac{\E P}{2V_{0}}=\E{P/\left(2V_{0}\right)}
\end{align*}
and
\[
\V{P_{r}/V_{0}}  =\frac{\V{P_{r}}}{V_{0}^{2}}=\frac{\V P+2\E{P_{r}}}{4V_{0}^{2}}
  =\V{P/\left(2V_{0}\right)}+\frac{\E{P/(2V_{0})}}{2V_{0}}.
\]
As a consequence, this gives the relation
\[
\eta=\frac{\V{P/\left(2V_{0}\right)}+\E{P/(2V_{0})}/(2V_{0})}{\V{P/(2V_{0})}}=\frac{2V_{0}x+1}{2V_{0}x}
\]
\end{proof}

\section{Impact of Gene Replication}\label{subsec:S_mod2}

We present here results relative to the model presented in Figure~3 %\ref{fig:mod2}
 of the main article with volume growth, partition at division and gene replication.  For this intermediate model, we still follow the analysis scheme presented in Figure~7A %\ref{fig:procedure}
 of the main article, but here, we are also able to produce analytic results for the variance of mRNAs and proteins.

\subsection{Theoretical analysis}

\subsubsection{\label{subsec:S_mod2_mRNA-number}Dynamics of mRNA number }

The aim of this section is to prove Theorem~\ref{thm:xs_rep} which states that at any time of the cell, the mRNA number follows a Poisson distribution. To do so, we first give a description of the number mRNAs at any time in the cell cycle using a Marked Poisson Point Process.  With this description, we will be able to show Proposition~\ref{prop:Distrib_M_0}, that the distribution of $M_{0}$ at the beginning of the cell cycle is a Poisson distribution. This proposition will allow to finally prove the main theorem of the subsection.

If time $t=0$ is the beginning of a new cell cycle and if the system is already at steady state in the same sense as the previous models (see Section~\ref{subsec:S_mod1_mrna_dynamic}).  We consider that $M_{0}$, the number of mRNAs at birth is known.  As in Section~\ref{subsec:S_mod1_mrna_dynamic}, we assort mRNAs in independent groups; here we consider three categories.
\begin{itemize}
\item mRNAs which were present at the birth of the cell. Each of them is characterized by its lifetime given by an exponential time of rate $\sigma_{1}$. The $i$-th mRNA is still present at time $t$ if and only if $E_{\sigma_{1}}^{i}>t$, with $\left(E_{\sigma_{1}}^{i}\right)$ being i.i.d. exponential random variables of parameter $\sigma_{1}$.
\item mRNAs created since the birth of the cell by the first copy of the gene. The $i$-th mRNA of this group is characterized by the time of creation $t_{i}$ given by a Poisson Process of rate $\lambda_{1}$ and its lifetime $\delta_{i}$ given by an exponential time of rate $\sigma_{1}$.
\item mRNAs created since the gene replication by the second copy of the gene. As in the previous group, the $i$-th mRNA is characterized by the time of creation $t_{i}$ given by a Poisson Process of rate $\lambda_{1}$ and its lifetime $\delta_{i}$ given by an exponential time of rate $\sigma_{1}$. But here, the Poisson Process of rate $\lambda_{1}$ begins at time $\tau_{R}$, the time of replication of the gene.
\end{itemize}
As in Section~\ref{subsec:S_mod1_mrna_dynamic}, one can represent the number of mRNAs of the second and the third group as two independent MPPPs $\mathcal{N}$ and $\mathcal{N}'$. The first variable $x$ of each of these MPPPs will represent the time. The intensity of each of the MPPP is the same,
\begin{align*}
\nu(\diff x,\diff y) & =\lambda_{1}\diff x\otimes\sigma_{1}e^{-\sigma_{1}y}\diff y.
\end{align*}
The only difference between $\mathcal{N}$ and $\mathcal{N}'$ is
the fact that they begin at time $0$ for $\mathcal{N}$ and at time
$\tau_{R}$ for $\mathcal{N}'$ (see \ref{fig:MPPP-2}). As
a consequence, if we consider an mRNA of either group, the conditions
of its existence at time $t\in[0,\tau_{D}[$ are respectively,
\begin{itemize}
\item if it is in the second group: $(t_{i},\delta_{i})\in\Delta_{t}$ with
$\Delta_{t}=\left\{ (x,y),\,0<x<t,\,y>t-x\right\} ,$
\item if it is in the third group: $(t_{i},\delta_{i})\in\Delta'_{t}$ with
$\Delta'_{t}=\left\{ (x,y),\,\tau_{R}<x<t,\,y>t-x\right\} .$
\end{itemize}

\begin{figure}%appendfig
\caption{\label{fig:MPPP-2}\textbf{The illustration of the Marked Point Poisson Processes
that describe the dynamic of mRNAs}: each mRNA is characterized by
the point $(t_{i},\delta_{i})$, with $(t_{i},\delta_{i})$ following
MPPP $\mathcal{N}$ or $\mathcal{N}'$. Both MPPPs are of intensity
$\nu$. The random variable $t_{i}$ represents its birth time and
$\delta_{i}$ its lifetime, hence this mRNA exists from time $t_{i}$
up to time $t_{i}+\delta_{i}$. The only difference for the two processes
is the starting time: the process $\mathcal{N}$ in (A) begins
at birth (in particular, an mRNA is still present at time $t$, if
and only if the point $(t_{i},\delta_{i})$ is in the set $\Delta_{t}=\left\{ (x,y),\,0<x<t,\,y>t-x\right\} $);
the process $\mathcal{N}'$ , in (B), begins at replication
(in particular an mRNA is still present at time $t$, if and only
if the point $(t_{i},\delta_{i})$ is in the set $\Delta'_{t}=\left\{ (x,y),\,\tau_{R}<x<t,\,y>t-x\right\} $). }
\end{figure}

Hence, we can describe the number of mRNAs at any time $t\in[0,\tau_{D}[$ as follows,
\begin{equation}
M_{t}=\sum_{i=1}^{M_{0}}\ind{E_{\sigma_{1}}^{i}>t}+\mathcal{N}\left(\Delta_{t}\right)+\ind{t\geq\tau_{R}}\mathcal{N}'_{\lambda_{1}}\left(\Delta'_{t}\right).\label{eq:M_s-2}
\end{equation}
Each term corresponds to each group of mRNAs previously described.

At first we want to characterize the distribution of $M_{0}$, the
number of mRNAs at the birth of the cell. To do so, we use the steady state
hypothesis that implies that $M_{0}\overset{\mathcal{D}}{=}M_{\tau_{D}}$.
\begin{prop}
\label{prop:Distrib_M_0}At steady state, the number of mRNAs at birth
$M_{0}$ follows a Poisson distribution of parameter:
\[
x_{0}=\frac{\lambda_{1}}{\sigma_{1}}\left[1-\frac{e^{-(\tau_{D}-\tau_{R})\sigma_{1}}}{2-e^{-\tau_{D}\sigma_{1}}}\right]\mbox{.}
\]
\end{prop}
\begin{proof}
When $s{=}\tau_{D}-$,  by Relation~\eqref{eq:M_s-2},
\[
M_{\tau_{D}-}=\sum_{i=1}^{M_{0}}\ind{E_{\sigma_{1}}^{i}>\tau_{D}-}+\mathcal{N}_{\lambda_{1}}\left(\Delta_{\tau_{D}-}\right)+\mathcal{N}'_{\lambda_{1}}\left(\Delta'_{\tau_{D}-}\right)\mbox{.}
\]
The first term corresponds to initial messengers not degraded after the time $\tau_{D}$. Suppose that $M_{0}$ is distributed according to a Poisson distribution with parameter $x_{0}$, then the random variable
\[
\sum_{i=1}^{M_{0}}\ind{E_{\sigma_{1}}^{i}>\tau_{D}-}
\]
follows also a Poisson distribution with parameter $x_{0}e^{-\tau_{D}\sigma_{1}}$. Operation of thinning of Poisson processes, see \cite{kingman_poisson_1993} for example.

The second term corresponds to mRNAs that were created by the first
copy of the gene and which are still present at division. Since $\mathcal{N}$
is a MPPP, $\mathcal{N}\left(\Delta_{\tau_{D}-}\right)$ is a Poisson
random variable (Proposition~1.13 of \cite{robert_stochastic_2010})
with parameter
\[
\nu\left(\Delta_{\tau_{D}-}\right)=\int_{0}^{\tau_{D}}\int_{\tau_{D}-x}^{\infty}\lambda_{1}\sigma_{1}e^{-\sigma_{1}y}\,\diff y\,\diff x=\frac{\lambda_{1}}{\sigma_{1}}\left(1-e^{-\tau_{D}\sigma_{1}}\right)\mbox{.}
\]

The third term corresponds to mRNAs that were created by the second copy of the gene (replicated at time $\tau_{R}$) and which are still present at division. As before, $\mathcal{N}'\left(\Delta'_{\tau_{D}-}\right)$ is a Poisson random variable with parameter
\[
\nu\left(\Delta'_{\tau_{D}-}\right)=\int_{\tau_{R}}^{\tau_{D}}\int_{\tau_{D}-x}^{\infty}\lambda_{1}\sigma_{1}e^{-\sigma_{1}y}\,\diff y\,\diff x=\frac{\lambda_{1}}{\sigma_{1}}\left(1-e^{-\left(\tau_{D}-\tau_{R}\right)\sigma_{1}}\right).
\]

As $M_{\tau_{D}-}$ is the sum of three independent Poisson random
variables, one gets that
\begin{align*}
M_{\tau_{D}-} & \sim\mathcal{P}\left(x_{0}e^{-\sigma_{1}\tau_{D}}+\frac{\lambda_{1}}{\sigma_{1}}\left(1-e^{-\tau_{D}\sigma_{1}}\right)+\frac{\lambda_{1}}{\sigma_{1}}\left(1-e^{-\left(\tau_{D}-\tau_{R}\right)\sigma_{1}}\right)\right)\\
 & \sim\mathcal{P}\left(x_{0}e^{-\sigma_{1}\tau_{D}}+\frac{\lambda_{1}}{\sigma_{1}}\left(2-e^{-\tau_{D}\sigma_{1}}-e^{-\left(\tau_{D}-\tau_{R}\right)\sigma_{1}}\right)\right).
\end{align*}

Between $\tau_{D}-$ and $\tau_{D}$, with the random sampling, each
mRNA has an equal chance to stay or to disappears, therefore
\[
M_{\tau_{D}}=\sum_{1=0}^{M_{\tau_{D}-}}B_{1/2,i}
\]
with $\left(B_{1/2,i}\right)$ i.i.d. Bernoulli random variables with parameter $1/2$. The random variable $B_{1/2,i}$ determines if the $i$-th mRNA is in the next considered cell or not. The random variable $M_{\tau_{D}}$ hence follows a Poisson distribution such that
\[
M_{\tau_{D}}\sim\left.\mathcal{P}\left(\left[x_{0}e^{-\sigma_{1}\tau_{D}}+\frac{\lambda_{1}}{\sigma_{1}}\left(2-e^{-\tau_{D}\sigma_{1}}-e^{-\left(\tau_{D}-\tau_{R}\right)\sigma_{1}}\right)\right]\right/2\right).
\]
Since the system is at steady state, one has $M_{0}\overset{\mathcal{D}}{=}M_{\tau_{D}}$,
therefore
\[
x_{0}=\frac{1}{2}\left(x_{0}e^{-\sigma_{1}\tau_{D}}+\frac{\lambda_{1}}{\sigma_{1}}\left(2-e^{-\tau_{D}\sigma_{1}}-e^{-\left(\tau_{D}-\tau_{R}\right)\sigma_{1}}\right)\right),
\]
which gives
\[
x_{0}=\frac{\lambda_{1}}{\sigma_{1}}\left[1-\frac{e^{-(\tau_{D}-\tau_{R})\sigma_{1}}}{2-e^{-\tau_{D}\sigma_{1}}}\right]\mbox{.}
\]
Since the steady state distribution is unique, the number of mRNAs
at birth follows a Poisson distribution of parameter $x_{0}$ at steady state.
\end{proof}

We have determined the steady state distribution of the embedded Markov Chain $\left(M_{i\tau_{D}}\right)_{i\in\N}$. Now, we analyze the distribution of mRNA number at any instant $t$ of the cell cycle.
\begin{thm}
\label{thm:xs_rep}At steady state, at  time $t$ in the cell cycle,
the distribution of the mRNA number $M_{t}$ is Poisson with  parameter
\[
x_{t}=\frac{\lambda_{1}}{\sigma_{1}}\left[1-\frac{e^{-(t+\tau_{D}-\tau_{R})\sigma_{1}}}{2-e^{-\tau_{D}\sigma_{1}}}+\ind{t\geq\tau_{R}}\left(1-e^{-(t-\tau_{R})\sigma_{1}}\right)\right]\mbox{.}
\]
In particular, the mean and the variance of mRNA concentration are known at any time $t$ of the cell cycle,
\begin{align*}
\E{M_{t}/V(t)} & =  \frac{\lambda_{1}}{\sigma_{1}V(t)}\left[1-\frac{e^{-(t+\tau_{D}-\tau_{R})\sigma_{1}}}{2-e^{-\tau_{D}\sigma_{1}}}+\ind{t\geq\tau_{R}}\left(1-e^{-(t-\tau_{R})\sigma_{1}}\right)\right],\\
\V{M_{t}/V(t)} & =  \frac{\lambda_{1}}{\sigma_{1}V(t)^{2}}\left[1-\frac{e^{-(t+\tau_{D}-\tau_{R})\sigma_{1}}}{2-e^{-\tau_{D}\sigma_{1}}}+\ind{t\geq\tau_{R}}\left(1-e^{-(t-\tau_{R})\sigma_{1}}\right)\right].
\end{align*}

\end{thm}
\begin{proof}
At a moment $t$ of the cell cycle,  the moment-generating
function of $M_{t}$ at $\xi<0$ is given by
\[
\E{\exp\left(\xi M_{t}\right)}=\E{\exp\left(\xi\left(\sum_{i=1}^{M_{0}}\ind{E_{\sigma_{1}}^{i}>t}+\mathcal{N}\left(\Delta_{t}\right)+\ind{t\geq\tau_{R}}\mathcal{N}'_{\lambda_{1}}\left(\Delta'_{t}\right)\right)\right)}.
\]
Since $M_{0}$, $E_{\sigma_{1}}^{i}$, $\mathcal{N}_{\lambda_{1}}$
and $\mathcal{N}'_{\lambda_{1}}$ are all independent, it follows
that
\[
\E{\exp\left(\xi M_{t}\right)}=\E{\exp\left(\sum_{i=0}^{M_{0}}\xi\ind{E_{\sigma_{1}}^{i}>t}\right)}\cdot\E{\exp\left(\xi\mathcal{N}\left(\Delta_{t}\right)\right)}\cdot\E{\exp\left(\xi\ind{t\geq\tau_{R}}\mathcal{N}'\left(\Delta'_{t}\right)\right)}.
\]
For the first factor, since all the random variables $\ind{E_{\sigma_{1}}^{i}>t}$ are i.i.d. Bernoulli variables with parameter $e^{-t\sigma_{1}}$ and independent of $M_{0}$, one has
\begin{align*}
\E{\exp\left(\sum_{i=0}^{M_{0}}\xi\ind{E_{\sigma_{1}}^{i}>t}\right)} & =\E{\E{\exp\left(\xi\ind{E_{\sigma_{1}}^{1}>t}\right)|M_{0}}^{M_{0}}}\\
 &= \E{\exp\left(1+e^{-t\sigma_{1}}(e^{\xi}-1)\right)^{M_{0}}}.
\end{align*}
With Proposition~\ref{prop:Distrib_M_0}, $M_{0}$ is known to be a Poisson random variable of parameter $x_{0}$, hence, with the probability generating function of a Poisson random variable,
\[
\E{\exp\left(\sum_{i=0}^{M_{0}}\xi\ind{E_{\sigma_{1}}^{i}>t}\right)}=\E{\exp\left(x_{0}e^{-t\sigma_{1}}\left(e^{\xi}-1\right)\right)}
\]
holds. For the second factor, one can recall that $\mathcal{N}\left(\Delta_{t}\right)$ is a Poisson random variable. As in Proposition~\ref{prop:Distrib_M_0}, its parameter can be calculated
\[
\nu\left(\Delta_{t}\right)=\int_{0}^{t}\int_{\tau_{D}-x}^{\infty}\lambda_{1}\sigma_{1}e^{-\sigma_{1}y}\,\diff y\,\diff x=\frac{\lambda_{1}}{\sigma_{1}}\left(1-e^{-t\sigma_{1}}\right).
\]
Identically for the third factor, $\mathcal{N}'\left(\Delta'_{t}\right)$ is a Poisson random variable of parameter
\[
\nu\left(\Delta'_{t}\right)=\int_{\tau_{R}}^{t}\int_{\tau_{D}-x}^{\infty}\lambda_{1}\sigma_{1}e^{-\sigma_{1}y}\,\diff y\,\diff x=\frac{\lambda_{1}}{\sigma_{1}}\left(1-e^{-\left(t-\tau_{R}\right)\sigma_{1}}\right).
\]

As a consequence, the moment generating function of $M_{t}$ is
\begin{align*}
\E{\exp\left(\xi M_{t}\right)} & =  \E{\left(x_{0}e^{-t\sigma_{1}}+\frac{\lambda_{1}}{\sigma_{1}}\left(1-e^{-t\sigma_{1}}\right)+\ind{t>\tau_{R}}\frac{\lambda_{1}}{\sigma_{1}}\left(1-e^{-(t-\tau_{R})\sigma_{1}}\right)\right)\left(e^{\xi}-1\right)}\\
 & =  \E{\left(x_{0}e^{-t\sigma_{1}}+\frac{\lambda_{1}}{\sigma_{1}}\left(1-e^{-t\sigma_{1}}+\ind{t>\tau_{R}}\left(1-e^{-(t-\tau_{R})\sigma_{1}}\right)\right)\right)\left(e^{\xi}-1\right)}
\end{align*}
which is the moment-generating function of a Poisson random variable
of parameter
\[
x_{0}e^{-t\sigma_{1}}+\frac{\lambda_{1}}{\sigma_{1}}\left(1-e^{-t\sigma_{1}}+\ind{t>\tau_{R}}\left(1-e^{-(t-\tau_{R})\sigma_{1}}\right)\right).
\]
\end{proof}

\subsubsection{\label{subsec:S_mod2_Protein-number}Dynamics of Proteins }

\global\long\def\Es#1#2{\left\langle #1\right\rangle _{#2}}

As for the previous analysis of the mRNA number, we search an expression
for protein production through the cell cycle. This case is more
complicated than the mRNA case and we will only calculate analytical
expressions only for the first two moments of $P_{t}$.

Propositions~\ref{prop:EPs_and_VPs} and~\ref{prop:EPs_and_VPs-2}
are the main theoretical results of this section: for any time $t$
of the cell cycle, it gives explicit expressions for the mean $\E{P_{t}}$
and the variance $\V{P_{t}}$ of the protein number. This result is
important as it will be used to directly calculate  the mean $\Ehat{P/V}$
and variance $\Vhat{P/V}$ of the protein concentration averaged across
the cell cycle without using simulations: only with the parameters
of the model ($\lambda_{1}$, $\sigma_{1}$, $\lambda_{2}$, $\tau_{R}$
and $\tau_{D}$), we will be able to know the behavior of the protein
concentration in terms of variance.

In order to prove the Propositions~\ref{prop:EPs_and_VPs} and~\ref{prop:EPs_and_VPs-2},
we will characterize $\E{P_{t}}$ and $\V{P_{t}}$ in the two following
cases:
\begin{enumerate}
\item First, we consider the case before replication ($t<\tau_{R}$ ). We begin by considering that the state of the cell at birth $(M_{0},P_{0})$ is known and we calculate the first two moments of $P_{t}$ for any time $t<\tau_{R}$ (Corollary~\ref{cor:EPs_EPs2}). Then, we integrate over all the possible initial states $(M_{0},P_{0})$ to determine expressions for $\E{P_{t}}$ and $\V{P_{t}}$ for any time $t<\tau_{R}$ (Proposition~\ref{prop:EPs_and_VPs}). These expressions are dependent of the first moments of $(M_{0},P_{0})$: they depend on $\E{M_{0}}$, $\E{P_{0}}$, $\V{M_{0}}$, $\V{P_{0}}$ and $\Cov{M_{0},P_{0}}$.
\item Then we consider the case after replication ($t\geq\tau_{R}$ ). Similarly the first case, we will consider that the state of the cell at replication $(M_{\tau_{R}},P_{\tau_{R}})$ is known and we calculate the first two moments of $P_{t}$ for any time $\tau_{R}\leq t<\tau_{D}$ (Proposition~ \ref{prop:EPs_and_VPs-2}). After integration, expressions for $\E{P_{t}}$ and $\V{P_{t}}$ for any time $t$ after replication are determined, these expressions depend on $\E{M_{\tau_{R}}}$, $\E{P_{\tau_{R}}}$, $\V{M_{\tau_{R}}}$, $\V{P_{\tau_{R}}}$ and $\Cov{M_{\tau_{R}},P_{\tau_{R}}}$ (Proposition~\ref{prop:EPs_and_VPs-2}).
\end{enumerate}
In the end, in Propositions~\ref{prop:EPs_and_VPs} and~\ref{prop:EPs_and_VPs-2}, are presented the mean and variance of protein number at any time $t$ of the cell cycle, only depending on the first moments of $\left(M_{0},P_{0}\right)$ and $\left(M_{\tau_{R}},P_{\tau_{R}}\right)$. Additional results then determine explicitly the first moments of $\left(M_{0},P_{0}\right)$ and $\left(M_{\tau_{R}},P_{\tau_{R}}\right)$ so that the mean and variance of protein number will be fully characterized.

\paragraph*{Description of the  Process of the Number of Proteins  }
Before beginning, we describe the number of proteins $P_{t}$ at any time $t$. We will use this description in the following proofs. Similarly to mRNA case Equation~\eqref{eq:M_s-2}, we group them into two categories.
\begin{itemize}
\item The $P_{0}$ proteins that were there at birth and which remain in the cell during all the cell cycle (as said in the main article the proteolysis is not considered in this model).
\item The proteins that were created during the cell cycle. The rate of production depends on the current number of mRNAs. For that we consider $\left(\mathcal{N}_{\lambda_{2}}^{i}\right)_{i\in\N}$, a sequence of i.i.d. Poisson Point Processes of intensity $\lambda_{2}$; if the $i$-th mRNA exists at time $t$ (that is to say if $i\leq M_{t}$), then the number of proteins produced by this mRNA between $t$ and $t+\diff t$ is $\mathcal{N}_{\lambda_{2}}^{i}(\diff t)$. Hence, the total number of proteins produced between $t$ and $t+\diff t$ is then $\sum_{i=1}^{\infty}\ind{i\leq M_{u}}\mathcal{N}_{\lambda_{2}}^{i}\left(\diff t\right)$.
\end{itemize}
To summarize, the number of proteins at a time $t$ of the cell cycle is
\begin{equation}
P_{t}=P_{0}+\sum_{i=1}^{\infty}\int_{0}^{t}\ind{i\leq M_{u}}\mathcal{N}_{\lambda_{2}}^{i}\left(\diff u\right)\mbox{.}\label{eq:P_s-2}
\end{equation}
The first term is the number of proteins at birth, and the second
takes into account all proteins created between times $0$ and $t$.

\paragraph*{Protein Number Before Replication}
We begin  with the case before replication,  $t<\tau_{R}$. We use the notation $\Es{\cdot}{M_{0},P_{0}}$ as the conditional expectation given $(M_0,P_0)$, i.e. $\Es{\cdot}{M_{0},P_{0}}=\E{\cdot|\left(M_{0},P_{0}\right)}$.  We first characterize the first two moments of $P_{t}$ conditionally on $\left(M_{0},P_{0}\right)$.  As for the mRNAs, we determine at first the moment-generating function of $P_{t}$.
\begin{prop}
\label{prop:MGF_Ps}For any $t\in[0,\tau_{R}[$, the conditional moment generating
function of $P_{t}$ can be expressed as
\[
\Es{\exp\left(\xi P_{t}\right)}{M_{0},P_{0}}=\exp\left(\xi P_{0}\right)\cdot h_{t}\left(\lambda_{2}\left(e^{\xi}-1\right)\right)
\]
for any $\xi<0$ and such as $h_{t}$ is the moment generating function
of $\int_{0}^{t}M_{u}\,\diff u$. The expression of $h_{t}$ is given by
\[
h_{t}(\xi):=\exp\left[M_{0}\log\left[\frac{\sigma_{1}-\xi e^{-(\sigma_{1}-\xi)t}}{\sigma_{1}-\xi}\right]+\lambda_{1}\frac{\xi}{\sigma_{1}-\xi}\left(t-\frac{1-e^{-(\sigma_{1}-\xi)t}}{\sigma_{1}-\xi}\right)\right].
\]
\end{prop}
\begin{proof}
With Equation~\eqref{eq:P_s-2}, it is easy to show that
{\small
\[
\Es{\exp\left(\xi P_{t}\right)}{M_{0},P_{0}}  =  \exp\left(\xi P_{0}\right)\cdot\Es{\prod_{i=1}^{\infty}\Es{\exp\left(\xi\int_{0}^{t}\ind{i\leq M_{u}}\mathcal{N}_{\lambda_{2}}^{i}\left(\diff u\right)\right)|\left(M_{u}\right)_{u\leq t}}{M_{0},P_{0}}}{M_{0},P_{0}}.
\]
}
We then consider the Laplace functional of the Poisson process $\mathcal{N}_{\lambda_{2}}^{i}$,
% \resizebox{10}{!}{
{\small
\begin{align*}
\left.\E{\exp\left(\xi\int_{0}^{t}\ind{i\leq M_{u}}\mathcal{N}_{\lambda_{2}}^{i}\left(\diff u\right)\right)\right|\left(M_{u}\right)_{u\leq t}, P_0}& =  \exp\left[\lambda_{2}\int_{0}^{\infty}\left(\exp\left(\xi\ind{i\leq M_{u}}\ind{u\leq t}\right)-1\right)du\right]\\
 & =  \exp\left[\lambda_{2}\left(e^{\xi}-1\right)\int_{0}^{t}\ind{i\leq M_{u}}du\right].
\end{align*}
% }
}
By making the product for $i$ from $1$ to infinity, one gets
\[
\prod_{i=1}^{\infty}\Es{\exp\left(\xi\int_{0}^{t}\ind{i\leq M_{u}}\mathcal{N}_{\lambda_{2}}^{i}\left(\diff u\right)\right)|\left(M_{u}\right)_{u\leq t}}{M_{0},P_{0}}=  \exp\left[\lambda_{2}\left(e^{\xi}-1\right)\int_{0}^{t}M_{u}du\right].
\]
As a consequence, it indeed follows that
\[
\Es{\exp\left(\xi P_{t}\right)}{M_{0},P_{0}}=\exp\left(\xi P_{0}\right)\cdot h_{t}\left(\lambda_{2}\left(e^{\xi}-1\right)\right).
\]
Using the expression~\eqref{eq:M_s-2} of $M_{t}$, integrated between time $0$ and $t<\tau_{R}$ gives the result. For more details of the calculations, see Chapter~3 of \cite{dessalles_stochastic_2017}).
\end{proof}
As the moment generating function of $P_{t}$ has been characterized, it is possible to deduce, by derivation, the first two moments of $P_{t}$ knowing $(M_{0},P_{0})$ for any time $t$ before the gene replication.
\begin{cor}
\label{cor:EPs_EPs2}At steady state, for  $t\in[0,\tau_{R}[$,
 the first two conditional moments of $P_{t}$ are given by
\begin{align*}
\Es{P_{t}}{M_{0},P_{0}} & =  P_{0}+\lambda_{2}\left(\frac{\lambda_{1}}{\sigma_{1}}t+\left(M_{0}-\frac{\lambda_{1}}{\sigma_{1}}\right)\frac{1-e^{-\sigma_{1}t}}{\sigma_{1}}\right),\\
\Es{P_{t}^{2}}{M_{0},P_{0}} & =  \left(\Es{P_{t}}{M_{0},P_{0}}\right)^{2}+M_{0}\frac{\lambda_{2}}{\sigma_{1}}\left(1-e^{-\sigma_{1}t}+\frac{\lambda_{2}}{\sigma_{1}}\left[1-e^{-\sigma_{1}t}\left(e^{-\sigma_{1}t}+2t\sigma_{1}\right)\right]\right)\\
 &  +\frac{\lambda_{1}\lambda_{2}}{\sigma_{1}^{2}}\left[t\sigma_{1}-1+e^{-\sigma_{1}t}+2\frac{\lambda_{2}}{\sigma_{1}}\left(\sigma_{1}t\left(1+e^{-\sigma_{1}t}\right)-2\left(1-e^{-\sigma_{1}t}\right)\right)\right]
\end{align*}
\end{cor}

\begin{proof}
The first two moments of $P_{t}$ can be obtained by derivation of the moment generating function of  Proposition~\ref{prop:MGF_Ps},
\[
\Es{P_{t}}{M_{0},P_{0}}  =  \lim_{\xi\rightarrow 0}\frac{d}{d\xi}\left[\exp\left(\xi P_{0}\right)h_{t}\left(\lambda_{2}\left(e^{\xi}-1\right)\right)\right]
  =  P_{0}+\lambda_{2}h_{t}'\left(0\right)
\]
and
\begin{align*}
\Es{P_{t}^{2}}{M_{0},P_{0}}  & =  \lim_{\xi\rightarrow0}\frac{d^{2}}{d\xi^{2}}\left[\exp\left(\xi P_{0}\right)h_{t}\left(\lambda_{2}\left(e^{\xi}-1\right)\right)\right]\\
  & =  \left(\Es{P_{t}}{M_{0},P_{0}}\right)^{2}+\lambda_{2}h_{t}'\left(0\right)+\left(\lambda_{2}\right)^{2}\left(h_{t}''\left(0\right)-h_{t}'\left(0\right)^{2}\right)
\end{align*}

The calculations of $h_{t}'\left(0\right)^{2}$, $h_{t}''\left(0\right)^{2}$
allow to show the result (see Chapter~3 of \cite{dessalles_stochastic_2017}
for the details of the calculation).
\end{proof}
The previous corollary gives expressions for $\Es{P_{t}}{M_{0},P_{0}}$ and $\Es{P_{t}^{2}}{M_{0},P_{0}}$. In the next proposition, we integrate these expressions over all birth states $\left(M_{0},P_{0}\right)$ to find formulas for $\E{P_{t}}$ and $\V{P_{t}}$ for any time $t<\tau_{R}$ before replication. These expression depends on joint moments of $M_{0}$ and $P_{0}$.
\begin{prop}
\label{prop:EPs_and_VPs}At any time $t\in[0,\tau_{R}[$ before replication, the mean and the variance of $P_{t}$ are given by
\begin{align*}
\E{P_{t}} & =  \E{P_{0}}+\lambda_{2}\left(\frac{\lambda_{1}}{\sigma_{1}}t+\left(x_{0}-\frac{\lambda_{1}}{\sigma_{1}}\right)\frac{1-e^{-\sigma_{1}t}}{\sigma_{1}}\right),\\
\V{P_{t}} & =  \V{P_{0}}+2\lambda_{2}\frac{1-e^{-\sigma_{1}t}}{\sigma_{1}}\Cov{P_{0},M_{0}}+\left(\lambda_{2}\frac{1-e^{-\sigma_{1}t}}{\sigma_{1}}\right)^{2}x_{0}\\
  & +x_{0}\frac{\lambda_{2}}{\sigma_{1}}\left(1-e^{-\sigma_{1}t}+\frac{\lambda_{2}}{\sigma_{1}}\left[1-e^{-\sigma_{1}t}\left(e^{-\sigma_{1}t}+2t\sigma_{1}\right)\right]\right)\\
  & +\frac{\lambda_{1}\lambda_{2}}{\sigma_{1}^{2}}\left[t\sigma_{1}-1+e^{-\sigma_{1}t}+2\frac{\lambda_{2}}{\sigma_{1}}\left(\sigma_{1}t\left(1+e^{-\sigma_{1}t}\right)-2\left(1-e^{-\sigma_{1}t}\right)\right)\right]
\end{align*}
where $x_{0}$ is defined in Proposition~\ref{prop:Distrib_M_0}.
\end{prop}

\begin{proof}
By considering the mean of the random variable $\E{P_{t}|(M_{0},P_{0})}$ in Corollary~\ref{cor:EPs_EPs2}, the result for
$\E{P_{t}}$ is easy to get. For the variance, consider the expression of $\E{P_{t}^{2}|(M_{0},P_{0})}$
\begin{multline*}
\E{P_{t}^{2}} =  \E{\Es{P_{t}}{M_{0},P_{0}}^{2}}+\E{M_{0}}\frac{\lambda_{2}}{\sigma_{1}}\left(1-e^{-\sigma_{1}t}+\frac{\lambda_{2}}{\sigma_{1}}\left[1-e^{-\sigma_{1}t}\left(e^{-\sigma_{1}t}+2t\sigma_{1}\right)\right]\right)\\
 +\frac{\lambda_{1}\lambda_{2}}{\sigma_{1}^{2}}\left[t\sigma_{1}-1+e^{-\sigma_{1}t}+2\frac{\lambda_{2}}{\sigma_{1}}\left(\sigma_{1}t\left(1+e^{-\sigma_{1}t}\right)-2\left(1-e^{-\sigma_{1}t}\right)\right)\right]
\end{multline*}
and
\begin{multline*}
\V{P_{t}}  =  \E{\Es{P_{t}}{M_{0},P_{0}}^{2}}-\E{P_{t}}^{2}+\E{M_{0}}\frac{\lambda_{2}}{\sigma_{1}}\left(1-e^{-\sigma_{1}t}+\frac{\lambda_{2}}{\sigma_{1}}\left[1-e^{-\sigma_{1}t}\left(e^{-\sigma_{1}t}+2t\sigma_{1}\right)\right]\right)\\
   +\frac{\lambda_{1}\lambda_{2}}{\sigma_{1}^{2}}\left[t\sigma_{1}-1+e^{-\sigma_{1}t}+2\frac{\lambda_{2}}{\sigma_{1}}\left(\sigma_{1}t\left(1+e^{-\sigma_{1}t}\right)-2\left(1-e^{-\sigma_{1}t}\right)\right)\right].
\end{multline*}
Now, for the expression of $\E{\Es{P_{t}}{M_{0},P_{0}}^{2}}-\E{P_{t}}^{2}$,
\begin{align*}
\E{\Es{P_{t}}{M_{0},P_{0}}^{2}}&-\E{P_{t}}^{2}  =  \E{P_{0}^{2}}+\E{\left(\lambda_{2}\left(\frac{\lambda_{1}}{\sigma_{1}}t+\left(M_{0}-\frac{\lambda_{1}}{\sigma_{1}}\right)\frac{1-e^{-\sigma_{1}t}}{\sigma_{1}}\right)\right)^{2}}\\
  & +2\E{P_{0}\times\lambda_{2}\left(\frac{\lambda_{1}}{\sigma_{1}}t+\left(M_{0}-\frac{\lambda_{1}}{\sigma_{1}}\right)\frac{1-e^{-\sigma_{1}t}}{\sigma_{1}}\right)}\\
  &-\E{P_{0}}^{2}+\E{\lambda_{2}\left(\frac{\lambda_{1}}{\sigma_{1}}t+\left(M_{0}-\frac{\lambda_{1}}{\sigma_{1}}\right)\frac{1-e^{-\sigma_{1}t}}{\sigma_{1}}\right)}^{2}\\
  & -2\E{P_{0}}\E{\lambda_{2}\left(\frac{\lambda_{1}}{\sigma_{1}}t+\left(M_{0}-\frac{\lambda_{1}}{\sigma_{1}}\right)\frac{1-e^{-\sigma_{1}t}}{\sigma_{1}}\right)}\\
&=  \V{P_{0}}+\V{\lambda_{2}\left(\frac{\lambda_{1}}{\sigma_{1}}t+\left(M_{0}-\frac{\lambda_{1}}{\sigma_{1}}\right)\frac{1-e^{-\sigma_{1}t}}{\sigma_{1}}\right)}\\
&\hspace{2cm}   +2\Cov{P_{0},\lambda_{2}\left(\frac{\lambda_{1}}{\sigma_{1}}t+\left(M_{0}-\frac{\lambda_{1}}{\sigma_{1}}\right)\frac{1-e^{-\sigma_{1}t}}{\sigma_{1}}\right)}.
\end{align*}
Finally, one just has to remark that due to Proposition~\ref{prop:Distrib_M_0}
$\E{M_{0}}=\V{M_{0}}=x_{0}$.
\end{proof}

\paragraph*{Protein Number After Replication}
For a time $t$ such as $\tau_{R}\leq t<\tau_{D}$ . We adopt a similar approach as for the previous case, the state just after replication $\left(M_{\tau_{R}},P_{\tau_{R}}\right)$ is known, and we want to determine the first two moments of $P_{t}$ for any time $t$ after the replication.
\begin{prop}
\label{prop:EPs_EPs2-2}At steady state, for a time $t\in[\tau_{R},\tau_{D}[$,
conditionally on the state of the cell at replication $(M_{\tau_{R}},P_{\tau_{R}})$,
the first two moments of $P_{t}$ are given by
\[
\Es{P_{t}}{M_{\tau_{R}},P_{\tau_{R}}} =  P_{\tau_{R}}+\lambda_{2}\left(2\frac{\lambda_{1}}{\sigma_{1}}\left(t-\tau_{R}\right)+\left(M_{\tau_{R}}-2\frac{\lambda_{1}}{\sigma_{1}}\right)\frac{1-e^{-\sigma_{1}\left(t-\tau_{R}\right)}}{\sigma_{1}}\right),
\]
\begin{align*}
\Es{P_{t}^{2}}{M_{\tau_{R}},P_{\tau_{R}}}  &=  \left(\Es{P_{t}}{M_{\tau_{R}},P_{\tau_{R}}}\right)^{2}\\
   &+M_{\tau_{R}}\frac{\lambda_{2}}{\sigma_{1}}\left(1-e^{-\sigma_{1}\left(t-\tau_{R}\right)}+\frac{\lambda_{2}}{\sigma_{1}}\left[1-e^{-\sigma_{1}\left(t-\tau_{R}\right)}\left(e^{-\sigma_{1}\left(t-\tau_{R}\right)}+2\left(t-\tau_{R}\right)\sigma_{1}\right)\right]\right)\\
   &+2\frac{\lambda_{1}\lambda_{2}}{\sigma_{1}^{2}}\left[\left(t-\tau_{R}\right)\sigma_{1}-1+e^{-\sigma_{1}\left(t-\tau_{R}\right)}+\right.\\
   &\phantom{+2\frac{\lambda_{1}\lambda_{2}}{\sigma_{1}^{2}}[[}\left.2\frac{\lambda_{2}}{\sigma_{1}}\left(\sigma_{1}\left(t-\tau_{R}\right)\cdot\left(1+e^{-\sigma_{1}\left(t-\tau_{R}\right)}\right)-2\left(1-e^{-\sigma_{1}\left(t-\tau_{R}\right)}\right)\right)\right].
\end{align*}
\end{prop}
\begin{proof}
After the replication, the rate of mRNA production is doubled, but otherwise, the dynamic is identical as it was before the replication.  One can hence easily adapt the proofs of Proposition~\ref{prop:MGF_Ps} and Corollary~\ref{cor:EPs_EPs2}, by replacing the initial state by the state at replication $(M_{\tau_{R}},P_{\tau_{R}})$, by considering that the mRNA production rate is $2\lambda_{1}$, and that the time spent since the initial state is $t-\tau_{R}$.
\end{proof}
We can then integrate the previous expressions on all possible states
at replication $(M_{\tau_{R}},P_{\tau_{R}})$. It follows that
\begin{prop}
\label{prop:EPs_and_VPs-2}At any time $t\in[\tau_{R},\tau_{D}[$
after replication, depending on joint moments of $P_{\tau_{R}}$ and
$M_{\tau_{R}}$, the mean and the variance of $P_{t}$ are given by
\begin{align*}
\E{P_{t}} & =  \E{P_{\tau_{R}}}+\lambda_{2}\left(2\frac{\lambda_{1}}{\sigma_{1}}\left(t-\tau_{R}\right)+\left(x_{\tau_{R}}-2\frac{\lambda_{1}}{\sigma_{1}}\right)\frac{1-e^{-\sigma_{1}\left(t-\tau_{R}\right)}}{\sigma_{1}}\right),\\
\V{P_{t}} & =  \V{P_{\tau_{R}}}+2\lambda_{2}\frac{1-e^{-\sigma_{1}\left(t-\tau_{R}\right)}}{\sigma_{1}}\Cov{P_{\tau_{R}},M_{\tau_{R}}}+\left(\lambda_{2}\frac{1-e^{-\sigma_{1}\left(t-\tau_{R}\right)}}{\sigma_{1}}\right)^{2}x_{\tau_{R}}\\
  & +x_{\tau_{R}}\frac{\lambda_{2}}{\sigma_{1}}\left(1-e^{-\sigma_{1}\left(t-\tau_{R}\right)}+\frac{\lambda_{2}}{\sigma_{1}}\left[1-e^{-\sigma_{1}\left(t-\tau_{R}\right)}\left(e^{-\sigma_{1}\left(t-\tau_{R}\right)}+2\left(t-\tau_{R}\right)\sigma_{1}\right)\right]\right)\\
  & +2\frac{\lambda_{1}\lambda_{2}}{\sigma_{1}^{2}}\left[\left(t-\tau_{R}\right)\sigma_{1}-1+e^{-\sigma_{1}\left(t-\tau_{R}\right)}\right.\\
  &\phantom{11111111111}\left.+2\frac{\lambda_{2}}{\sigma_{1}}\left(\sigma_{1}\left(t-\tau_{R}\right)\left(1+e^{-\sigma_{1}\left(t-\tau_{R}\right)}\right)-2\left(1-e^{-\sigma_{1}\left(t-\tau_{R}\right)}\right)\right)\right],
\end{align*}
with $x_{\tau_{R}}$ as defined in Theorem~\ref{thm:xs_rep}.
\end{prop}
\begin{proof}
It is similar to the proof of Proposition~\ref{prop:EPs_and_VPs}.
\end{proof}

\paragraph*{Protein Number in the Whole Cell Cycle}

In order to have an analytic expression for the mean $\E{P_{t}}$ and variance $\V{P_{t}}$ for any time $t$ of the cell cycle, we need to have expressions for the means $\E{P_{0}}$ and $\left\langle P_{\tau_{R}}\right\rangle $, the variances $\V{P_{0}}$ and $\V{P_{\tau_{R}}}$ as well as the covariances $\Cov{P_{0},M_{0}}$ and $\Cov{P_{\tau_{R}},M_{\tau_{R}}}$.  The general idea is to use the steady state properties that give a relation between the distributions at birth and at division. Indeed, it gives:
\[
P_{\tau_{D}}\overset{\cal D}{=}P_{0}\quad\text{and}\quad\left(M_{\tau_{D}},P_{\tau_{D}}\right)\overset{\cal D}{=}\left(M_{0},P_{0}\right).
\]

Indeed, between times $\tau_{D}-$ and $\tau_{D}$, the proteins undergo a random partitioning, and since the system is at steady state, the distribution of the number of proteins after division $P_{\tau_{D}}$ is the same as the distribution of proteins at birth $P_{0}$. As a consequence:
\[
\sum_{i=1}^{P_{\tau_{D}-}}B_{i,1/2}\overset{\cal D}{=}P_{0}
\]
with $\left(B_{i,1/2}\right)$ being independent Bernoulli random
variables of parameter $1/2$ and being all independent of $P_{\tau_{D}-}$.
\begin{lem}
\label{lem:Binom_div}The mean and the variance of $P_{0}$ depend on the mean and the variance of $P_{\tau_{D}-}$ in the following way
\[
\E{P_{\tau_{D}-}}  =  2\E{P_{0}}\qquad \V{P_{\tau_{D}-}}  =  4\V{P_{0}}-2\E{P_{0}}.
\]
\end{lem}
\begin{proof}
With the moment-generating function of $P_{0}$, one gets
\[
\E{\exp\left[\xi P_{0}\right]}=\E{\prod_{1=1}^{P_{\tau_{D}-}}\E{\exp\left[B_{i,1/2}\right]}}=\E{\left(\frac{1+e^{\xi}}{2}\right)^{P_{\tau_{D}-}}}=\E{\exp\left[\log\left(\frac{1+e^{\xi}}{2}\right)P_{\tau_{D}-}\right]}
\]
As a consequence, by denoting $\eta(\xi):=\E{\exp\left[\xi P_{\tau_{D}-}\right]}$
the moment generating function of $P_{\tau_{D}-}$, it follows:
\begin{align*}
\frac{\diff}{\diff\xi}\E{\exp\left[\xi P_{0}\right]} & =  \frac{e^{\xi}}{1+e^{\xi}}\cdot\eta'\left(\log\left(\frac{1+e^{\xi}}{2}\right)\right)\\
\frac{\diff^{2}}{\diff\xi^{2}}\E{\exp\left[\xi P_{0}\right]} & =  \frac{e^{\xi}}{\left(1+e^{\xi}\right)^{2}}\cdot\eta'\left(\log\left(\frac{1+e^{\xi}}{2}\right)\right)+\left(\frac{e^{\xi}}{1+e^{\xi}}\right)^{2}\cdot\eta''\left(\log\left(\frac{1+e^{\xi}}{2}\right)\right).
\end{align*}
As $\xi$ goes to $0$, one getss
\[
\E{P_{0}}  =  \frac{\E{P_{\tau_{D}-}}}{2} \text{ and }
\E{P_{0}^{2}}  =  \frac{1}{4}\cdot\E{P_{\tau_{D}-}}+\frac{1}{4}\cdot\E{P_{\tau_{D}-}^{2}}.
\]
The lemma is proved.
\end{proof}

We then use this Lemma to calculate the means $\E{P_{0}}$ and $\left\langle P_{\tau_{R}}\right\rangle $
and the variances $\V{P_{0}}$ and $\V{P_{\tau_{R}}}$.
\begin{prop}
\label{prop:EP0_EPtauR}For $\eta=1,2$, denote,
\[
f_{\eta}(t):=\eta\frac{\lambda_{1}}{\sigma_{1}}\left(t-\tau\right)+\left(x_{\tau}-\eta\frac{\lambda_{1}}{\sigma_{1}}\right)\frac{1-e^{-\sigma_{1}\left(t-\tau\right)}}{\sigma_{1}}
\]
with $\tau=0$ in the case of $\eta=1$ (before replication) and
$\tau=\tau_{R}$ for the case $\eta=2$ (after replication). In that
case, we have that:
\[
\E{P_{0}}  =\lambda_{2}\left(f_{1}\left(\tau_{R}\right)+f_{2}\left(\tau_{D}\right)\right) \text{ and }\E{P_{\tau_{R}}}  =\lambda_{2}\left(2f_{1}\left(\tau_{R}\right)+f_{2}\left(\tau_{D}\right)\right).
\]
\end{prop}

\begin{proof}
With Propositions~\ref{prop:EPs_and_VPs} and~\ref{prop:EPs_and_VPs-2},
one gets
\[
\E{P_{\tau_{D}}}  =\E{P_{\tau_{R}}}+\lambda_{2}\left(f_{2}\left(\tau_{D}\right)\right)
  =\E{P_{0}}+\lambda_{2}\left(f_{1}\left(\tau_{R}\right)+f_{2}\left(\tau_{D}\right)\right).
\]
We conclude with the Lemma~\ref{lem:Binom_div}.
\end{proof}

\begin{prop}
\label{prop:EP0_VP0-1}For $\eta=1,2$, define
\begin{multline*}
g_{\eta}(t)  :=\left(\lambda_{2}\frac{1-e^{-\sigma_{1}\left(t-\tau\right)}}{\sigma_{1}}\right)^{2}x_{\tau}\\
  +x_{\tau}\frac{\lambda_{2}}{\sigma_{1}}\left(1-e^{-\sigma_{1}\left(t-\tau\right)}+\frac{\lambda_{2}}{\sigma_{1}}\left[1-e^{-\sigma_{1}\left(t-\tau\right)}\left(e^{-\sigma_{1}\left(t-\tau\right)}+2\left(t-\tau\right)\sigma_{1}\right)\right]\right)\\
  +\eta\frac{\lambda_{1}\lambda_{2}}{\sigma_{1}^{2}}\left[\left(t-\tau\right)\sigma_{1}-1+e^{-\sigma_{1}\left(t-\tau\right)} +2\frac{\lambda_{2}}{\sigma_{1}}\left(\sigma_{1}\left(t-\tau\right)\left(1+e^{-\sigma_{1}\left(t-\tau\right)}\right)-2\left(1-e^{-\sigma_{1}\left(t-\tau\right)}\right)\right)\right].
\end{multline*}
with $\tau=0$ in the case of $\eta=1$ (before replication) and
$\tau=\tau_{R}$ for the case $\eta=2$ (after replication). In that
case, we have that:
\begin{align*}
\V{P_{0}}  =&\frac{1}{3}\left\{ 2\E{P_{0}}+\phantom{\frac{1}{1}}\right.\\
            &\phantom{111} 2\frac{\lambda_{2}}{\sigma_{1}}\left[\left(1-e^{-\sigma_{1}\tau_{R}}\right)\Cov{P_{0},M_{0}}+\left(1-e^{-\sigma_{1}\left(\tau_{D}-\tau_{R}\right)}\right)\Cov{P_{\tau_{R}},M_{\tau_{R}}}\right]\\
            &\phantom{111} \left.+g_{1}\left(\tau_{R}\right)+g_{2}\left(\tau_{D}\right)\vphantom{2\E{P_{0}}+2\frac{\lambda_{2}}{\sigma_{1}}\left[\left(1-e^{-\sigma_{1}\tau_{R}}\right)\Cov{P_{0},M_{0}}+\left(1-e^{-\sigma_{1}\left(\tau_{D}-\tau_{R}\right)}\right)\Cov{P_{\tau_{R}},M_{\tau_{R}}}\right]}\right\} .
\end{align*}
\end{prop}
\begin{proof}
By considering the expressions of Proposition~\ref{prop:EPs_and_VPs-2}
for $t=\tau_{D}-$,
\[
\V{P_{\tau_{D}-}}=\V{P_{\tau_{R}}}+2\lambda_{2}\frac{1-e^{-\sigma_{1}\left(\tau_{D}-\tau_{R}\right)}}{\sigma_{1}}\Cov{P_{\tau_{R}},M_{\tau_{R}}}+g_{2}\left(\tau_{D}\right).
\]
Similarly, the expression of Proposition~\ref{prop:EPs_and_VPs}
for $t=\tau_{R}-$ gives the expression of $\V{P_{\tau_{R}}}$ by
continuity. We have
\begin{align*}
\V{P_{\tau_{D}-}}  = & \V{P_{0}}+2\frac{\lambda_{2}}{\sigma_{1}}\left[\left(1-e^{-\sigma_{1}\tau_{R}}\right)\Cov{P_{0},M_{0}}+\left(1-e^{-\sigma_{1}\left(\tau_{D}-\tau_{R}\right)}\right)\Cov{P_{\tau_{R}},M_{\tau_{R}}}\right]\\
   &+g_{1}\left(\tau_{R}\right)+g_{2}\left(\tau_{D}\right).
\end{align*}
Lemma~\ref{lem:Binom_div} describes the effect of the binomial sampling
between $\tau_{D}-$ and $\tau_{D}$ on the mean and the variance
of $P$. Since, we are at steady state of cell cycles, one has
\begin{align*}
3\V{P_{0}}  = & 2\E{P_{0}}+2\frac{\lambda_{2}}{\sigma_{1}}\left[\left(1-e^{-\sigma_{1}\tau_{R}}\right)\Cov{P_{0},M_{0}}+\left(1-e^{-\sigma_{1}\left(\tau_{D}-\tau_{R}\right)}\right)\Cov{P_{\tau_{R}},M_{\tau_{R}}}\right]\\
   &+g_{1}\left(\tau_{R}\right)+g_{2}\left(\tau_{D}\right).
\end{align*}
\end{proof}
The expression of $\V{P_{\tau_{R}}}$ can then be deduced from Proposition~\ref{prop:EPs_and_VPs}.
\begin{prop}
\label{prop:Cov_P0_M0}For $\eta=1,2$, define
\begin{align*}
k_{\eta}\left(t\right) & :=  \frac{\eta\lambda_{1}\lambda_{2}}{\sigma_{1}^{2}}\E{M_{\tau}}e^{-\left(t-\tau\right)\sigma_{1}}\left(\left(t-\tau\right)\sigma_{1}-\left(1-e^{-\sigma_{1}\left(t-\tau\right)}\right)\right)\\
  & +\frac{\eta\lambda_{1}}{\sigma_{1}}\E{P_{\tau}}\left(1-e^{-\left(t-\tau\right)\sigma_{1}}\right)
   +\frac{\eta\lambda_{1}\lambda_{2}}{\sigma_{1}^{2}}\E{M_{\tau}}\left(1-e^{-\left(t-\tau\right)\sigma_{1}}\right)^{2}\\
  & +\frac{\lambda_{2}}{\sigma_{1}}e^{-\left(t-\tau\right)\sigma_{1}}\left(\left(\E{M_{\tau}^{2}}-\E{M_{\tau}}\right)\left(1-e^{-\sigma_{1}\left(t-\tau\right)}\right)+\sigma_{1}\left(t-\tau\right)\E{M_{\tau}}\right)\\
  & +\frac{\eta\lambda_{1}\lambda_{2}}{\sigma_{1}^{2}}\left[\frac{\eta\lambda_{1}}{\sigma_{1}}\left(1-e^{-\left(t-\tau\right)\sigma_{1}}\right)\left(\left(t-\tau\right)\sigma_{1}-\left(1-e^{-\sigma_{1}\left(t-\tau\right)}\right)\right)+\right.\\
  & \phantom{111111111} \left.\left(1-e^{-\sigma_{1}\left(t-\tau\right)}\left(\left(t-\tau\right)\sigma_{1}+1\right)\right)\right].
\end{align*}
with $\tau=0$ in the case of $\eta=1$ (before replication) and
$\tau=\tau_{R}$ for the case $\eta=2$ (after replication). In that
case,  the covariances can be expressed as
\[
\Cov{M_{0},P_{0}}=\frac{1}{\left(4-e^{-\tau_{D}\sigma_{1}}\right)}\left\{ k_{1}\left(\tau_{R}\right)e^{-\left(\tau_{D}-\tau_{R}\right)\sigma_{1}}+k_{2}\left(\tau_{D}\right)\right\} -\E{M_{0}}\E{P_{0}}
\]
and
\[
\Cov{M_{\tau_{R}},P_{\tau_{R}}}=\left(\Cov{M_{0},P_{0}}+\E{M_{0}}\E{P_{0}}\right)e^{-\tau_{R}\sigma_{1}}+k_{1}\left(\tau_{R}\right)-\E{M_{\tau_{R}}}\E{P_{\tau_{R}}}.
\]
\end{prop}
\begin{proof}
The proof follows the same arguments as in the proofs of the previous propositions. Details of the calculations can be found in Chapter~3 of \cite{dessalles_stochastic_2017}.
\end{proof}

\subsection{\label{subsec:S_mod2_param}Parameter Estimation}

As in the previous intermediate model, we set the doubling time $\tau_{D}$ to $150\,\text{min}$ and the volume at birth $V_{0}=1.3\,\text{\textmu m}^{3}$.  For each gene, we have to determine four different parameters $\lambda_{1}$, $\sigma_{1}$, $\lambda_{2}$ and $\tau_{R}$. We have considered the genes of \cite{taniguchi_quantifying_2010} for which the empirical mean of messengers $\mu_{m}$ and proteins $\mu_{p}$ concentrations, as well as the mRNA half-life time $\tau_{m}$ have been measured.  We still deduce the mRNA degradation rate $\sigma_{1}$ with the mRNA half-life time $\tau_{m}$ (such that $\sigma_{1}=\log2/\tau_{m}$).

%% figS2

\subsubsection{\label{subsec:AppGene-replication-times-1}A Model for the Instants of Gene Replication}

In this model, the time at which each gene is replicated is estimated as follows: we first determine the time of DNA replication initiation (the time $\tau_{I}$ in the cell cycle); as we consider that the DNA-polymerase replicates DNA at constant speed, we can deduce the time of replication of each gene only by knowing its position in the DNA.

The article \cite{wallden_fluctuations_2015} investigates the replication initiation. It is shown that the initiation occurs at a fixed volume per replication origin, and thus independently from the time since the previous division. Furthermore, this  volume seems to be constant for different conditions. For slow growing bacteria (with only one DNA replication per cell cycle), such as those in \cite{taniguchi_quantifying_2010}, the volume at which DNA replication initiation occurs is $V_{I}=1.8\,\text{\textmu m}^{3}$.  As in our model, the volume is considered as growing exponentially, we define the time of replication initiation $\tau_{I}$ as
\[
\tau_{I}=\frac{\tau_{D}}{\log2}\log\frac{V_{I}}{V_{0}}.
\]
The initiation of DNA replication occurs at $\tau_{I}$, the remaining delay to gene replication of each gene is considered as deterministic (we consider the speed of DNA replication as  constant).  The whole chromosome is replicated in around $40\,\text{min}$ \citep{grant_dnaa_2011}, therefore the distance of the gene from the origin of replication is sufficient to determine the time it takes for the DNA-polymerase to replicate it. The position of each gene was determined with Ecogene database \citep{zhou_ecogene_2013}.

\subsubsection{Estimation of $\lambda_{1}$ and $\lambda_{2}$ in an  Homogeneous  Population}

We still have to determine the rates $\lambda_{1}$ and
$\lambda_{2}$. One can interpret the empirical
average mRNA and protein concentration of the experiment (respectively
$\mu_{m}$ and $\mu_{p}$) as the global average of mRNA and protein
concentrations of the model (respectively $\Ehat{M/V}$ and $\Ehat{P/V}$).

Contrary to the previous intermediate model, the mean concentrations
$\E{M_{t}/V(t)}$ and $\E{P_{t}/V(t)}$ change during the cell cycle.
As depicted in Equation~(6)  % \eqref{eq:Ehat_P}
 of the main article, in
order to consider the concentrations $\Ehat{M/V}$ and $\Ehat{P/V}$
averaged over the cell population, one have to explicit the age distribution
$u$ of the population. We consider at first that the distribution
is homogeneous between age $0$ and $\tau_{D}$ (see Section~\ref{subsec:other_pop_distrib}
for a more realistic distribution). Then, the global averages are
known through the integration over the cell cycle of the mean formulas
of Theorem~\ref{thm:xs_rep} and the Propositions~\ref{prop:EPs_and_VPs}
and~\ref{prop:EPs_and_VPs-2} we can write the global average of
mRNA and protein concentrations as
\begin{align*}
\Ehat{M/V} & =  \frac{\lambda_{1}}{\sigma_{1}}\frac{1}{\tau_{D}}\int_{0}^{\tau_{D}}\frac{1}{V_{0}2^{t/\tau_{D}}}\left(1-\frac{e^{-(t+\tau_{D}-\tau_{R})\sigma_{1}}}{2-e^{-\tau_{D}\sigma_{1}}}+\ind{t\geq\tau_{R}}\left(1-e^{-(t-\tau_{R})\sigma_{1}}\right)\right)\,\diff t,\\
\Ehat{P/V} & =  \lambda_{2}\frac{1}{\tau_{D}}\int_{0}^{\tau_{D}}\frac{1}{V_{0}2^{t/\tau_{D}}}\left(f_{1}\left(\tau_{R}\right)+f_{2}\left(\tau_{D}\right)+f_{1}\left(\tau_{R}\wedge t\right)+\ind{t\geq\tau_{R}}f_{2}(t)\right)\,\diff t.
\end{align*}
As a consequence, parameters $\lambda_{1}$ and $\lambda_{2}$ can
be expressed as follows:
\begin{align*}
\lambda_{1} & =  \sigma_{1}\tau_{D}\mu_{m}\left(\int_{0}^{\tau_{D}}\frac{1}{V_{0}2^{t/\tau_{D}}}\left(1-\frac{e^{-(t+\tau_{D}-\tau_{R})\sigma_{1}}}{2-e^{-\tau_{D}\sigma_{1}}}+\ind{t\geq\tau_{R}}\left(1-e^{-(t-\tau_{R})\sigma_{1}}\right)\right)\,\diff t\right)^{-1},\\
\lambda_{2} & =  \tau_{D}\mu_{p}\left(\int_{0}^{\tau_{D}}\frac{1}{V_{0}2^{t/\tau_{D}}}\left(f_{1}\left(\tau_{R}\right)+f_{2}\left(\tau_{D}\right)+f_{1}\left(\tau_{R}\wedge t\right)+\ind{t\geq\tau_{R}}f_{2}(t)\right)\,\diff t\right)^{-1}.
\end{align*}
For each gene, all parameters can be hence determined.

\subsubsection{\label{subsec:other_pop_distrib}Impact of the Distribution of the Population of Cells}

As previously noticed, the definitions of $\Ehat{M/V}$ and $\Ehat{P/V}$ depends on the population age distribution. In real experimental populations of cells (like in \cite{taniguchi_quantifying_2010}) the number of cells in the population is exponentially growing: any dividing cell gives birth to two daughter cells. The distribution of ages is therefore not uniform.

Using a classic age distribution $u$ in the definitions $\Ehat{M/V}$ and $\Ehat{P/V}$ (Equation~(6) of the main article) for exponentially growing populations (see \cite{collins_rate_1962,sharpe_bacillus_1998,robert_division_2014} for instance), we have performed a parameter estimation that takes into account this effect. For any gene, protein variance is estimated in both cases: either with an uniform population or an exponentially growing population. The variances in both cases are almost identical (the histogram ratio of both variances is centered around $1$ with a standard deviation of $8\cdot10^{-3}$).

The distribution considered does  not have a significant impact on the variance of the model. This is due to the fact that the mean concentration $\V{P_{t}/V(t)}$ of any protein remains approximately constant during the cell cycle, there is therefore no significant difference of protein concentration dosage at the beginning or at the end of the cell cycle. We observe the same effect in the case of the complete model of next section (with the sharing of RNA-polymerases and ribosomes).

\section{Impact of the Sharing of RNA-Polymerases and Ribosomes}

\subsection{\label{subsec:S_mod3_Detailed_descr}A Detailed Description of the Model}

The unit of production of one particular protein is presented in \ref{fig:app_mod3_Model-schema}.
We recall that, for any time $t$, the copy number of the $i$-th
gene is $G_{i}(t)$, the number of mRNA is $M_{i}(t)$ and the number
of proteins is $P_{i}(t)$, the number of RNA-polymerases sequestered
on the $i$-th gene is $E_{Y,i}(t)$ and the number of ribosomes sequestered
on an mRNA of type $i$ is $E_{R,i}(t)$. The number of non-sequestered
RNA-polymerases and ribosomes are respectively denoted as $F_{Y}(t)$
and $F_{R}(t)$.

\begin{figure}%appendfig
    \centering
    \resizebox{\textwidth}{!}{
    \begin{tikzpicture}
    \matrix [column sep=29mm, row sep=10mm,ampersand replacement=\&] {
      \&
      \node[gray] (fy) [draw, shape=rectangle] {$F_Y$}; \&
      \&
      \node[gray] (fr) [draw, shape=rectangle] {$F_R$}; \&
      \\

      \node (g) [draw, shape=rectangle,blue] {$G_i$}; \&
      \node (ey) [draw, shape=rectangle,blue] {$E_{Y,i}$}; \&
      \node (m) [draw, shape=rectangle,blue] {$M_i$}; \&
      \node (er) [draw, shape=rectangle,blue] {$E_{R,i}$}; \&
      \node (p) [draw, shape=rectangle,blue] {$P_i$};\\

      \&
      \node (none1) [draw=none,fill=none,blue] {}; \&
      \node (empty) [draw=none,fill=none,blue] {$\emptyset$}; \&
      \node (none2) [draw=none,fill=none] {}; \&
      \node (none3) [draw=none,fill=none,color=red,align=center] {
    	Divisions at \\the volume $2 V_0$}; \&

      \\
       \node (gj) [draw, shape=rectangle,green] {$G_j$}; \&
      \node (eyj) [draw, shape=rectangle,green] {$E_{Y,j}$}; \&
      \node (mj) [draw, shape=rectangle,green] {$M_j$}; \&
      \node (erj) [draw, shape=rectangle,green] {$E_{R,j}$}; \&
      \node (pj) [draw, shape=rectangle,green] {$P_j$}; \\
      \&
      \node (none1j) [draw=none,fill=none,green] {}; \&
      \node (emptyj) [draw=none,fill=none,green] {$\emptyset$}; \&
      \node (none2j) [draw=none,fill=none,green] {}; \&
      \node (none3j) [draw=none,fill=none,color=red,align=center] {
    	Divisions at \\the volume $2 V_0$}; \\
    };
    \draw[->,blue!50,thick] (fy.west)  to [left,bend right=70]  (ey.west);
    \draw[->,blue!50,thick] (ey.east)  to [left,bend right=70]  (fy.east);
    \draw[->,blue!50,thick] (fr.west)  to [left,bend right=70]  (er.west);
    \draw[->,blue!50,thick] (er.east)  to [left,bend right=70]  (fr.east);
    \draw[->,green!50,thick] (fy.west)  to [left,bend right=70]  (eyj.west);
    \draw[->,green!50,thick] (eyj.east)  to [left,bend right=70]  (fy.east);
    \draw[->,green!50,thick] (fr.west)  to [left,bend right=70]  (erj.west);
    \draw[->,green!50,thick] (erj.east)  to [left,bend right=70]  (fr.east);

    \draw[->,blue] (g) -- (ey);
      \draw[->,blue] (g.east) -- +(1,0) |- (none1.center)  -- +(1.7,0) |- (m.west);
      \draw[blue] (none1.center) node[above] {$\lambda_{1,i} G_i  F_Y/V$};
    \draw[blue] (ey.east)+(-0.85,.9) node[right] {$\mu_{1,i} E_{Y,i}$};
      \draw[gray] (fy.east)+(0.3,0)  node[right,above] {\tiny +1};
      \draw[gray] (fy.west)+(-0.3,0) node[left,above]  {\tiny -1};
    \draw[->,blue] (m) -- (empty) node[midway,right] {$\sigma_{1,i} M_i$};
    \draw[->,blue] (m) -- (er) node[above,midway] {$\lambda_{2,i} M_i  F_R /V$};
    \draw[->,blue] (er) -- (p) node[above,midway] {$\mu_{2,i} E_{R,i}$};
      \draw[gray] (fr.east)+(0.3,0)  node[right,above] {\tiny +1};
      \draw[gray] (fr.west)+(-0.3,0) node[left,above]  {\tiny -1};
      \draw[blue] (m.west)+(-0.2,0) node[left,above]  {\tiny +1};
      \draw[blue] (m.south)+(-0.2,0) node[left,below]  {\tiny -1};
      \draw[blue] (p.west)+(-0.2,0) node[left,above]  {\tiny +1};

    \draw[-> ,red] (m) -- (none3);
    \draw[-> ,red] (p) -- (none3);

    \draw[->,green] (gj) -- (eyj);
      \draw[->,green] (gj.east) -- +(1,0) |- (none1j.center)  -- +(1.7,0) |- (mj.west);
      \draw[green] (none1j.center) node[above] {$\lambda_{1,j} G_j  F_Y/V$};
    \draw[green] (eyj.east)+(-0.4,.9) node[right] {$\mu_{1,j} E_{Y,j}$};
      %\draw (fy.east)+(0.3,0)  node[right,above] {\tiny +1};
      %\draw (fy.west)+(-0.3,0) node[left,above]  {\tiny -1};
    \draw[->,green] (mj) -- (emptyj) node[midway,right] {$\sigma_{1,j} M_j$};
    \draw[->,green] (mj) -- (erj) node[above,midway] {$\lambda_{2,j} M_j  F_R /V$};
    \draw[->,green] (erj) -- (pj) node[above,midway] {$\mu_{2,j} E_{R,j}$};
      %\draw (fr.east)+(0.3,0)  node[right,above] {\tiny +1};
      %\draw (fr.west)+(-0.3,0) node[left,above]  {\tiny -1};
      \draw[green] (mj.west)+(-0.2,0) node[left,above]  {\tiny +1};
      \draw[green] (mj.south)+(-0.2,0) node[left,below]  {\tiny -1};
      \draw[green] (pj.west)+(-0.2,0) node[left,above]  {\tiny +1};

    \draw[-> ,red] (mj) -- (none3j);
    \draw[-> ,red] (pj) -- (none3j);

\end{tikzpicture}}
\caption{\label{fig:app_mod3_Model-schema}\textbf{Production unit of the
$i$-th and $j$-th protein with the common pools of free RNA-polymerases
and ribosomes.}}
\end{figure}

\begin{description}
\item [{Transcription}] In the current model, the process of mRNA production is considered as taking part in two steps: first, the binding of the RNA-polymerase and initiation; and second, the elongation and termination of the mRNA. For the first step, \emph{inside a unit volume}, the rate at which an RNA-polymerase binds on the promoter of the $i$-th gene is given by the law of mass action
\[
\lambda_{1,i}\frac{G_{i}(t)}{V(t)}\frac{F_{Y}(t)}{V(t)}.
\]
with $\lambda_{1,i}$ accounts for the specificity of the promoter (its affinity for the RNA-polymerase, the chromosome conformation, etc.). As we are interested in the rate of reactions inside the whole cell of volume $V(t)$, the rate of reaction is then
\[
\lambda_{1,i}G_{i}(t)\frac{F_{Y}(t)}{V(t)}.
\]
The elongation time is given by an exponential random variable of rate $\mu_{1,i}$. Once the elongation terminates, the RNA-polymerase is released in the cytoplasm (increasing the number of free RNA-polymerases $F_{Y}$ by one unit). A messenger is considered created as soon as its elongation begins: the reason for it is that in bacteria (unlike eukaryotes), since transcriptions and translations happen in the same medium, a translation can begin on an mRNA on which the transcription is not finished. As for the previous models, each messenger of type $i$ has a lifetime given by an exponential random variable of rate $\sigma_{1,i}$.
\item [{Translation}] Similarly to the transcription, the rate at which a ribosome encounters an mRNA of type $i$ and initiate translation is $\lambda_{2,i}M_{i}(t)F_{R}(t)/V(t)$ where $\lambda_{2,i}$ will account for mRNA specific aspects (RBS affinity for ribosomes, etc.). The total number of ribosomes sequestered on messengers of type $i$ is $E_{Y,i}(t)$ and each elongation time follows an exponential distribution of rate $\mu_{2,i}$. Here we consider that the protein is created after the termination (since the protein is usually fully functional once its translation is completed); the number of proteins $P_{i}(t)$ is then increased by one unit.  As previously we do not consider protein proteolysis since it usually occurs at much longer timescale than cell cycle.
\item [{DNA~Replication~and~Division~of~the~Cell}] At a time $t$, each gene $i\in\{1,...,K\}$ is characterized by the gene copy number $G_{i}(t)$. As previously only one DNA replication per cell cycle is considered: as a consequence then, for each $i\in\{1,...,K\}$, $G_{i}(t)$ is constant and equal to $1$ (before replication) or to $2$ (after replication). There is two modeling choice for when the DNA replication is initiated: it can occur at a fixed time after the last division or when the cell reaches a certain volume $V_{I}$.  The first simulations are made by considering the volume-dependent initiation event, but as we will see in Section~\ref{subsec:Other-deterministic-times}, simulations with the other modeling choice show no noticeable difference.  The volume $V_{I}$ is fixed to $1.8\,\text{\textmu m}^{3}$ (see \cite{wallden_fluctuations_2015} and Section~\ref{subsec:AppGene-replication-times-1} about this choice). We consider the speed of DNA replication as constant; as a consequence, once known the replication time $\tau_{I}$, the delay until the replication of $i$-th gene is fixed, and is given by the gene position.

For the division, we considered at first that, like in the previous models, the division occurs when the cell reaches exactly the volume $2V_{0}$ (with $V_{0}=1.3\,\text{\textmu m}^{3}$ as it was the case for the  intermediate models considered earlier). We will consider in Section~\ref{subsec:Uncertainty-in-divi-init} the case where the division timing is not as precise. As before, the effect of septation is a random sampling of messengers and proteins: each of them has an equal chance to be in the next considered cell or not. Moreover, at division, all genes have only one copy.
\item [{Volume~Increase}] As said in the main article, the volume $V(t)$ is no longer deterministic as it was the case in the previous intermediate models and it is considered as proportional to the current total mass of proteins in the cell. We denote by $\beta_{P}$ represents ratio mass-volume and by $w_{i}$ the mass of a type $i$ protein. In that case, we have by definition
\begin{equation}
V(t)=\sum_{i=1}^{K}w_{i}P_{i}(t)/\beta_{P}.\label{eq:volume_mass}
\end{equation}
Thus each protein of type $i$ created increases the total
volume of the cell with respect to the factor $w_{i}/\beta_{P}$.
The mass $w_{i}$ of a protein is determined according to its gene
length.
\item [{Production~of~RNA-polymerases~and~ribosomes}] The total number
of RNA-polymerases and ribosomes (whether allocated or not) are respectively
denoted by $N_{Y}(t)$ and $N_{R}(t)$. In a first step, we consider
that the both these quantities are in constant concentration, that
is to say
\[
N_{Y}(t)=\left\lfloor \beta_{Y}V(t)\right\rfloor \quad\text{and}\quad N_{R}(t)=\left\lfloor \beta_{R}V(t)\right\rfloor ,
\]
with $\beta_{Y}$ and $\beta_{R}$ constant parameters and where $\left\lfloor \centerdot\right\rfloor $ is the notation for the floor function. As the cell grows, new RNA-polymerases and ribosomes are added to the system in the corresponding proportion. When division occurs, ribosomes and RNA-polymerases will be set accordingly to the new volume. In Section~\ref{sec:S_mod3_Other-influences} we will consider the more complex case where both RNA-polymerases and ribosomes are directly produced through a gene expression process.
\end{description}

\subsection{Theoretical Analysis}

This complete model is more complex than the previous ones. It is
due in part to the feedback loop that proteins have on their own production:
the more proteins, the more the volume increases, thereby increasing
the total amount of ribosomes and hence the translation rates. This
complicates the complete analytical description of mRNA and protein
mean productions. In this section, we propose a description that mimics
the average behavior of our stochastic model: the goal is to be able
to fit parameters to real measures and use them for stochastic simulations.

\subsubsection{\label{subsec:S_mod3_DeterModel-pres}Presentation of the Deterministic
Production Model}

The description chosen to reflect the average behavior of the stochastic
model previously described is a system of ordinary differential equations
(ODEs) that describes the kinetics of each compound concentration
of the system.

We consider $K$ genes, each of them has a corresponding type of mRNA and protein. For a gene of type $i$, the concentration of gene copies is given by $g_{i}(t)$; mRNAs and protein concentrations are denoted by $m_{i}(t)$ and $p_{i}(t)$. Similarly, $f_{Y}(t)$ and $f_{R}(t)$ respectively represent the concentrations of free RNA-polymerases and free ribosomes; while $e_{Y,i}(t)$ and $e_{R,i}(t)$ denote the concentrations of RNA-polymerases and ribosomes currently sequestered to produce type $i$ proteins. All these quantities correspond to concentrations and not numbers of entities (their stochastic counterparts would be the concentrations $G_{i}(t)/V(t)$, $M_{i}(t)/V(t)$, $P_{i}(t)/V(t)$, etc.).

The reactions between different compounds are given by the \emph{law of mass action}, that is to say that the rate of chemical reaction is proportional to the reactants abundance. We will study the evolution of $m_{i}$, the concentration of mRNAs of type $i$ . The creation of a type $i$ mRNA is the result of a reaction between a free RNA-polymerase (whose concentration is $f_{Y}(t)$) and the gene $i$ (whose concentration is $g_{i}(t)$); $\lambda_{1,i}$ is interpreted as the affinity constant of the reaction. The type $i$ mRNA degradation is the result of a reaction that occurs at rate $\sigma_{1,i}$.

As in the usual description of the cell (see \cite{goelzer_cell_2011} for instance), one also must consider the dilution: without any molecule creation, the concentration of the compound still decreases as the volume grows due to dilution. If we consider that cell volume is growing exponentially, doubling of volume in a time $\tau_{D}$, then the rate of dilution is $\log2/\tau_{D}$. The exponential growth corresponds to the volume dynamics of real bacteria \citep{wang_robust_2010}, and we will see in Section~\ref{subsec:DeterModel-validation} that it is a good approximation of the growth of cells in stochastic simulations.

All these aspects considered altogether, the kinetics of the concentration
of mRNAs of type $i$ is given by the ODE:
\begin{equation}
\frac{\diff m_{i}}{\diff t}(t)=\lambda_{1,i}g_{i}(t)\cdot f_{Y}(t)-\sigma_{1,i}m_{i}(t)-\frac{\log2}{\tau_{D}}\cdot m_{i}(t).\label{eq:sys_mi}
\end{equation}
The first term represents the mRNA creation; the second, the mRNA
degradation; and the last, the dilution.

For the other reactions, for $i\in\{1,...,K\}$, one has
\begin{align}
\frac{\diff p_{i}}{\diff t}(t) & =\mu_{2,i}e_{R,i}(t)-\frac{\log2}{\tau_{D}}\cdot p_{i}(t),\label{eq:sys_pi}\\
\frac{\diff e_{Y,i}}{\diff t}(t) & =\lambda_{1,i}g_{i}(t)\cdot f_{Y}(t)-\mu_{1,i}e_{Y,i}(t)-\frac{\log2}{\tau_{D}}\cdot e_{Y,i}(t),\label{eq:sys_eyi}\\
\frac{\diff e_{R,i}}{\diff t}(t) & =\lambda_{2,i}m_{i}(t)\cdot f_{R}(t)-\mu_{2,i}e_{R,i}(t)-\frac{\log2}{\tau_{D}}\cdot e_{R,i}(t).\label{eq:sys_eri}
\end{align}

As for the stochastic model of the previous section, assume that the concentration of RNA-polymerases (allocated or not) is constant and equal to $\beta_{Y}$, i.e.
\begin{equation}
\beta_{Y}=f_{Y}(t)+\sum_{i=1}^{K}e_{Y,i}(t),\label{eq:sys_fy}
\end{equation}
since $\sum_{i}e_{Y,i}$ and $f_{Y}$ represent the concentrations
of respectively the allocated and non-allocated RNA-polymerases. It
is similar to the ribosomes as we have:

\begin{equation}
\beta_{R}=f_{R}(t)+\sum_{i=1}^{K}e_{R,i}(t).\label{eq:sys_fr}
\end{equation}

The classic strategy in literature to study such system (an analogous model is presented in \cite{borkowski_translation_2016}) is to consider the system in steady state growth: the gene concentration $g_{i}$ is considered as constantly equal to its average value during the cell cycle, and then one can calculate the concentrations of $m_{i}$, $p_{i}$, $e_{Y,i}$ and $e_{R,i}$ at steady state by writing the Equations~\eqref{eq:sys_mi} to~\eqref{eq:sys_fr} with the derivative term as null. Using such method to determine parameters are unfortunatly not precise enough: there is a clear shift between the stochastic protein concentration and the one that shoule be obtained.

We have described the cell during one cycle with a non-constant gene concentration. The instant  of replication of gene $i$ replication within the cycle  is denoted by $\tau_{R,i}$. in particular, at time $t$, the $i$-th gene copy number is known: $g_{i}(t)=\left(1+\ind{t\geq\tau_{R,i}}\right)/(V_{0}2^{t/\tau_{D}})$ (the factor $V_{0}2^{t/\tau_{D}}$ represents the volume). By analogy with the steady state condition presented in Section~\ref{subsec:S_mod1_mrna_dynamic}, it is likely that  a large number of cell cycles have already occurred, so that the concentration of any entities is the same at the beginning and at the end of the cell cycle. For each unit of production, the concentrations $m_{i}$, $p_{i}$, $e_{Y,i}$ and $e_{R,i}$ are such as
\begin{equation}
\forall i\in\{1,...,K\}\quad\begin{cases}
p_{i}(0)=p_{i}(\tau_{D}),\quad & m_{i}(0)=m_{i}(\tau_{D}),\\
e_{Y,i}(0)=e_{Y,i}(\tau_{D}),\quad & e_{R,i}(0)=e_{R,i}(\tau_{D}).
\end{cases}\label{eq:bondery-conditions}
\end{equation}
With these considerations, we have a system of ODEs to describe the average behavior of the main stochastic model during the cell cycle. In the next section, under some simplifications, we propose to give expressions for $m_{i}(t)$, $p_{i}(t)$, $e_{Y,i}(t)$, $e_{R,i}(t)$, $f_{Y}(t)$ and $f_{R}(t)$ as a function of all  parameters ($\lambda_{1,i}$, $\sigma{}_{1,i}$, etc.) and $g_{i}(t)$.

\subsubsection{\label{subsec:S_mod3_DeterModel-dyn}Dynamics of the Average Production Model}

In order to estimate the parameters, one needs to have expressions
for $m_{i}$, $e_{Y,i}$, $p_{i}$, $e_{R,i}$, $f_{R}$ and $f_{Y}$
of the previous ODEs for any time $t$ of the cell cycle. But the
interdependence between $e_{Y,i}$ and $f_{Y}$ on one hand and $e_{R,i}$
and $f_{R}$ on the hand raises difficulties when integrating these
equations. Explicit solution for the dynamics $m_{i}$, $e_{Y,i}$,
$p_{i}$, $e_{R,i}$, $f_{R}$ and $f_{Y}$ are therefore not easy
to obtain directly.

In order to have expressions for these quantities, we have chosen to   some simplifications. In the next sections, the stochastic simulations show a good correspondence between their average concentration of free RNA-polymerase and ribosomes and the ones predicted here; it will therefore justify \emph{a~posteriori }the simplifications that we make in this section.

For the RNA-polymerases, we denote by $\widetilde{\mu_{1}}:=\sum_{i}\mu_{1,i}/K$
the average elongation rates of transcription and the function $h$ such
as
\[
h(t):=\sum_{i=1}^{K}e_{Y,i}(t)\frac{\widetilde{\mu_{1}}}{\mu_{1,i}}.
\]
The dynamic of $h$ is given by summing the Equations~\eqref{eq:sys_eyi}
for $i$ from $1$ to $K$, and by using Equation~\eqref{eq:sys_fy}:
\begin{equation}
\frac{d}{\diff t}h(t)=f_{Y}(t)\cdot\widetilde{\mu_{1}}\left(1+\sum_{i=1}^{K}\frac{\lambda_{1,i}}{\mu_{1,i}}g_{i}(t)\right)-\beta_{Y}\widetilde{\mu_{1}}-\frac{\log2}{\tau_{D}}\cdot h(t).\label{eq:d_ey_i}
\end{equation}

The $h$ is simply a weighted sum of the $e_{Y,i}$ allocated RNA-polymerases. We decided to consider that such weighting has little influence, and that $h$ does not greatly differ from the uniform sum $\sum_{i}e_{Y,i}$, that is to say:
\[
h(t)=\sum_{i=1}^{K}e_{Y,i}(t)\frac{\widetilde{\mu_{1}}}{\mu_{1,i}}\quad\simeq\quad\sum_{i=1}^{K}e_{Y,i}(t)=\beta_{Y}-f_{Y}(t).
\]
It would be in particular true if all elongation rates $\mu_{1,i}$
are identical for all genes (i.e. if $\mu_{1,i}\equiv\widetilde{\mu_{1}}$
for all $i$).

With this simplification, from Relation~\eqref{eq:d_ey_i}, one obtains a differential equation for  $f_{Y}$
\begin{equation}
\frac{\diff}{\diff t}f_{Y}(t)=\widetilde{\mu_{1}}\beta_{Y}\left(\frac{\log2}{\widetilde{\mu_{1}}\tau_{D}}+1\right)-\widetilde{\mu_{1}}\left(1+\frac{\log2}{\widetilde{\mu_{1}}\tau_{D}}+\sum_{i=1}^{K}\frac{\lambda_{1,i}}{\mu_{1,i}}g_{i}(t)\right)f_{Y}(t).\label{eq:d_fy}
\end{equation}
One can remark that the concentrations of free RNA-polymerases is on a quick timescale. Indeed, as there are of the order of  $1.4\times10^{3}$ mRNAs in the cell, see~\cite{neidhardt_chemical_1996} that last approximately $4$ minutes, see~\cite{taniguchi_quantifying_2010}, it gives of the order of  $6$ translations per second. As a consequence, one can expect that that $f_{Y}$ quickly reaches its steady state during the cell cycle. This consideration will be justified \emph{a~posteriori} by the agreement with  stochastic simulations.

With these considerations, we set  the derivative term of Equation~\eqref{eq:d_fy} to be null, hence
\[
f_{Y}(t)=\beta_{Y}\frac{{\displaystyle 1+\frac{\log2}{\widetilde{\mu_{1}}\tau_{D}}}}{{\displaystyle \sum_{i=1}^{K}\frac{\lambda_{1,i}}{\mu_{1,i}}g_{i}(t)+1+\frac{\log2}{\widetilde{\mu_{1}}\tau_{D}}}}.
\]

In the next section it will be shown that $\log2/(\widetilde{\mu_{1}}\times\tau_{D})\sim10^{-3}\ll1$, we will therefore neglect the contribution of this term.
 With a similar argument for free ribosomes, we get
\[
f_{Y}(t)=\beta_{Y}\frac{{\displaystyle 1}}{1+{\displaystyle \sum_{i=1}^{K}\frac{\lambda_{1,i}}{\mu_{1,i}}g_{i}(t)}}\quad\text{and}\quad f_{R}(t)=\beta_{R}\frac{{\displaystyle 1}}{1+{\displaystyle \sum_{i=1}^{K}\frac{\lambda_{2,i}}{\mu_{2,i}}m_{i}(t)}}.
\]

With global quantities $f_{Y}$ and $f_{R}$ known, we are able
to give expression for gene-specific variables. For each $i\in\{1,...,K\}$,
one can integrate Equation~\eqref{eq:sys_mi} and find that:
\[
\frac{\diff m_{i}}{\diff t}(t)=\lambda_{1,i}g_{i}(t)\cdot f_{Y}(t)-\sigma_{1,i}m_{i}(t)-\frac{\log2}{\tau_{D}}\cdot m_{i}(t).
\]
With the boundary conditions of Equation~\eqref{eq:bondery-conditions},
it is easy to deduce that:
\begin{equation}
m_{i}(t)=\lambda_{1,i}\frac{e^{-\sigma_{1,i}t}}{2^{t/\tau_{D}}}\left[\int_{0}^{t}2^{u/\tau_{D}}e^{\sigma_{1,i}u}g_{i}(u)f_{Y}(u)\,\diff u+\frac{\int_{0}^{\tau_{D}}2^{u/\tau_{D}}e^{\sigma_{1,i}u}g_{i}(u)f_{Y}(u)\,\diff u}{2e^{\sigma_{1,i}\tau_{D}}-1}\right].\label{eq:mi_final}
\end{equation}
Since the quantities $g_{i}$ , $f_{Y}$ are known, we have an explicit
solution for $m_{i}$.

Similarly for $e_{Y,i}(t)$ and $e_{R,i}(t)$,
\begin{align*}
e_{Y,i}(t) & =\lambda_{1,i}\frac{e^{-\mu_{1,i}t}}{2^{t/\tau_{D}}}\left[\int_{0}^{t}2^{u/\tau_{D}}e^{\mu_{1,i}u}g_{i}(u)f_{Y}(u)\,\diff u+\frac{\int_{0}^{\tau_{D}}2^{u/\tau_{D}}e^{\mu_{1,i}u}g_{i}(u)f_{Y}(u)\,\diff u}{2e^{\mu_{1,i}\tau_{D}}-1}\right],\\
e_{R,i}(t) & =\lambda_{2,i}\frac{e^{-\mu_{2,i}t}}{2^{t/\tau_{D}}}\left[\int_{0}^{t}2^{u/\tau_{D}}e^{\mu_{2,i}u}m_{i}(u)f_{R}(u)\,\diff u+\frac{\int_{0}^{\tau_{D}}2^{u/\tau_{D}}e^{\mu_{2,i}u}m_{i}(u)f_{R}(u)\,\diff u}{2e^{\mu_{2,i}\tau_{D}}-1}\right].
\end{align*}
Consider now the type $i$ protein concentration. By integrating
the Equation~\eqref{eq:sys_mi}, and by considering the boundary
condition of Equation~\eqref{eq:bondery-conditions}, one gets the relation
\begin{equation}
p_{i}(t)=\frac{\mu_{2,i}}{2^{t/\tau_{D}}}\int_{0}^{\tau_{D}}\left(1+\ind{u<t}\right)2^{u/\tau_{D}}e_{R,i}(u)\,\diff u.\label{eq:pi_final}
\end{equation}

As in the previous models, we are interested in average concentrations
over the cell cycle. Since, in the system of ODEs, we define average
concentrations over the cell cycle of free RNA-polymerases and ribosomes
respectively as
\begin{equation}
\overline{f_{Y}}=\frac{1}{\tau_{D}}\int_{0}^{\tau_{D}}\beta_{Y}\frac{1}{{\displaystyle \sum_{i=1}^{K}\frac{\lambda_{1,i}}{\mu_{1,i}}g_{i}(t)}+1}\,\diff t\quad\text{and}\quad\overline{f_{R}}=\frac{1}{\tau_{D}}\int_{0}^{\tau_{D}}\beta_{R}\frac{1}{{\displaystyle \sum_{i=1}^{K}\frac{\lambda_{2,i}}{\mu_{2,i}}m_{i}(t)+}1}\,\diff t.\label{eq:fy-fr}
\end{equation}
We defined similarly the concentrations $\overline{m_{i}}$ and $\overline{p_{i}}$
averaged over the cell cycle. By integrating Equations~\eqref{eq:pi_final}
and~\eqref{eq:pi_final}, it follows:

\begin{align}
\overline{m_{i}}=&\frac{\lambda_{1,i}}{\sigma_{1,i}\tau_{D}+\log2}\int_{0}^{\tau_{D}}g_{i}(u)f_{Y}(u)\,\diff u\\
\overline{p_{i}}=&\frac{\lambda_{2,i}\mu_{2,i}\tau_{D}}{\log2\left(\mu_{2,i}\tau_{D}+\log2\right)}\int_{0}^{\tau_{D}}m_{i}(u)f_{R}(u)\,\diff u.\label{eq:av-mi-pi}
\end{align}

Now we have expressions of the average concentrations of $\overline{m_{i}}$,
$\overline{p_{i}}$, $\overline{f_{R}}$ and $\overline{f_{Y}}$ for
any time $t$ in the cell cycle that will be used in the next subsection
to determine the parameters.

\subsection{\label{subsec:S_mod3_param}Estimation of Parameters}
The stochastic model of this section are used to describe the production of all proteins of the cell. Recall that \cite{taniguchi_quantifying_2010} has only considered $1018$ genes, out of which only $841$ have their mRNA production measured. In a first step, we only take into account the $841$ genes with protein and mRNA production measured and consider that it would represent the whole genome; in Section~\ref{subsec:S_Other-Additional-genes} we will study the case of a simulation with a complete set of genes representing a full genome of about $2000$ genes.

The determination of the model parameters $\sigma_{1,i}$ of mRNA
degradation of type $i$ , of the doubling time $\tau_{D}$ and the
time $\tau_{R,i}$ of gene replication is the same as for the previous
intermediate models (see Sections~\ref{subsec:S_mod1_param} and~\ref{subsec:S_mod2_param}).

We still need to determine all reaction rates for every protein type ($\lambda_{1,i}$, $\mu_{1,i}$, $\lambda_{2,i}$ and $\mu_{2,i}$ for $i\in\{1,...,K\}$) as well as concentration parameters of RNA-polymerases, and ribosomes (respectively $\beta_{Y}$ and $\beta_{R}$ ), the proportion between the volume and the protein mass $\beta_{P}$, the mass of each proteins $w_{i}$ and the copy number $g_{i}$ of any gene.

Reference \cite{taniguchi_quantifying_2010} does not give the quantities of non-allocated RNA-polymerases or ribosomes. To determine the set of parameters, we fix the average concentration of free RNA-polymerases and ribosomes. Note that we can have multiple sets of parameters depending on this choice.  In the simulations, we will examine several simulations with different values for average free RNA-polymerase and ribosome concentrations to see their impact on the dynamic of the model (Section~\ref{sec:Impact-RNAP-ribo}).

The rates $\mu_{1,i}$, $\mu_{2,i}$ of mRNAs and protein elongation rates can be deduced from the gene length of the $i$-th gene. In the description of the model, we have considered that the length of the mRNA is characterized by its length; so a rate the parameter $\mu_{1,i}$ is given by the mRNA elongation speed ($39\,\text{Nucl}/\text{s}$ in \cite{bremer_modulation_1996} for slowly growing cells) divided by the length of the $i$-th gene. Similarly $\mu_{1,i}$ is given by the protein elongation speed ($12\,\text{aa}/\text{s}$ in \cite{bremer_modulation_1996} for slowly growing cells) divided by the number of amino-acids coded by the $i$-th gene divided . The mass of each protein $w_{i}$ is also deduced from the length of the gene as it determines the number of amino-acids of the protein.

What remains to determine are the concentration parameters of RNA-polymerases, and ribosomes ($\beta_{Y}$ and $\beta_{R}$), the proportion between the volume and the mass of proteins $\beta_{P}$, as well as the activities of the gene and the mRNA (respectively $\lambda_{1,i}$ and $\lambda_{2,i}$) in each unit of production $i\in\{1,...,K\}$. To do so, we interpret the mRNA and protein concentration of each type measured in \cite{taniguchi_quantifying_2010} as the average concentration of each mRNA and proteins over the cell cycle of this model (respectively $\overline{m_{i}}$ and $\overline{p_{i}}$).  Moreover, as previously said, the average concentrations of free RNA-polymerases $\overline{f_{Y}}$ and free ribosomes $\overline{f_{R}}$ are fixed.

We want now to compute $\beta_{P}$, $\beta_{Y}$, $\beta_{R}$, $\lambda_{1,i}$ and $\lambda_{2,i}$ based on known values for $\overline{f_{Y}}$, $\overline{f_{R}}$, $\overline{m_{i}}$ and $\overline{p_{i}}$.  We determine the parameter $\beta_{P}$. In the description of the stochastic model, Equation~\eqref{eq:volume_mass} states that at any moment, the volume is considered to be proportional to the total mass of proteins. Interpreting $\overline{p_{i}}$ as the average concentration of the protein of type $i$ leads by integration of Equation~\eqref{eq:volume_mass} to
\[
\beta_{P}=\sum_{i=1}^{K}w_{i}\overline{p_{i}}.
\]

We continue with the parameters relevant to the transcription:
$\lambda_{1,i}$ and $\beta_{Y}$. With Equations~\eqref{eq:fy-fr}
and~\eqref{eq:av-mi-pi}, $\beta_{Y},\lambda_{1,1},...,\lambda_{1,K}$
are solution of the system
\begin{equation}
\begin{cases}
\beta_{Y} & =\overline{f_{Y}}\left({\displaystyle \frac{1}{\tau_{D}}\int_{0}^{\tau_{D}}{\displaystyle \left(\sum_{i=1}^{K}\frac{\lambda_{1,i}}{\mu_{1,i}}g_{i}(t)+1\right)^{-1}}\,\diff t}\right)^{-1}\\
\lambda_{1,i} & =\overline{m_{i}}\cdot{\displaystyle \left(\sigma_{1,i}\tau_{D}+\log2\right)\cdot\left({\displaystyle \int_{0}^{\tau_{D}}g_{i}(u)f_{Y}(u)\,\diff u}\right)^{-1}}\quad\forall i\in\{1,...,K\}.
\end{cases}\label{eq:fixe_point1}
\end{equation}
Since $\overline{f_{Y}}$, $\overline{m_{i}}$ and $g_{i}(t)$ have
already been settled, we can use a fixed point optimization procedure
to determine $\beta_{Y}$ and all $\lambda_{1,i}$. Then, as these
parameters are determined, we now have an explicit expression for
$f_{Y}(t)$ for any time $t$ of the cell cycle.

We have to determine the parameters relevant to translation, namely $\lambda_{2,i}$ and $\beta_{R}$. Here again, we use a fixed point optimization procedure to deliver the result. With Equation~\eqref{eq:fy-fr} and the expression of $\overline{p_{i}}$ in Equation~\eqref{eq:av-mi-pi}, $\beta_{R},\lambda_{2,1},...,\lambda_{2,K}$ are solutions of the system
\begin{equation}
\begin{cases}
\beta_{R} & =\overline{f_{R}}\times\left({\displaystyle \frac{1}{\tau_{D}}\int_{0}^{\tau_{D}}\left(\sum_{i=1}^{K}\frac{\lambda_{2,i}}{\mu_{2,i}}m_{i}(t)+1\right)^{-1}\,\diff t}\right)^{-1}\\
\lambda_{2,i} & =\overline{p_{i}}\times\left({\displaystyle {\displaystyle \frac{\mu_{2,i}\tau_{D}}{\log2\left(\mu_{2,i}\tau_{D}+\log2\right)}\int_{0}^{\tau_{D}}m_{i}(u)f_{R}(u)}\,\diff u}\right)^{-1}\quad\forall i\in\{1,...,K\}.
\end{cases}\label{eq:fixe_point2}
\end{equation}
By fixing the average amount of free RNA-polymerases and ribosomes, it is possible, through this procedure to determine  parameters with  the experimental measures.

\subsection{\label{subsec:DeterModel-validation}Validation of the Average Production Model}

%% figS3

The description of the average production through the system of ODE (Section~\ref{subsec:S_mod3_DeterModel-pres}) makes the computation of parameters of the stochastic model possible. We need to check that the deterministic description globally corresponds to the average behavior of the stochastic model; for instance, one has to validate that stochastic simulations with the  parameters previously determined, are consistent with the number  of mRNAs and proteins observed.

Here, we present the results of a particular simulation, whose parameters are presented  in S3 Fig(A).%\ref{fig:S3A}.
Its average behavior will be compared with the expressions derived from the system of ODEs. The simulation presented here takes into account  the $841$ genes with protein and mRNA production described in \cite{taniguchi_quantifying_2010}, and we have fixed the number of free RNA-polymerases and ribosomes in order to compute the parameters.

The system of ODEs assumes volume growth is exponential with rate $\log2/\tau_{D}$. In S3 Fig(A), % \ref{fig:mod3_exp-growth},
the volume of the cell indeed seems to grow exponentially in the simulations; the growth rate corresponds to the expected a doubling time of $\tau_{D}$.

For each type of gene, S3 Fig(B), %\ref{fig:S3B}
 shows the ratio between protein production observed in the simulations divided by protein production expected (and similarly for the mRNAs in inset). It appears that the correspondence is correct, especially for the highly expressed proteins. It is less precise for the protein less expressed but, globally, the correspondence  seems good enough.

Computed from the stochastic simulations, the main S3~Fig(C) % \ref{fig:S3C}
 and S3~Fig(D) %, \ref{fig:S3D}
present the mean number of free RNA-polymerases and ribosomes as a function of cell volume. The mean of each free entity is not constant during the cell cycle. The dashed lines represent the expected value of free entities given by the model of ODEs (Equation~\eqref{eq:fy-fr}). It is indeed a good approximation for the behavior of free RNA-polymerases and ribosomes.  The stochastic simulation displays relative quick timescale for the evolution of free RNA-polymerases (of the order of the second) and even quicker for the free ribosomes (insets of S3 Fig(C) and S3 Fig(D)).%,, \ref{fig:S3C} and \ref{fig:S3D}).

All these results support the idea that the expressions derived from the system of ODEs are accurate to describe the average behavior of the stochastic model.
\subsection{\label{sec:Impact-RNAP-ribo}Impact of Free RNA-polymerases and Ribosomes}
As in Section~\ref{subsec:S_mod3_DeterModel-pres}, the parameter computation supposes that  the average concentrations of free RNA-polymerases $\overline{f_{Y}}$ and ribosomes $\overline{f_{R}}$ are fixed.  In
S4 Fig  %\ref{fig:S4},
is presented several simulations where the average concentrations of these free entities are changed.

\subsubsection{\label{subsec:F-RNAP-and_ribo}Few Free Ribosomes and Many Free RNA-polymerases}

The first simulation, corresponding to S4 Fig(A) %,\ref{fig:S4A}
 and S4 Fig(B), % \ref{fig:S4B},
 considers a low concentration of free ribosomes and a high concentration of free RNA-polymerases. This situation seems to reasonable in this biological setting.   Consequently, free ribosomes and free polymerases are subject to a large competition between transcripts (see \cite{warner_economics_2001} in the case of the yeast). At the same time, the parameters are fixed  in such a way that most of  RNA-polymerases are non-allocated, i.e. not bound on the DNA, see~\cite{klumpp_growth-rate-dependent_2008}, see Figure~4 of the main article and
S3 Fig(A). S4 Fig(A) compares the variance of the complete stochastic model, with the one predicted by the previous model (with gene replication and random partitioning); it shows that for $90\%$ of the genes, protein variance ratio is above $0.9$ (the mean of the ratio is $0.96$).

In these simulations, we look at the distributions of free RNA-polymerases and ribosomes. In S4 Fig(B), we show these distributions at three different phases in the cell cycle: we have selected cells of a given volume, either $1.40\,\text{\textmu m}^{3}$,$1.95\,\text{\textmu m}^{3}$ or $2.50\,\text{\textmu m}^{3}$, which correspond to the beginning, middle and end of the cell cycle. These distributions change as the volume increase (so that the average follows the curves shown in S3 Fig(C) and S3 Fig(D)).

In order to interpret the observed distributions of free RNA-polymerases and ribosomes at a certain extend, we can propose a simplified model of RNA-polymerase and ribosome allocation for which  translation and the translation are then considered separately, and that there is no notion of cell growth. The idea would be to approach the ``local'' steady state of RNA-polymerases and ribosomes before any significant change in the volume (for more details about the simplified model, see Section~\ref{subsec:S_mod3_simplified_models}).

This simplified description predicts that for a given volume $V$, the distribution of free RNA-polymerases and ribosomes would be both a binomial distribution. These predicted binomial distributions are plotted in
S4 Fig(B) in thick lines. In the RNA-polymerase case, the binomial distribution globally fit the histograms.  The ribosome distribution is  singular: the parameters of the binomial distribution $\left(N,\phi\right)$ are such that $\phi\ll N$. It is due to the low concentration of free ribosomes chosen for the parameters computation. But even this denatured case shows a good correspondence between the binomial distribution and the simulation histograms.

%% figS4

\subsubsection{\label{subsec:S_mod3_Impact-free-rnap}Influence of Free RNA-polymerase Concentration}

By keeping the low concentration of free ribosomes, we have produced a series of parameters where the average concentration of free RNA-polymerases was fixed successively to $1$, $10$, $100$ and $1000$ $\text{copies}/\text{\textmu m}^{3}$.  In each case, we have deduced a set of parameters, where the affinity constants $\lambda_{1,i}$, $\lambda_{2,i}$ are still calculated so that average mRNA and protein concentrations still correspond to the experimental measures. By performing simulations, we observe that protein variability remains in the same order of magnitude.  Even more, as shown in S4 Fig(C), the gap between the multi-protein model and the previous gene centered model (with volume growth, random partition and gene replication), seems to be reduced for the lowest free RNA-polymerase concentration: the variance ratio between the two models is $0.96$ ($90\%$ of the genes have a variance ratio above $0.92$).

Even for extremely low RNA-polymerase concentrations, the distribution
of free RNA-polymerases and ribosomes are still well predicted by
the simplified model presented in Section~\ref{subsec:S_mod3_simplified_models}
(see S4 Fig(D)).

\subsubsection{\label{subsec:S_mod3_Impact-free-ribo}Influence of Free Ribosome Concentration}

With a high concentration of free RNA-polymerases, we observe the influence of the quantity of free ribosomes on protein variance. We have computed a set of parameters based on average concentrations of non-allocated ribosomes of $1$, $10$, $100$ and $1000$ $\text{copies}/\text{\textmu m}^{3}$.  It can be first remarked that for very high free concentrations, the binomial fit of the simplified model (described in Section~\ref{subsec:S_mod3_simplified_models}) is not relevant to describe the free ribosome distribution.

In this case again, changes to the average concentration of free ribosomes are negligible. As the average concentration of free ribosomes increases, the variance of each protein decreases. As shown in S4 Fig(E), for a concentration of $1000\,\text{copies}/\text{\textmu m}^{3}$, the variance of the multi-protein represents on average $0.98$ of the one predicted by the gene-centered model ($90\%$ of the genes have a variance ratio above $0.93$).

Fluctuations in the number of free ribosomes seem to be the main source of the small additional variability observed in the multi-protein model (compared with the previous intermediate model); and this effect seems less important as the number of free ribosomes is high. But in real bacteria, the number of free ribosomes usually seems quite low due to the high cost of ribosome production; then, a low number of free ribosomes (like in the simulation of Section~\ref{subsec:F-RNAP-and_ribo}) seems more plausible than this simulation.

To confirm the specific influence of fluctuations of ribosomes on protein variability, we have performed  simulations with a modified version of the model. The multi-protein model has been changed in such a way that the concentration of non-allocated ribosomes is fixed  during the whole simulation (meanwhile the free RNA-polymerases are still fluctuating). Results about protein variability are similar to what is shown in S4~Fig(E): the variance of each protein concentration is equivalent to what was described by the gene-centered model.

\subsection{\label{sec:S_mod3_Other-influences}Other Possible Influences on Protein Variability }

\begin{figure}%appendfig
    \centering
    \begin{tikzpicture}[scale=4]
    \node[draw,shape=rectangle] (fy) at (0,0){$F_Y$};
    \node[draw,shape=rectangle] (sy) at (-1,0){$D_Y$};
    \node[draw,shape=rectangle] (eyi) at (1,0){$E_{Y,i}$};

    \node[draw,shape=rectangle] (ey1) at (1,0.4){$E_{Y,1}$};
    \node[draw,shape=rectangle] (eyk) at (1,-0.4){$E_{Y,K}$};
    \node[]  at (1,0.2){...};
    \node[]  at (1,-0.2){...};

    \draw[transform canvas={yshift=.3ex},->]
    	(fy) -- (sy) node[above,midway]
    	{$\lambda_{+}F_{Y}G/(K\cdot V)$};
    \draw[transform canvas={yshift=-.3ex},->]
    	(sy) -- (fy) node[below,midway] {$\lambda_{-}\cdot D_Y$};
    \draw[transform canvas={yshift=.3ex},->]
    	(eyi) -- (fy) node[above,midway] {$\lambda_{1,i} G_i  F_Y/V$};
    \draw[transform canvas={yshift=-.3ex},->]
    	(fy) -- (eyi) node[below,midway] {$\mu_{1,i} E_{Y,i}$};

    \draw[->,black!60]
    	[bend right=27] ($(ey1.west)+(0,0.015)$) to
    	 ($(fy.east)+(0,0.055)$);
    \draw[->,black!60]
    	($(fy.east)+(0,0.025)$) to
    	[bend right=-27] ($(ey1.west)+(0,-0.015)$);

    \draw[->,black!60]
    	[bend right=-27]	($(eyk.west)+(0,-0.015)$) to
    	($(fy.east)+(0,-0.055)$);
    \draw[->,black!60]
    	($(fy.east)+(0,-0.025)$) to
    	[bend right=27] ($(eyk.west)+(0,0.015)$);

    \draw (sy.north)+(0,0.3) node[above]{Non-specifically bound};
    \draw (fy.north)+(0,0.3) node[above]{Free};
    \draw (ey1.north)+(0,0.05) node[above]{Specifically bound};

    \draw (ey1.north west)+(-0.05,0.05)   node (a) {};
    \draw (eyk.south east)+(0.05,-0.05)   node (b) {};
    \path[draw]  (a) rectangle (b);

    \end{tikzpicture}

\caption{\label{fig:Model-seques-rnap}\textbf{Model variant with non-specific RNA-polymerase binding.}
An RNA-polymerase can be either among the specifically bound on the
DNA (in $E_{Y,i}$ for $i\in\{1,...,K\}$) or free (among $F_{Y}$)
or non-specifically bound on the DNA (among $D_{Y}$).}
\end{figure}

In this section, based on the set of parameters of the simulation
Section~\ref{subsec:F-RNAP-and_ribo} (with few free ribosomes and
many free RNA-polymerases), we make variations on some modeling choices
for some cellular mechanisms: a larger set of genes, RNA-polymerases
and ribosomes as a result of gene expression, the introduction of
RNA-polymerase non-specific binding on the DNA, considering uncertainty
in the division and DNA replication processes, etc. We will show that
 protein variability is quite robust to any of these changes: as
for the results presented in Section~\ref{subsec:F-RNAP-and_ribo},
 protein variance is still increased by at most $10\%$ compared
to the gene-centered model.

\subsubsection{\label{subsec:S_Other-Additional-genes}Additional Genes }
The genome of \emph{E. coli } has approximately $2000$ expressed genes.  But the measures of \cite{taniguchi_quantifying_2010} take into account only a part of it.  A total $1018$ protein types were considered in this reference, and among them, only $841$ types have their mRNA production estimated. In order to better represent the complete genome of the bacteria, we have created a set of parameters with an extended pool of additional randomly created genes so that the total number of genes would be $2000$.

For each new gene, we have sampled its average protein and mRNA concentration, an mRNA lifetime and  gene position. By studying the data of \cite{taniguchi_quantifying_2010}, we have investigated the possible statistical correlations between these quantities; it appears that only the mRNA and protein concentration are correlated (as it is shown in the Figure~7B %\ref{fig:mrna_prot_taniguchi}
 of the main article). We therefore have sampled the mRNAs lifetime and the gene position and length independently from the two other quantities.

As in the dataset, the genes are evenly distributed in the DNA, the gene position is assumed to be  uniformly distributed. The empirical mRNA lifetime distribution fitted a log-normal distribution; we have chosen the mRNA lifetime accordingly.

For the mRNA and the protein expression, we have taken into account their correlation. The first step of the procedure is to sample protein production of the new gene according protein empirical distribution of the dataset. We then obtain a realistic protein production for the new gene. In a second step, the mRNA production is chosen depending on its protein production already determined. We have subdivised the dataset in 10 classes according to protein production (see the different colors of Figure~5A %\ref{fig:additional-genes}
  of the main article); we then consider the classes in which protein production of the new gene falls in. We then sample its mRNA production according to the empirical distribution of the mRNA productions of this specific class only. Thus,  the mRNA production of the new gene is sampled with a specific distribution that depends on its protein production With this procedure, the newly created genes have a protein and mRNA productions that seem to be in adequation with the original dataset (see Figure~5A %\ref{fig:additional-genes}
		of the main article).

    \subsubsection{\label{subsec:S_Other-Production-of-RNAP-ribo}Production of RNA-polymerase and Ribosomes}
Simulations with the completed genome show no significant difference in terms of protein variability. In particular, the variance ratio between protein concentration of the gene-centered model and the multi-protein model is not different as in Section~\ref{subsec:F-RNAP-and_ribo}.

In the complete stochastic model as it was presented in Section~\ref{subsec:S_mod3_Detailed_descr}, all ribosomes and all RNA-polymerases are supposed to have constant concentrations (respectively $\beta_{R}$ and $\beta_{Y}$). In reality, both RNA-polymerases and ribosomes are composed of different subunits, each subunit is either a protein or, in the case of ribosomes, a functional RNA. The variability of the production of these subunits can have an overall impact on the global production.

We have performed a preliminary simulation that takes into account this aspect: the goal is not to have a precise description of mechanisms of RNA-polymerase and ribosome productions, but rather to have an insight in the magnitude of additional variability it can induce.  In this version of the model, the expression of one gene represents RNA-polymerase production and the expression of another gene represents the ribosome production. It refers to a case where the RNA-polymerases and ribosomes would be composed of only one proteic subunit.

We therefore created two genes, whose protein production was fixed to correspond to the measured concentration of RNA-polymerases and ribosomes.  The mRNA production and lifetime, the gene position and length have been chosen by the same procedure as described in the previous subsection.

This simulation brings an additional variability in the growth rate: volume growth is more variable. These fluctuations are directly correlated with the number of ribosomes in the cell (Figure~5B  %\ref{fig:production-RNAP-ribo}
  of the main article, above). But surprisingly, these additional variability has no significant impact in protein variability. The Figure~5B  %ref{fig:production-RNAP-ribo}
 of the main article (below) also shows the distribution of the protein FabH for cells of different volumes. This case does not differ from the case where the total amount of RNA-polymerases and ribosomes are in constant concentration.

We can propose a possible interpretation of these results. The fluctuations in the total number of ribosomes seems to influence primarily the speed of growth as shown in Figure~5B %\ref{fig:production-RNAP-ribo}
  of the main article.  When ribosomes are produced, it accelerates the global production of all types of proteins thus increasing the volume. As a consequence, both the production of each type of protein and the volume are co-regulated.  Fluctuations in the total number of ribosomes affect volume growth and the production of the $i$-th protein in the same way such as in a cell of a given volume, the $i$-th protein distribution is relatively unchanged.

\subsubsection{\label{subsec:Other-Sequestered-polymerases}Non-specifically Bound
Polymerases}

In the complete stochastic model as it was described in Section~\ref{subsec:S_mod3_Detailed_descr}, RNA-polymerases are either on the DNA involved in a transcription process, or is among the $F_{Y}$ free RNA-polymerases that freely evolve in the cytoplasm. But it has been experimentally shown that large portion of RNA-polymerases can bind non-specifically on the DNA, without initiating transcription. For instance, \cite{klumpp_growth-rate-dependent_2008} estimated that around $90\%$ of the RNA-polymerases are non-specifically bound to the DNA.

We have created an alternative version of the stochastic model to introduce a third possible class for RNA-polymerases. RNA-polymerases can also bound non-specifically on the DNA. The binding rate is modeled as follows: at any time $t$, a free RNA-polymerase bind on the DNA at a rate that depends on the number of free RNA-polymerases $F_{Y}(t)$ and on the DNA concentration $G(t)/(K\cdot V(t))$ (where $G(t)=\sum_i G_i(t)$); the global rate is hence $\lambda_{+}F_{Y}(t)G(t)/(K\cdot V(t))$ where $\lambda_{+}$ is a parameter that represents the natural affinity of RNA-polymerases for the DNA. Once an RNA-polymerase is bound, it is released in a time represented by an exponential random variable of rate $\lambda_{-}$ (see \ref{fig:Model-seques-rnap}).

We performed a simulation where the parameters $\lambda_{+}$ and $\lambda_{-}$ are chosen such that around $90\%$ of the RNA-polymerases are sequestered on the DNA at any time, as it was the case in \cite{klumpp_growth-rate-dependent_2008}.  The variability of protein  concentration does not seem to be impacted in this case.

\subsubsection{\label{subsec:Uncertainty-in-divi-init}Uncertainty in the
Replication Initiation and Division Timing}

In the complete stochastic model as it was initially described, replication initiation and division occur when the cell reaches the respective volumes of $V_{I}$ and $2V_{0}$. In practice it does not occur in this way. We propose here a modification of the stochastic model to take into account this aspect.

The way replication and division occur is still a disputed topic, see for example \cite{tyson_sloppy_1986,wang_robust_2010,soifer_single-cell_2014,osella_concerted_2014}.  For the division, one hypothesis (referred as ``sizer model'') is that the division decision depends on the current size of the cell (the size can refer to the mass or the volume, but as explained in Section~\ref{subsec:S_mod3_Detailed_descr}, the density constraint, see \cite{marr_growth_1991}, ensures a proportionality relation between these two quantities). With this hypothesis, at any instant, the instantaneous probability to divide depends only on the current cell size. In a first approximation, the cell size distributions observed experimentally can be explained by this ``sizer model'', see \cite{robert_division_2014,osella_concerted_2014}.  It is therefore this framework that we have considered to represent the cell division decision.

At time $t$ of the simulation, with a cell of volume $V(t)$, we introduce an instantaneous division rate $b_{D}(V(t))$ with $b_{D}$ is a positive function (the probability to divide between times $t$ and $t+\diff t$ is given by $b_{D}(V(t))\,\diff t$ ). The division decision is hence only volume dependent. The function $b_{D}$ is chosen so that the division occurs around the volume $2V_{0}$ with $V_{0}=1.3\,\text{\textmu m}^{3}$ and division precision can be fixed (for more information about the function $b_{D}$, see Chapter~4 of \cite{dessalles_stochastic_2017}).

Similarly for the replication initiation decision, the stochastic model initially described considers a fixed volume $V_{I}$ at which the DNA replication is initiated. We introduce variability in this cell decision, in the same way as we do for the division: at time $t$, we consider a replication initiation rate $b_{I}(V(t))$ such as the function $b_{I}$ is chosen in order to have a replication initiation that occurs around volume $V_{I}$.

We did several simulations with different functions $b_{D}$ and $b_{I}$ in order to have different precisions in the division and replication initiation decisions. Protein variability does seem to be changed significantly by any of these scenarios.

\subsubsection{\label{subsec:Other-deterministic-times}Deterministic Time for Replication }

When the stochastic model has been initially presented (in Section~\ref{subsec:S_mod3_Detailed_descr}), we have proposed two ways to model the time of DNA replication initiation $\tau_{I}$ . It can either be a deterministic time after the last division, or it can happen when the cell reaches the specific volume $V_{I}$. We have checked that this modeling choice has no significant influence on the global dynamic of the system, in particular in the protein noise.

\subsection{\label{subsec:S_mod3_ESD}Environmental State Decomposition}

In order to perform the environmental state decomposition, one has to specify the environment $Z$ with which the conditioning is made in Equation~(8) of the main article. In order to have a decomposition analog to the result of the dual reporter technique, one has to take into account the environmental aspects that would be similar in the expression of two identical genes in the same cell. In our model, two gene expressions in the same cell would undergo the same volume growth (as it was the case for the first two intermediate models), the same number of free RNA-polymerases, and the same number of free ribosomes. But on the contrary, the division would be specific to each expression. For a protein of type $i$, the decomposition would therefore be:
\[\small
\V{P_{i}/V}=\underbrace{\vphantom{\left[\Ehat{P_{i}|}\right]}\Ehat{\Vhat{P_{i}(t)/V(t)|\left(V(t),F_{Y}(t),F_{R}(t)\right)}}}_{\Vhatt{P_{i}/V}{int}}+\underbrace{\Vhat{\Ehat{P_{i}(t)/V(t)|\left(V(t),F_{Y}(t),F_{R}(t)\right)}}}_{\Vhatt{P_{i}/V}{ext}}.
\]
Estimations of $\Vhatt{P_{i}/V}{int}$ and $\Vhatt{P_{i}/V}{ext}$ are directly made on the results of the simulations. The result is that the extrinsic contribution of the variance $\Vhatt{P_{i}/V}{ext}$ represent only a very small portion of $\Vhat{P_{i}/V}$ for any gene $i$.

\subsection{\label{subsec:S_mod3_simplified_models}Simplified Models for Transcription
and Translation}

Reference \cite{fromion_stochastic_2015} proposes a multi-protein model for translation, with a  limited number of ribosomes, where each type of mRNA is supposed to be in constant quantities and where the maximum number of ribosomes on one single mRNA is limited.  The system also evolves in a fixed volume as it is the case for classic models.

In order to have a prediction for the number of respectively free
RNA-polymerases and free ribosomes; we consider two analog models
that are slightly simplified versions of the model of \cite{fromion_stochastic_2015}.
They respectively represent the transcription and translation part
and they are completely independent.

The goal is, for each of these models, to provide the equivalent of
the first results of \cite{fromion_stochastic_2015} and we will show
that the expected distribution of free RNA-polymerases (or free ribosomes)
is binomial in these simplified cases.

\subsubsection*{A Model for Transcription}

As explained in Section~\ref{subsec:F-RNAP-and_ribo}, one can interpret the model of \cite{fromion_stochastic_2015} as taking place in a fixed volume $V$ (it would correspond to a small portion of the cell cycle in the complete stochastic model, a portion where the volume of the cell does not change much). We also consider that the gene copy of each unit of production remains constant; as a consequence, the gene copy number of the $i$-th gene $G_{i}$ is constant and known. As in the complete stochastic model, and contrary to the model of \cite{fromion_stochastic_2015}, we consider that there is no limiting number of elongating RNA-polymerases on one gene.

In a pool of $K$ genes, denote by $N_{Y}$ the constant total number of polymerases. We consider the random variables $E_{Y,i}$ for $i\in\{1,...,K\}$ be the number of RNA-polymerases attached to the $i$-th gene. As a consequence, the random variable
\begin{equation}
F_{Y}:=N_{Y}-\sum_{i=1}^{K}E_{Y,i}\label{eq:rnap-decomp}
\end{equation}
is the number of free RNA-polymerases in the system.

The process $X(t)=\left(E_{Y,i}(t),i\in\{1,...K\}\right)$ takes place
in the state place $t$ the subset $S$ of $\mathbb N^{K}$ such as
\[
S:=\left\{ x\in\N^{K},\sum_{i=1}^{K}x_{i}\leq N_{Y}\right\} .
\]
There are at most $N_{Y}$ RNA-polymerases that can be attached to genes at the same time. We can describe the Markov process transition by the following \emph{Q}-matrix: by setting the vector $e_{i}=(\delta_{i'=i})_{i'\in\{1,...,K\}}$ ($\delta$ is used here as the Kronecker delta), for any $x,y\in S$,
\[
\begin{cases}
q(x,x+e_{i})=\lambda_{1,i}G_{i}\lambda_{1,i}f(x)/V & \text{for any }i\in\{1,...K\},\\
q(x,x-e_{i})=\mu_{1,i}x_{i} & \text{for any }i\in\{1,...K\}\text{, if \ensuremath{x_{i}>0}},\\
q(x,y)=0 & \text{if }\left\Vert x-y\right\Vert >1.
\end{cases}
\]
where
\[
f_{Y}(x):=N_{Y}-\sum_{i=1}^{K}x_{i}
\]
the number of free RNA-polymerases. Equation~\eqref{eq:rnap-decomp}
leads in particular to $f(x-e_{i})=f(x)\boldsymbol{+}1$ for all $i$.

As in \cite{fromion_stochastic_2015}, we are looking for an invariant reversible probability measure $\pi$ of this Markov process, i.e.  for any $i\in\{1,\dots,K\}$,
\[
\pi(x)\mu_{1,i}x_{i}  = \pi(x-e_{i})\cdot\lambda_{1,i}G_{i}\left(f_{Y}(x)+1\right)/V.\label{eq:equilibre-point-fixe}
\]
\begin{prop}
The invariant measure $\pi$ of the number of RNA-polymerases in each gene has the following form
\[
\pi(x)=\frac{1}{Z}\cdot\frac{1}{f_{Y}(x)!}\prod_{i=1}^{K}\frac{\left(G_{i}\lambda_{1,i}/(V\mu_{1,i})\right)^{x_{i}}}{x_{i}!}
\]
for any $x\in S$ and with $Z>0$ the normalization constant.
\end{prop}
\begin{proof}
We only need to check that $\pi$ satisfies Equation~\eqref{eq:equilibre-point-fixe}. For a gene $i\in\{1,...,K\}$, we take $x\in S$ such that $x-e_{i}\in S$ .
\begin{align*}
\pi(x)\mu_{1,i}x_{i} & =  \frac{1}{Z}\cdot\frac{1}{f_{Y}(x)!}\prod_{i'=1}^{K}\frac{\left(G_{i'}\lambda_{1,i'}/(V\mu_{1,i'})\right)^{x_{i'}}}{x_{i'}!}\mu_{1,i}x_{i}\\
% & =  \frac{1}{Z}\cdot\frac{1}{f_{Y}(x)!}\prod_{i'=1}^{K}\frac{\left(G_{i'}\lambda_{1,i'}/(V\mu_{1,i'})\right)^{x_{i'}}}{x_{i'}!}\cdot\frac{x_{i}}{G_{i}\lambda_{1,i}/(V\mu_{1,i})}\cdot\frac{G_{i}\lambda_{1,i}}{V}\\
 & =  \frac{1}{Z}\cdot\frac{1}{(f_{Y}(x)+1)!}\prod_{i'\ne i}^{K}\left(\frac{\left(G_{i'}\lambda_{1,i'}/(V\mu_{1,i'})\right)^{x_{i'}}}{x_{i'}!}\right)\cdot\frac{\left(G_{i}\lambda_{1,i}/(V\mu_{1,i})\right)^{(x_{i}-1)}}{(x_{i}-1)!}\cdot\\
 &\phantom{111111111111111111111111111}\frac{G_{i}\lambda_{1,i}}{V}(f_{Y}(x)+1)\\
 & =  \pi(x-e_{i})\cdot G_{i}\lambda_{1,i}(f_{Y}(x)+1)/V.
\end{align*}
So $\pi$ satisfies Equation~\eqref{eq:equilibre-point-fixe}.
\end{proof}
We can now derive from the previous proposition the steady state distribution of $F_{Y}$, the number of free RNA-polymerases of the process.
\begin{prop}
\label{prop:app-FY}The distribution of the  number of free polymerases $F_{Y}$ is given by
\[
\mathbb P\left[F_{Y}=n\right]=\binom{N_{Y}}{n}\frac{\Lambda_{Y}^{N_{Y}-n}}{(1+\Lambda_{Y})^{N_{Y}}}\mbox{,}
\]
with $\Lambda$ defined such as
\[
\Lambda_{Y}:=\sum_{i=1}^{K}G_{i}\lambda_{1,i}/(V\mu_{1,i}).
\]
\end{prop}
$F_{Y}$ follows a binomial distribution ${\cal B}\left(\phi,N\right)$ with parameters $\phi=(1+\Lambda_{Y})^{-1}$ and $N=N_{Y}$.
\begin{proof}
From  Equation~\eqref{eq:equilibre-point-fixe}, it follows that for  $n\in\{1,...,N_{Y}\}$,
\begin{align*}
\mathbb P\left[F_{Y}=n\right] % & =  \sum_{x\in S}\pi(x)\ind{\sum_{i}x_{i}=N_{Y}-n}\\
 & =  \sum_{x\in S}\frac{1}{Z}\cdot\frac{1}{n!}\prod_{i=1}^{K}\frac{\left(G_{i}\lambda_{1,i}/(V\mu_{1,i})\right)^{x_{i}}}{x_{i}!}\ind{\sum_{i}x_{i}=N_{Y}-n}\\
 & =  \frac{1}{Z}\cdot\frac{1}{n!}\sum_{x\in S}\prod_{i=1}^{K}\frac{\left(G_{i}\lambda_{1,i}/(V\mu_{1,i})\right)^{x_{i}}}{x_{i}!}\ind{\sum_{i}x_{i}=N_{Y}-n}\\
 & =  \frac{1}{Z_{F_{Y}}}\cdot\frac{1}{n!}\mathbb P\left[\sum_{i=1}^{K}\sum_{k=1}^{G_{i}}C_{i,k}=N_{Y}-n\right]
\end{align*}
with $\forall i\in\{1,\dots,K\}$ and $\forall k\in\{1,\dots,G_{i}\}$,
$C_{i,k}\sim\cal P(\lambda_{1,i}/(V\mu_{1,i}))$.

Since the independent random variables $C_{p,k}$ are Poisson,  their sum is also Poisson  with  parameter $\Lambda:=\sum_{i=1}^{K}G_{i}\lambda_{1,i}/(V\mu_{1,i})$.
By summing up the previous relation, one gets that

\begin{align*}
1 & =  \frac{1}{Z_{F_{Y}}}\cdot\sum_{n=0}^{N_{Y}}\frac{1}{n!}\mathbb P\left[C_{1,1}=N_{Y}-n\right] \\
& =  \frac{1}{Z_{F_{Y}}}\cdot\sum_{n=0}^{N_{Y}}\frac{1}{n!}e^{-\Lambda}\frac{\Lambda^{N_{Y}-n}}{(N_{Y}-n)!}
& =  \frac{1}{Z_{F_{Y}}}\cdot\frac{1}{N_{Y}!}e^{-\Lambda}\cdot\sum_{n=0}^{N_{Y}}\binom{N_{Y}}{n}\Lambda^{N_{Y}-n},
\end{align*}
hence $Z_{F_{Y}}=e^{-\Lambda}(1+\Lambda)^{N_{Y}}/N_{Y}!$.
\end{proof}

\subsubsection*{A Model for Translation}
The model for translation considered here is analogous to the transcription case. We still consider that the volume $V$ is fixed and that for each gene, the number $M_{i}$ of mRNAs of type $i$ is known and constant (because of these, the process describes here is independent of transcription).  As in the complete stochastic model, and contrary to the model of \cite{fromion_stochastic_2015}, we consider that there is no limiting number of elongating ribosomes on one mRNA.

Similarly to the transcription, we can define $N_{R}$ (the total number of ribosomes), $E_{R,i}$ (the number of ribosomes elongating an mRNA of type $i$) and $F_{R}$ (the number of free ribosomes) such as
\[
F_{R}:=N_{R}-\sum_{i=1}^{K}E_{R,i}.
\]
The rate at which a ribosome is sequestered on a type $i$ mRNA is
therefore $M_{i}\lambda_{2,i}/V$, and the rate at which an elongation
terminates on a type $i$ mRNA is $\mu_{2,i}E_{R,i}$.

As this model is analog to the transcription case, we can also prove
that
\begin{prop}
\label{prop:app-FR}The number of free ribosomes $F_{R}$ follows
\[
\mathbb P\left[F_{R}=n\right]=\binom{N_{R}}{n}\frac{\Lambda_{R}^{N_{R}-n}}{(1+\Lambda_{R})^{N_{R}}}\mbox{,}
\]
with $\Lambda_{R}$ defined such as
\[
\Lambda_{R}:=\sum_{i=1}^{K}M_{i}\lambda_{2,i}/(V\mu_{2,i}),
\]
 $F_{R}$ follows a binomial distribution ${\cal B}\left(\phi,N\right)$
for which $\phi=(1+\Lambda_{R})^{-1}$ and $N=N_{R}$.
\end{prop}


\begin{thebibliography}{10}

\bibitem{elowitz_stochastic_2002}
Elowitz MB, Levine AJ, Siggia ED, Swain PS.
\newblock Stochastic gene expression in a single cell.
\newblock Science. 2002;297(5584):1183--1186.
\newblock doi:{10.1126/science.1070919}.

\bibitem{ozbudak_regulation_2002}
Ozbudak EM, Thattai M, Kurtser I, Grossman AD, van Oudenaarden A.
\newblock Regulation of noise in the expression of a single gene.
\newblock Nature Genetics. 2002;31(1):69--73.
\newblock doi:{10.1038/ng869}.

\bibitem{losick_stochasticity_2008}
Losick R, Desplan C.
\newblock Stochasticity and cell fate.
\newblock Science. 2008;320(5872):65--68.
\newblock doi:{10.1126/science.1147888}.

\bibitem{balaban_bacterial_2004}
Balaban NQ, Merrin J, Chait R, Kowalik L, Leibler S.
\newblock Bacterial persistence as a phenotypic switch.
\newblock Science (New York, NY). 2004;305(5690):1622--1625.
\newblock doi:{10.1126/science.1099390}.

\bibitem{acar_stochastic_2008}
Acar M, Mettetal JT, van Oudenaarden A.
\newblock Stochastic switching as a survival strategy in fluctuating
  environments.
\newblock Nature Genetics. 2008;40(4):471--475.
\newblock doi:{10.1038/ng.110}.

\bibitem{taniguchi_quantifying_2010}
Taniguchi Y, Choi PJ, Li GW, Chen H, Babu M, Hearn J, et~al.
\newblock Quantifying \emph{E. coli} proteome and transcriptome with
  single-molecule sensitivity in single cells.
\newblock Science. 2010;329(5991):533--538.
\newblock doi:{10.1126/science.1188308}.

\bibitem{rigney_stochastic_1977}
Rigney DR, Schieve WC.
\newblock Stochastic model of linear, continuous protein synthesis in bacterial
  populations.
\newblock Journal of Theoretical Biology. 1977;69(4):761--766.
\newblock doi:{10.1016/0022-5193(77)90381-2}.

\bibitem{berg_model_1978}
Berg OG.
\newblock A model for the statistical fluctuations of protein numbers in a
  microbial population.
\newblock Journal of theoretical biology. 1978;71(4):587--603.

\bibitem{paulsson_models_2005}
Paulsson J.
\newblock Models of stochastic gene expression.
\newblock Physics of Life Reviews. 2005;2(2):157--175.
\newblock doi:{10.1016/j.plrev.2005.03.003}.

\bibitem{huh_non-genetic_2011}
Huh D, Paulsson J.
\newblock Non-genetic heterogeneity from random partitioning at cell division.
\newblock Nature genetics. 2011;43(2):95--100.
\newblock doi:{10.1038/ng.729}.

\bibitem{huh_random_2011}
Huh D, Paulsson J.
\newblock Random partitioning of molecules at cell division.
\newblock Proceedings of the National Academy of Sciences.
  2011;108(36):15004--15009.
\newblock doi:{10.1073/pnas.1013171108}.

\bibitem{swain_intrinsic_2002}
Swain PS, Elowitz MB, Siggia ED.
\newblock Intrinsic and extrinsic contributions to stochasticity in gene
  expression.
\newblock Proceedings of the National Academy of Sciences.
  2002;99(20):12795--12800.
\newblock doi:{10.1073/pnas.162041399}.

\bibitem{fromion_stochastic_2013}
Fromion V, Leoncini E, Robert P.
\newblock Stochastic gene expression in cells: a point process approach.
\newblock SIAM Journal on Applied Mathematics. 2013;73(1):195--211.
\newblock doi:{10.1137/120879592}.

\bibitem{shahrezaei_analytical_2008}
Shahrezaei V, Swain PS.
\newblock Analytical distributions for stochastic gene expression.
\newblock Proceedings of the National Academy of Sciences.
  2008;105(45):17256--17261.
\newblock doi:{10.1073/pnas.0803850105}.

\bibitem{schwabe2012}
Schwabe A, Rybakova KN, Bruggeman FJ.
\newblock Transcription {{Stochasticity}} of {{Complex Gene Regulation
  Models}}.
\newblock Biophysical Journal. 2012;103(6):1152--1161.
\newblock doi:{10.1016/j.bpj.2012.07.011}.

\bibitem{dessalles_stochastic_2017}
Dessalles R.
\newblock Stochastic models for protein production: the impact of
  autoregulation, cell cycle and protein production interactions on gene
  expression.
\newblock Universit\'e Paris-Saclay, Paris; 2017.

\bibitem{koch_protein_1955}
Koch AL, Levy HR.
\newblock Protein turnover in growing cultures of \emph{Escherichia coli}.
\newblock Journal of Biological Chemistry. 1955;217(2):947--958.

\bibitem{soltani_intercellular_2016}
Soltani M, Vargas-Garcia CA, Antunes D, Singh A.
\newblock Intercellular variability in protein levels from stochastic
  expression and noisy cell cycle processes.
\newblock PLOS Comput Biol. 2016;12(8):e1004972.
\newblock doi:{10.1371/journal.pcbi.1004972}.

\bibitem{bertaux2018}
Bertaux F, Marguerat S, Shahrezaei V.
\newblock Division Rate, Cell Size and Proteome Allocation: Impact on Gene
  Expression Noise and Implications for the Dynamics of Genetic Circuits.
\newblock Royal Society Open Science. 2018;5(3):172234.
\newblock doi:{10.1098/rsos.172234}.

\bibitem{barziv2016}
{Bar-Ziv} R, Voichek Y, Barkai N.
\newblock Dealing with {{Gene}}-{{Dosage Imbalance}} during {{S Phase}}.
\newblock Trends in Genetics. 2016;32(11):717--723.
\newblock doi:{10.1016/j.tig.2016.08.006}.

\bibitem{narula2015}
Narula J, Kuchina A, Lee DYD, Fujita M, S\"uel GM, Igoshin OA.
\newblock Chromosomal {{Arrangement}} of {{Phosphorelay Genes Couples
  Sporulation}} and {{DNA Replication}}.
\newblock Cell. 2015;162(2):328--337.
\newblock doi:{10.1016/j.cell.2015.06.012}.

\bibitem{walker_generation_2016}
Walker N, Nghe P, Tans SJ.
\newblock Generation and filtering of gene expression noise by the bacterial
  cell cycle.
\newblock BMC Biology. 2016;14:11.
\newblock doi:{10.1186/s12915-016-0231-z}.

\bibitem{kempe2014}
Kempe H, Schwabe A, Cr\'emazy F, Verschure PJ, Bruggeman FJ, Matera AG.
\newblock The Volumes and Transcript Counts of Single Cells Reveal
  Concentration Homeostasis and Capture Biological Noise.
\newblock Molecular Biology of the Cell. 2014;26(4):797--804.
\newblock doi:{10.1091/mbc.e14-08-1296}.

\bibitem{padovanmerhar2015}
{Padovan-Merhar} O, Nair GP, Biaesch AG, Mayer A, Scarfone S, Foley SW, et~al.
\newblock Single {{Mammalian Cells Compensate}} for {{Differences}} in
  {{Cellular Volume}} and {{DNA Copy Number}} through {{Independent Global
  Transcriptional Mechanisms}}.
\newblock Molecular Cell. 2015;58(2):339--352.
\newblock doi:{10.1016/j.molcel.2015.03.005}.

\bibitem{thomas2019}
Thomas P.
\newblock Intrinsic and Extrinsic Noise of Gene Expression in Lineage Trees.
\newblock Scientific Reports. 2019;9(1):474.
\newblock doi:{10.1038/s41598-018-35927-x}.

\bibitem{fromion_stochastic_2015}
Fromion V, Leoncini E, Robert P.
\newblock A stochastic model of the production of multiple proteins in cells.
\newblock SIAM Journal on Applied Mathematics. 2015;75(6):2562--2580.
\newblock doi:{10.1137/140994782}.

\bibitem{thomas2018}
Thomas P, Terradot G, Danos V, Wei{\ss}e AY.
\newblock Sources, Propagation and Consequences of Stochasticity in Cellular
  Growth.
\newblock Nature Communications. 2018;9(1):4528.
\newblock doi:{10.1038/s41467-018-06912-9}.

\bibitem{lin2018}
Lin J, Amir A.
\newblock Homeostasis of Protein and {{mRNA}} Concentrations in Growing Cells.
\newblock Nature Communications. 2018;9(1):4496.
\newblock doi:{10.1038/s41467-018-06714-z}.

\bibitem{wang_robust_2010}
Wang P, Robert L, Pelletier J, Dang WL, Taddei F, Wright A, et~al.
\newblock Robust growth of \emph{Escherichia coli}.
\newblock Current biology: CB. 2010;20(12):1099--1103.
\newblock doi:{10.1016/j.cub.2010.04.045}.

\bibitem{tyson_sloppy_1986}
Tyson JJ, Diekmann O.
\newblock Sloppy size control of the cell division cycle.
\newblock Journal of Theoretical Biology. 1986;118(4):405--426.
\newblock doi:{10.1016/S0022-5193(86)80162-X}.

\bibitem{soifer_single-cell_2014}
Soifer I, Robert L, Barkai N, Amir A.
\newblock Single-cell analysis of growth in budding yeast and bacteria reveals
  a common size regulation strategy.
\newblock arXiv:14104771 [cond-mat, q-bio]. 2014;.

\bibitem{osella_concerted_2014}
Osella M, Nugent E, Lagomarsino MC.
\newblock Concerted control of \emph{Escherichia coli} cell division.
\newblock Proceedings of the National Academy of Sciences.
  2014;111(9):3431--3435.
\newblock doi:{10.1073/pnas.1313715111}.

\bibitem{robert_division_2014}
Robert L, Hoffmann M, Krell N, Aymerich S, Robert J, Doumic M.
\newblock Division in \emph{Escherichia coli} is triggered by a size-sensing
  rather than a timing mechanism.
\newblock BMC Biology. 2014;12(1):17.
\newblock doi:{10.1186/1741-7007-12-17}.

\bibitem{hilfinger_separating_2011}
Hilfinger A, Paulsson J.
\newblock Separating intrinsic from extrinsic fluctuations in dynamic
  biological systems.
\newblock Proceedings of the National Academy of Sciences of the United States
  of America. 2011;108(29):12167--12172.
\newblock doi:{10.1073/pnas.1018832108}.

\bibitem{blake_noise_2003}
Blake WJ, K{\ae}rn M, Cantor CR, Collins JJ.
\newblock Noise in eukaryotic gene expression.
\newblock Nature. 2003;422(6932):633--637.
\newblock doi:{10.1038/nature01546}.

\bibitem{marr_growth_1991}
Marr AG.
\newblock Growth rate of \emph{Escherichia coli}.
\newblock Microbiological Reviews. 1991;55(2):316--333.

\bibitem{neidhardt_chemical_1996}
Neidhardt FC, Umbarger HE.
\newblock Chemical composition of \emph{Escherichia coli}.
\newblock In: \emph{Escherichia coli} and \emph{Salmonella}: cellular and
  molecular biology. 2nd ed. ASM Press; 1996.

\bibitem{klumpp_growth-rate-dependent_2008}
Klumpp S, Hwa T.
\newblock Growth-rate-dependent partitioning of rna polymerases in bacteria.
\newblock Proceedings of the National Academy of Sciences.
  2008;105(51):20245--20250.
\newblock doi:{10.1073/pnas.0804953105}.

\bibitem{dessalles_stochastic_2017b}
Dessalles R, Fromion V, Robert P.
\newblock A stochastic analysis of autoregulation of gene expression.
\newblock Journal of Mathematical Biology. 2017;75(5):1253--1283.
\newblock doi:{10.1007/s00285-017-1116-7}.

\bibitem{siwiak_transimulation_2013}
Siwiak M, Zielenkiewicz P.
\newblock Transimulation - protein biosynthesis web service.
\newblock PLOS ONE. 2013;8(9):e73943.
\newblock doi:{10.1371/journal.pone.0073943}.

\bibitem{collins_rate_1962}
Collins JF, Richmond MH.
\newblock Rate of growth of \emph{Bacillus cereus} between divisions.
\newblock Journal of General Microbiology. 1962;28(1):15--33.
\newblock doi:{10.1099/00221287-28-1-15}.

\bibitem{sharpe_bacillus_1998}
Sharpe ME, Hauser PM, Sharpe RG, Errington J.
\newblock \emph{Bacillus subtilis} cell cycle as studied by fluorescence
  microscopy: constancy of cell length at initiation of dna replication and
  evidence for active nucleoid partitioning.
\newblock Journal of Bacteriology. 1998;180(3):547--555.

\bibitem{robert_stochastic_2010}
Robert, Philippe
\newblock Stochastic networks and queues
\newblock Springer. 2010;
\newblock ISBN 978-3-642-05625-3.

\bibitem{kingman_poisson_1993}
Kingman, J. F. C.
\newblock Poisson processes
\newblock Oxford University Press. 1993.
\newblock ISBN 978-0-19-159124-2.

\bibitem{gillespie_exact_1977}
Gillespie, Daniel T.
\newblock Exact stochastic simulation of coupled chemical reactions
\newblock The Journal of Physical Chemistry. 1977.
\newblock doi:{10.1021/j100540a008}

\bibitem{siwiak_transimulation_2013}
Siwiak, Marlena and Zielenkiewicz, Piotr
\newblock Transimulation - protein biosynthesis web service
\newblock PLOS ONE. 2013
\newblock doi:{10.1021/j100540a008}


\bibitem{wallden_fluctuations_2015}
Wallden, Mats and Fange, David and Ullman, Gustaf and Marklund, Erik G. and Elf, Johan
\newblock Fluctuations in growth rates determine the generation time and size distributions of \emph{E. coli} cells
\newblock arXiv:1504.03145 [q-bio]. 2015

\bibitem{zhou_ecogene_2013}
Zhou, Jindan and Rudd, Kenneth E.
\newblock Ecogene 3.0
\newblock Nucleic Acids Research. 2013
\newblock  doi:{10.1093/nar/gks1235}

\bibitem{goelzer_cell_2011}
Goelzer, Anne and Fromion, Vincent and Scorletti, G\'erard
\newblock Cell design in bacteria as a convex optimization problem
\newblock Automatica. 2011
\newblock  doi:{10.1016/j.automatica.2011.02.038}


\bibitem{borkowski_translation_2016}
Borkowski, Olivier and Goelzer, Anne and Schaffer, Marc and Calabre, Magali and M\"ader, Ulrike and Aymerich, St\'ephane and Jules, Matthieu and Fromion, Vincent
\newblock Translation elicits a growth rate-dependent, genome-wide, differential protein production in \emph{Bacillus subtilis}
\newblock Molecular Systems Biology. 2016
\newblock  doi:{10.15252/msb.20156608}

\bibitem{bremer_modulation_1996}
Bremer, Hans and Dennis, Patrick P.
\newblock Modulation of chemical composition and other parameters of the cell at different exponential growth rates
\newblock ASM Press. 1996

\bibitem{warner_economics_2001}
Warner, J. R. and Vilardell, J. and Sohn, J. H.
\newblock Economics of ribosome biosynthesis
\newblock {10.1101/sqb.2001.66.567}. 2001
\newblock  doi:{10.1101/sqb.2001.66.567}


\end{thebibliography}
\end{document}